\documentclass[11pt, letterpaper, fleqn]{article}
\pdfoutput=1
\usepackage{amsmath}
\usepackage{amsthm}
\usepackage{latexsym}
\usepackage{xspace}
\usepackage{amssymb}
\usepackage{fancybox}
\usepackage{epsfig, url}
\usepackage{fullpage}
\usepackage{mathtools,wasysym}
\usepackage{cite}
\usepackage{thmtools}
\usepackage{thm-restate,microtype}
\usepackage{enumerate}
\usepackage{framed}
\usepackage{dsfont}
\usepackage{graphicx}
\usepackage[pagebackref]{hyperref}
\hypersetup{
    pdftitle={}, 
    pdfauthor={}, 
    colorlinks=true, 
    linkcolor=blue, 
    citecolor=blue, 
    urlcolor=blue 
}

\usepackage{algorithmicx}
\usepackage{algorithm}
\usepackage{listings}
\usepackage{algcompatible}
\usepackage{algpseudocode}
\usepackage{booktabs}
\usepackage{calc}
\usepackage{tocloft}

\makeatletter
\renewcommand*{\ALG@name}{Polynomial}
\makeatother

\renewcommand{\backref}[1]{}

\renewcommand{\backrefalt}[4]{%
\ifcase #1
\or
[p.\ #2]
\else
[pp.\ #2]
\fi}

\makeatletter
\newcommand{\para}{
  \@startsection{paragraph}{4}
  {\z@}{2ex \@plus 3.3ex \@minus .2ex}{-1em}
  {\normalfont\normalsize\bfseries}
}
\makeatother

\newenvironment{myitemize}
{ \begin{itemize}
    \setlength{\itemsep}{0pt}
    \setlength{\parskip}{0pt}
    \setlength{\parsep}{0pt}     }
{ \end{itemize}                  } 

\newtheorem{theorem}{Theorem}
\newtheorem{openproblem}[theorem]{Open Problem}
\newtheorem{condition}{Condition}
\newtheorem{remark}{Remark}
\newtheorem{definition}[theorem]{Definition}
\newtheorem{lemma}[theorem]{Lemma}
\newtheorem{corollary}[theorem]{Corollary}
\newtheorem{fact}[theorem]{Fact}
\newtheorem{proposition}[theorem]{Proposition}

\theoremstyle{plain}

\usepackage[capitalize,nameinlink]{cleveref} 
\crefname{condition}{Condition}{Condition}
\crefname{openproblem}{Open Problem}{Open Problem}
\crefname{fact}{Fact}{Fact}

\newcommand{\genericdomain}{\mathcal{X}}

\newcommand{\reals}{\mathbb{R}}
\newcommand{\R}{\mathbb{R}}

\newcommand{\E}{\mathbb{E}}

\newcommand{\sgn}{\mathrm{sign}}

\renewcommand{\P}{\mathcal{P}}

\newcommand{\poly}{\mathrm{poly}}

\newcommand{\SURJ}{\mathsf{ SURJ}}
\newcommand{\SDU}{\mathsf{ SDU}}
\newcommand{\SE}{\mathsf{ IST}}

\newcommand{\Surjectivity}{Surjectivity}
\newcommand{\DIST}{\mathsf{DIST}}
\newcommand{\dumSURJ}{\mathsf{dSURJ}}
\newcommand{\dumDIST}{\mathsf{dDIST}}

\newcommand{\OR}{\mathsf{ OR}}

\newcommand{\NOR}{\mathsf{ NOR}}

\newcommand{\THR}{\mathsf{THR}}
\newcommand{\bits}{\{-1,1\}}

\newcommand{\B}{\{0,1\}}

\newcommand{\promcomp}{G^{\le N}}

\newcommand{\ls}{\star}

\DeclareMathAlphabet{\mathpzc}{OT1}{pzc}{m}{it}
\DeclareMathAlphabet{\mathcal}{OMS}{cmsy}{m}{n}

\newcommand{\N}{\mathbb{N}}

\newcommand{\tO}{\tilde{O}}
\newcommand{\tOmega}{\tilde{\Omega}}
\newcommand{\tTheta}{\tilde{\Theta}}

\newcommand{\cX}{\mathcal{X}}

\newcommand{\eps}{\varepsilon}

\newcommand{\id}{{\mathds{1}}}

\usepackage{xcolor}

\newcommand{\ignore}[1]{}
\newcommand{\provisionallyremove}[1]{}

\newcommand{\Rcal}{\mathcal{R}}

\newcommand{\Scal}{\mathcal{S}}

\newcommand{\adeg}{\widetilde{\operatorname{deg}}}
\newcommand{\ubdeg}{\widetilde{\operatorname{ubdeg}}}
\newcommand{\dpdeg}{\widetilde{\operatorname{dpdeg}}}

\newcommand{\AND}{\mathsf{AND}}
\newcommand{\GAPAND}{\mathsf{GapAND}}

\renewcommand{\sgn}{\operatorname{sgn}}
\newcommand{\fr}{\operatorname{\#}}

\makeatletter
\newcommand*{\itemequation}[3][]{
  \item
  \begingroup
    \refstepcounter{equation}
    \ifx\\#1\\%
    \else
      \label{#1}
    \fi
    \sbox0{#2}
    \sbox2{$\displaystyle#3\m@th$}
    \sbox4{ \@eqnnum}
    \dimen@=.5\dimexpr\linewidth-\wd2\relax
    \let\CenterInSpace=N
    \ifcase
        \ifdim\wd0>\dimen@
          \z@
        \else
          \ifdim\wd4>\dimen@
            \z@
          \else
            \@ne
          \fi
        \fi
      \let\CenterInSpace=Y
    \fi
    \ifdim\dimexpr\wd0+\wd2+\wd4\relax>\linewidth
      \@latex@warning{Equation is too large}
    \fi
    \noindent
    \rlap{\copy0}
    \ifx\CenterInSpace Y
      \rlap{\hbox to \linewidth{\kern\wd0\hss\copy2\hss\kern\wd4}}
    \else
      \rlap{\hbox to \linewidth{\hfill\copy2\hfill}}
    \fi
    \hbox to \linewidth{\hfill\copy4}
    \hspace{0pt}
  \endgroup
  \ignorespaces
}
\makeatother

\title{The Polynomial Method Strikes Back:\\ Tight Quantum Query Bounds via Dual Polynomials}
\author{
Mark Bun\\
Boston University\\
\textsf{mbun@bu.edu}
\and 
Robin Kothari\\
Microsoft Research\\
\textsf{robin.kothari@microsoft.com}
\and
Justin Thaler\\
Georgetown University\\
\textsf{justin.thaler@georgetown.edu}
}

\date{}
\begin{document}
\maketitle 


\begin{abstract}
The approximate degree of a Boolean function $f$ is the least degree of a real polynomial that approximates $f$ pointwise to error at most $1/3$. The approximate degree of $f$ is known to be a lower bound on the quantum query complexity of $f$ (Beals et al., FOCS 1998 and J.\ ACM 2001).

We resolve or nearly resolve the approximate degree and quantum query complexities of several basic functions. Specifically, we show the following:

\begin{itemize}
\item $k$-distinctness: For any constant $k$, the approximate degree and quantum query complexity of the $k$-distinctness function is $\Omega(n^{3/4-1/(2k)})$. 
This is nearly tight for large $k$, as Belovs (FOCS 2012) has shown that for any constant $k$, the approximate degree and quantum query complexity of $k$-distinctness is $O(n^{3/4-1/(2^{k+2}-4)})$.

\item Image Size Testing: The approximate degree and quantum query complexity of testing the size of the image of a function $[n] \to [n]$ is $\tilde{\Omega}(n^{1/2})$. This proves a conjecture of Ambainis et al. (SODA 2016), and it implies tight lower bounds on the approximate degree and quantum query complexity of the following natural problems. 
\begin{itemize} 
\item $k$-junta testing: A tight $\tilde{\Omega}(k^{1/2})$ lower bound for $k$-junta testing, answering the main open question of Ambainis et al. (SODA 2016).

\item Statistical Distance from Uniform: A tight $\tilde{\Omega}(n^{1/2})$ lower bound for approximating the statistical distance from uniform of a distribution,
answering the main question left open by Bravyi et al.\ (STACS 2010 and IEEE Trans.\ Inf.\ Theory 2011).

\item Shannon entropy: A tight $\tilde{\Omega}(n^{1/2})$ lower bound for approximating Shannon entropy up to a certain additive constant, answering a question of Li and Wu (2017). 
\end{itemize}

\item Surjectivity: The approximate degree of the Surjectivity function is $\tilde{\Omega}(n^{3/4})$. The best prior lower bound was $\Omega(n^{2/3})$. Our result matches an upper bound of $\tilde{O}(n^{3/4})$ due to Sherstov (STOC 2018), which we reprove using different techniques. The quantum query complexity of this function is known to be $\Theta(n)$ (Beame and Machmouchi, Quantum Inf. Comput. 2012 and Sherstov, FOCS 2015).

\end{itemize}

Our upper bound for Surjectivity introduces new techniques for approximating Boolean
functions by low-degree polynomials. Our lower bounds are proved by significantly refining techniques recently introduced by Bun and Thaler (FOCS 2017). 
\end{abstract}


\newpage 

\setcounter{tocdepth}{2}
\setlength{\cftbeforesecskip}{0.55em}
\tableofcontents

\newpage 
\section{Introduction}
\para{Approximate degree.}
The approximate degree of a Boolean function $f \colon \{-1, 1\}^n \to \{-1, 1\}$, denoted $\adeg(f)$, is the least degree of a real polynomial $p$ such that $|p(x)-f(x)| \leq 1/3$ for all $x \in \bits^n$. Approximate degree is a basic measure of the complexity of a Boolean function,
and has diverse applications throughout theoretical computer science.

Upper bounds on approximate degree are at the heart of the most powerful known learning
algorithms in a number of models~\cite{ksdnf, klivansservedioomb, kkms, servediotanthaler,
readonceformulae, colt, osnewbounds}, algorithmic approximations for the inclusion-exclusion principle~\cite{kahn,sherstovinclusion}, and algorithms for differentially private data release
\cite{difpriv1, difpriv2}. 
A recent line of work~\cite{tal2, tal1} has used approximate degree upper bounds to show new lower bounds on the formula and graph complexity of explicit functions.

Lower bounds on approximate degree
have enabled progress in several areas of complexity theory, including communication complexity~\cite{patmat, bvdw, comm1, comm2, comm3, comm4, comm6, comm7, comm8, sherstovsurvey}, circuit complexity~\cite{mp, sherstovmajmaj},
oracle separations~\cite{beigel, bchtv}, and secret-sharing~\cite{viola}. Most importantly for this paper, approximate degree lower bounds have been critical in shaping our understanding of \emph{quantum query complexity}~\cite{beals, qqc2, aaronsonshi},

In spite of the importance of approximate degree, major gaps remain in our understanding. In particular, the approximate degrees of many basic functions are still unknown. Our goal in this paper is to resolve the approximate degrees of many natural functions which had previously withstood characterization.

\para{Quantum query complexity.}
While resolving the approximate degree of basic functions of interest is a test of our understanding of approximate degree, it is also motivated by the study of quantum algorithms.
In the quantum query model, a quantum algorithm is given query access to the bits of an input $x$, and the goal is to compute some function $f$ of $x$ while minimizing the number of queried bits. Quantum query complexity captures much of the power of quantum computing, and most quantum algorithms were discovered in or can easily be described in the query setting. 

Approximate degree was one of the first general lower bound techniques for quantum query complexity. 
In 1998, Beals et al.~\cite{beals} observed that the bounded-error quantum query complexity of a function $f$ is lower bounded by (one half times) the approximate degree of $f$.
Since polynomials are sometimes easier to understand than quantum algorithms, this observation led to a number of new lower bounds on quantum query complexity. 
This method of proving quantum query lower bounds is called the \emph{polynomial method}. 

After several significant quantum query lower bounds were proved via the polynomial method (including the work of Aaronson and Shi~\cite{aaronsonshi}, who proved  optimal lower bounds for the Collision and Element Distinctness problems), the polynomial method took a back seat. Since then, the positive-weights adversary method \cite{Amb02,BSS03,LM04,Zha05} and the newer negative-weights adversary method \cite{negativeweights,adversary2,LMR+11} have become the tools of choice for proving quantum query lower bounds (with some notable exceptions, such as Zhandry's recent tight lower bound for the set equality problem \cite{perm1}). This leads us to our second goal for this work.

In this work, we seek to resolve several open problems in quantum query complexity using approximate degree as the lower bound technique.
A distinct advantage of proving quantum query lower bounds with the polynomial method is that any such bound can be ``lifted'' via Sherstov's \emph{pattern matrix method} \cite{patmat} to a quantum \emph{communication} lower bound (even with unlimited shared entanglement \cite{ls09rank}); such a result is not known for any other quantum query lower bound technique. More generally, using approximate degree as a lower bound technique for quantum query complexity has other advantages, such as the ability to show lower bounds for zero-error and small-error quantum algorithms~\cite{smallerrorquantum}, unbounded-error quantum algorithms~\cite{beals}, and time-space tradeoffs~\cite{KSdW07}.

\para{Quantum query complexity and approximate degree.} 
In this work we illustrate the power of the polynomial method by proving optimal or nearly optimal bounds on several functions studied in the quantum computing community. These results are summarized in \Cref{tab:summary}, and definitions of the problems considered can be found in \Cref{sec:results}.
Since the upper bounds for these functions were shown using quantum algorithms, our results resolve both the quantum query complexity and approximate degree of these functions.

\setlength{\tabcolsep}{9pt}
\renewcommand{\arraystretch}{1.3}
\begin{table}[t]
\centering
{\small
\begin{tabular}{ l l l l }
\hline
Problem & Best Prior Upper Bound & Our Lower Bound & Best Prior Lower Bound \\
\hline
& & &\\[-12pt]
$k$-distinctness 
& $O(n^{3/4-1/{(2^{k + 2} - 4)}})$ \cite{belovs} 
& $\tilde{\Omega}(n^{3/4-{1}/{(2k)}})$ 
& $\tilde{\Omega}(n^{2/3})$ \cite{aaronsonshi}\\[1pt]
Image Size Testing 
& $O(\sqrt{n} \log n)$ \cite{juntatesting} 
& $\tilde{\Omega}(\sqrt{n})$ 
& $\tilde{\Omega}(n^{1/3})$ \cite{juntatesting} \\
$k$-junta Testing 
& $O(\sqrt{k} \log k)$ \cite{juntatesting} 
& $\tilde{\Omega}(\sqrt{k})$ 
& $\tilde{\Omega}(k^{1/3})$ \cite{juntatesting} \\
$\SDU$ 
& $O(\sqrt{n})$ \cite{bravyi} 
& $\tilde{\Omega}(\sqrt{n})$
& $\tilde{\Omega}(n^{1/3})$ \cite{bravyi, aaronsonshi} \\
Shannon Entropy 
& $\tilde{O}(\sqrt{n})$ \cite{bravyi, entropy} 
& $\tilde{\Omega}(\sqrt{n})$ 
& $\tilde{\Omega}(n^{1/3})$ \cite{entropy} \\[4pt]
\hline
\end{tabular}
\caption{Our lower bounds on quantum query complexity and approximate degree vs.~prior work.}
\label{tab:summary}
}
\end{table}

For most of the functions studied in this paper, the positive-weights adversary bound provably cannot show optimal lower bounds due to the certificate complexity barrier~\cite{Zha05, spalekszegedy} and the property testing barrier~\cite{negativeweights}. 
While these barriers do not apply to the negative-weights variant (which is actually capable for proving tight quantum query lower bounds for \emph{all} functions~\cite{adversary2,LMR+11}), the negative-weights adversary method is often challenging to apply to specific problems, and the problems we consider have withstood characterization for a long time.

\medskip
For the functions presented in \Cref{tab:summary}, the approximate degree and quantum query complexity are essentially the same. 
This is not the case for the Surjectivity function, which has played an
important role in the literature on approximate degree and quantum query complexity.
Specifically, Beame and Machmouchi~\cite{beame} showed that Surjectivity has quantum query complexity $\tilde{\Theta}(n)$. On the other hand, Sherstov recently showed that Surjectivity has approximate degree $\tilde{O}(n^{3/4})$ \cite{sherstovpersonal}. 
Surjectivity is the only known example of a ``natural'' function separating approximate degree from quantum query complexity; prior examples of such functions~\cite{ambainissep, cheatsheets} were contrived, and (unlike Surjectivity) specifically constructed to separate the two measures.

Our final result gives a full characterization of the approximate degree of Surjectivity. We prove a new lower bound of $\tilde{\Omega}(n^{3/4})$, which matches Sherstov's upper bound up to logarithmic factors. We also give a new construction of an approximating polynomial of degree $\tilde{O}(n^{3/4})$, using very different techniques than \cite{sherstovpersonal}. We believe that our proof
of this $\tilde{O}(n^{3/4})$ upper bound is of independent interest. In particular,
our lower bound proof for Surjectivity is specifically tailored to showing 
optimality of our upper bound construction, in
a sense that can be made formal via complementary slackness. 
We are optimistic that our approximation techniques will be
useful for showing additional tight approximate degree bounds in the future.

\setlength{\tabcolsep}{8pt}
\renewcommand{\arraystretch}{1.3}
\begin{table}
\centering
{\small
\begin{tabular}{ l l l l l }
\hline
Problem  & Prior Upper Bound & Our Upper Bound & Our Lower Bound & Prior Lower Bound\\
\hline
& & & &\\[-12pt] 
Surjectivity & $\tO(n^{3/4})$ \cite{sherstovpersonal} & $\tO(n^{3/4})$ &$\tOmega(n^{3/4})$ & $\tilde{\Omega}(n^{2/3})$ \cite{aaronsonshi} \\[3pt]
\hline
\end{tabular}
\caption{Our bounds on the approximate degree of Surjectivity vs. prior work.}
\label{tab:summary2}
}
\end{table}

\subsection{Our Results}
\label{sec:results}

We now describe our results and prior work on these functions in more detail.

\subsubsection{Functions Considered}
We now informally describe the functions studied in this paper. These functions are formally defined in \Cref{sec:functions}.

Let $R$ be a power of two and $N \geq R$, and let $n=N \cdot \log_2 R$.
Most of the functions that we consider interpret their inputs in $\bits^n$ as a list of $N$ numbers from a range $[R]$, and determine whether this list satisfies various natural properties. 
We let the \emph{frequency} $f_i$ of range item $i \in R$ denote the number of times $i$ appears in the input list.

In this paper we study the following functions in which the input is $N$ numbers from a range $[R]$:
\begin{myitemize}
    \item Surjectivity ($\SURJ$): Do all range items appear at least once?
    \item $k$-distinctness: Is there a range item that appears $k$ or more times?
    \item Image Size Testing: Decide if all range items appear at least once or if at most $\gamma \cdot R$ range items appear at least once, under the promise that one of these is true.
    \item Statistical distance from uniform ($\SDU$): Interpret the input as a probability distribution $p$, where $p_i = f_i/N$. Compute the statistical distance of $p$ from the uniform distribution over $R$ up to some small additive error $\eps$.
    \item Shannon entropy: Interpret the input as a probability distribution $p$, where $p_i = f_i/N$. Compute the Shannon entropy $\sum_{i \in R} p_i \cdot \log(1/p_i)$ of $p$ up to additive error $\eps$.
\end{myitemize}
An additional function we consider that does not fit neatly into the framework above
is $k$-junta testing. 
\begin{myitemize}
    \item $k$-junta testing: Given an input in $\bits^n$ representing the truth table of a function $\bits^{\log n} \to \bits$, determine whether this function depends on at most $k$ of its input bits, or is at least $\eps$-far from any such function. 
\end{myitemize}

We resolve or nearly resolve the quantum query complexity and/or approximate degree of all of the functions above. Our lower bounds for $\SURJ$,
$k$-distinctness, Image Size Testing, $\SDU$, and entropy approximation all require $N$ to be ``sufficiently larger'' than $R$, by a certain constant  factor. For simplicity, throughout this introduction we do not make this requirement explicit, and for this reason we label
the theorems in this introduction informal.

\subsubsection{Results in Detail}

\para{Surjectivity.} In the Surjectivity problem we are given $N$ numbers from $[R]$ and must decide if every range item appears at least once in the input.

The quantum query complexity of this problem was studied by Beame and Machmouchi~\cite{beame}, who proved a lower bound of $\tOmega(n)$, which was later improved by Sherstov to the optimal $\Theta(n)$~\cite{sherstov15}.  Beame and Machmouchi~\cite{beame} explicitly leave open the question of characterizing the approximate degree of Surjectivity.
Recently, Sherstov~\cite{sherstovpersonal} showed an upper bound of $\tO(n^{3/4})$ on the approximate degree of this function.
The best prior lower bound was $\tilde{\Omega}(n^{2/3})$ \cite{aaronsonshi, adegsurj}. 

We give a completely different construction of an approximating polynomial for Surjectivity with degree $\tO(n^{3/4})$. We also prove a matching lower bound, which shows that the approximate degree of the Surjectivity function is $\tilde{\Theta}(n^{3/4})$. 

\begin{theorem}[Informal] \label{thm:introsurj} The approximate degree of $\SURJ$ is $\tilde{\Theta}(n^{3/4})$.
\end{theorem}

\para{$k$-distinctness.} In this problem, we are given $N$ numbers in $[R]$ and must decide if any range item appears at least $k$ times in the list (i.e., is there an $i\in [R]$ with $f_i \geq k$?). This generalizes the well-studied Element Distinctness problem, which is the same as $2$-distinctness.

Ambainis~\cite{ambainis} first used quantum walks to give an $O(n^{k/(k+1)})$ upper bound on the quantum query complexity of any problem with certificates of size $k$, including $k$-distinctness and $k$-sum.\footnote{In the $k$-sum problem, we are given $N$ numbers in $[R]$ and asked to decide if any $k$ of them sum to $0 \pmod R$.}
Later, Belovs introduced a beautiful new framework for designing quantum algorithms~\cite{belovsframework} and used it to improve the upper bound for $k$-distinctness to $O(n^{3/4-1/(2^{k + 2} - 4)})$ \cite{belovs}. 
Several subsequent works have used Belovs' $k$-distinctness algorithm as a black-box subroutine for solving more complicated problems (e.g., \cite{entropy, montanaro}).

As for lower bounds, Aaronson and Shi \cite{aaronsonshi} established an $\tilde{\Omega}(n^{2/3})$ lower bound on the approximate degree of $k$-distinctness for any $k\geq 2$. Belovs and \v{S}palek used the adversary method to prove a lower bound of $\Omega(n^{k/(k+1)})$ on the quantum query
complexity of $k$-sum, showing that Ambainis' algorithm is tight for $k$-sum. 
They asked whether their techniques can prove an $\omega(n^{2/3})$ quantum query lower bound for $k$-distinctness. 
We achieve this goal, but using the polynomial method instead of the adversary method.
Our main result is the following:

\begin{theorem}[Informal]  \label{thm:introkdist} For any $k\geq 2$, the approximate degree and quantum query complexity of $k$-distinctness 
is $\tilde{\Omega}(n^{3/4-1/(2k)})$. \end{theorem}

This is nearly tight for large $k$, as it approaches Belovs' upper bound of $O(n^{3/4-1/(2^{k + 2} - 4)})$. 
Note that both bounds approach $\Theta(n^{3/4})$ as $k\to \infty$.
It remains an intriguing open question to close the gap between $n^{3/4-1/(2^{k + 2} - 4)}$ and $n^{3/4-1/(2k)}$, especially for small values of $k \geq 3$. 

Our $k$-distinctness lower bound also implies an $\tilde{\Omega}(n^{3/4-1/(2k)})$ lower bound on the quantum query complexity of approximating the maximum frequency, $F_{\infty}$, of any element up to relative error less than $1/k$ \cite{montanaro}, improving over the previous best bound of $\tilde{\Omega}(n^{2/3})$.

\para{Image Size Testing.} In this problem, we are given $N$ numbers in $[R]$ and $0<\gamma<1$, and must decide if every range item appears at least once or if at most $\gamma\cdot R$ range items appear at least once. 
We show for any $\gamma > 0$, the problem has approximate degree and quantum query complexity $\tilde{\Omega}(\sqrt{n})$. 
This holds as long as $N=c\cdot R$ for a certain constant $c>0$.

\begin{theorem}[Informal]  \label{thm:imagesize} The approximate degree and quantum query complexity of Image Size Testing
is $\tilde{\Omega}(\sqrt{n})$. 
\end{theorem}

This lower bound is tight, matching a quantum algorithm of Ambainis, Belovs, Regev, and de Wolf \cite{juntatesting}, and resolves a conjecture from their work. The previous best lower bound was $\Omega(n^{1/3})$ \cite{juntatesting} obtained via reduction to the Collision lower bound \cite{aaronsonshi}. The classical query complexity of this problem is $\Theta(n/\log n)$~\cite{VV11}.

The version of image size testing we define is actually a special case of the one studied in \cite{juntatesting}. The version we define is solvable via the following simple algorithm making $O(\sqrt{n})$ queries: pick a random range item, and Grover search for an instance of that range item.  The fact that our lower bound holds even for this special case of the problem considered in prior works obviously only makes our lower bound stronger.

This lower bound also serves as a starting point to establish the next three lower bounds.

\para{$k$-junta Testing.} In this problem, we are given the truth table of a Boolean function and have to determine if the function depends on at most $k$ variables or if it is $\epsilon$-far from any such function. 

The best classical algorithm for this problem uses $O(k \log k+k/\eps)$ queries~\cite{Bla09}. The problem was first studied in the quantum setting by  At{\i}c{\i} and Servedio \cite{roccoquantum}, who gave a quantum algorithm making $O(k/\eps)$ queries. This was later improved by Ambainis et al.~\cite{juntatesting} to $\tilde{O}(\sqrt{k/\eps})$. They also proved a lower bound of $\Omega(k^{1/3})$. 
Via a connection established by Ambainis et al., our image size testing lower bound implies a $\tilde{\Omega}(\sqrt{k})$ lower bound on the approximate degree and quantum query complexity of $k$-junta testing (for some $\eps=\Omega(1)$). 

\begin{theorem}[Informal]  \label{thm:juntatesting} The approximate degree and quantum query complexity of $k$-junta testing 
is $\tilde{\Omega}(\sqrt{k})$. \end{theorem}

This matches the upper bound of~\cite{juntatesting}, resolving the main open question from their work.

\para{Statistical Distance From Uniform $(\SDU)$.} In this problem, we are given $N$ numbers in $[R]$, which we interpret 
as a probability distribution $p$, where $p_i=f_i/N$, the fraction of times $i$ appears.
The goal is to compute the statistical distance between $p$ and the uniform distribution to error $\eps$.

This problem was studied by Bravyi, Harrow, and Hassidim \cite{bravyi}, who gave an $O(\sqrt{n})$-query quantum algorithm approximating the statistical distance between
 two input distributions to additive error $\eps=\Omega(1)$.
We show that the approximate degree and quantum query complexity of this task are $\tilde{\Omega}(\sqrt{n})$, even when one of the distributions is known to be the uniform distribution. 

\begin{theorem}[Informal]  \label{thm:sdu} There is a constant $c > 0$ such that the approximate degree and quantum query complexity of approximating the statistical distribution of a distribution over a range of size $n$ from the uniform distribution over the same range to additive error $c$ is 
is $\tilde{\Omega}(\sqrt{n})$. 
\end{theorem}

This matches the upper bound of Bravyi et al.~\cite{bravyi} and answers the main question left open from that work. Note that the classical query complexity of this problem is $\Theta(n/\log n)$~\cite{VV11}.

\para{Entropy Approximation.} As in the previous problem, we interpret the input as a probability distribution, and the goal is to compute its Shannon entropy to additive error $\eps$. The classical query complexity of this problem is $\Theta(n/\log n)$~\cite{VV11}. We show that, for some $\eps=\Omega(1)$,  the approximate degree and quantum query complexity are $\tilde{\Omega}(\sqrt{n})$.

 \begin{theorem}[Informal]  \label{thm:introentropy} There is a constant $c > 0$ such that the approximate degree and quantum query complexity of approximating the Shannon
 entropy of a distribution over a range of size $n$ to additive error $c$ is 
is $\tilde{\Omega}(\sqrt{n})$. \end{theorem}

This too is tight, answering a question of Li and Wu \cite{entropy}.

\subsection{Prior Work on Lower Bounding Approximate Degree}
A relatively new lower-bound technique for approximate degree called the \emph{method of dual polynomials} plays an essential role in our paper.
This method of dual polynomials 
dates back to work of Sherstov \cite{sherstovhalfspaces1} and \v{S}palek \cite{spalek},
though dual polynomials had been used earlier to resolve longstanding questions
in communication complexity \cite{patmat, shizhu, sherstovmajmaj, comm7, comm9}.
To prove a lower bound for a function $f$
via this method, one exhibits an explicit \emph{dual polynomial} for $f$, 
which is a dual solution to a certain linear program capturing the approximate degree of $f$. 

A notable feature of the method of dual polynomials is that it is lossless,
in the sense that it can exhibit a tight lower bound on the approximate degree of 
any function $f$ (though actually applying the method to specific functions may be highly challenging). Prior to the method of dual polynomials, the primary tool available for proving approximate degree lower bounds was symmetrization, introduced by Minsky and Papert \cite{mp} in the 1960s. Although powerful, symmetrization is not a lossless technique.

Most prior work on 
the method of dual polynomials 
can be understood as establishing hardness amplification results. 
Such results show how to take a function $f$ that is ``somewhat hard'' to approximate
by low-degree polynomials, and turn $f$ into a related function $g$ that is much harder to approximate.
Here, harder means either that $g$ requires larger degree
to approximate to the same error as $f$, or that approximations to $g$ of a given degree incur 
much larger error than do approximations to $f$ of the same degree. 

\vspace{-.4mm}
\medskip \noindent \textbf{Results for Block-Composed Functions.} 
Until very recently, the method of dual polynomials had been used exclusively to prove hardness amplification results for \emph{block-composed} functions. 
That is, the harder function $g$ would be obtained by block-composing $f$ with another 
function $h$, i.e., $g=h \circ f$. 
Here, a function
$g \colon \bits^{n \cdot m} \to \bits$ is the block-composition of $h \colon \bits^n \to \bits$ and $f \colon \bits^m \to \bits$ 
if $g$ interprets its input as a sequence of $n$ blocks, 
applies $f$
to each block, and then feeds the $n$ outputs into $h$.

The method of dual polynomials turns out to be particularly suited to analyzing block-composed functions, as there 
are sophisticated ways of ``combining'' dual witnesses for $h$ and $f$ individually to give 
an effective dual witness for $h \circ f$ \cite{sherstovhalfspaces1, shizhu, sherstovandor, bt13, bt14, thaler, sherstov14, sherstov15, bchtv}.
Prior work on analyzing block-composed functions has, for example,
resolved the approximate degree of the function $f(x) = \bigwedge_{i=1}^n \bigvee_{j=1}^m x_{ij}$, known as the AND-OR tree, which had been open for 19 years
\cite{bt13, sherstovandor}, 
established new lower bounds for AC$^0$ under basic complexity measures including
discrepancy \cite{bt14, thaler, sherstov14, sherstov15}, sign-rank \cite{bt16}, and threshold degree \cite{sherstov15, sherstov14}, and
resolved a number of open questions about the power of statistical zero knowledge proofs \cite{bchtv}.

\medskip
 \noindent \textbf{Beyond Block-Composed Functions.}  While the aforementioned results
led to considerable progress in complexity theory,
 many basic questions require understanding the approximate degree
 of \emph{non}-block-composed functions. One prominent example with many
 applications is to exhibit an AC$^0$ circuit over $n$ variables with approximate degree $\Omega(n)$.
\label{op1sec}
Until very recently, the best result in this direction was Aaronson and Shi's well-known $\tilde{\Omega}(n^{2/3})$ lower bound on the approximate degree of the Element Distinctness function (which is equivalent to $k$-distinctness for $k=2$)\cite{aaronsonshi}. However, Bun and Thaler \cite{adegsurj} recently achieved a near-resolution of this problem by proving the following theorem.

\begin{theorem}[Bun and Thaler \cite{adegsurj}] \label{thm:mainthm} For any constant $\delta \!>\! 0$,
there is an AC$^0$ circuit with approximate degree $\Omega(n^{1-\delta})$. \end{theorem}

\noindent The reason that \Cref{thm:mainthm} \emph{required} moving beyond block-composed functions is the following result of Sherstov~\cite{sherstovrobust}.

\begin{theorem}[Sherstov] \label{thm:robust}
For any Boolean functions $f$ and $h$, $\adeg(h \circ f) = O\left(\adeg(h) \cdot \adeg(f)\right)$.
\end{theorem}

\Cref{thm:robust} implies that the approximate degree of $h \circ f$
(viewed as a function of its input size) is never higher than the approximate degree of
$f$ or $h$ individually (viewed as a function of their input sizes). For example, if $f$ and $h$ are both functions 
on $n$ inputs, and
both have approximate degree $O(n^{1/2})$, then $h \circ f$ has $N:=n^2$ inputs, and by \Cref{thm:robust}, $\adeg(h \circ f) = O(n^{1/2} \cdot n^{1/2}) = O(N^{1/2})$.

This means that
block-composing multiple AC$^0$ functions does not result in a function of higher approximate degree (as a function of its input size) than that of the individual functions.
Bun and Thaler \cite{adegsurj} overcome
 this hurdle by introducing a way of analyzing functions that
 cannot be written as a block-composition of simpler functions.
 
 Bun and Thaler's techniques set the stage to resolve the approximate degree
 of many basic functions using the method of dual polynomials. However, they were not refined enough to accomplish this on their own. Our lower bounds in this paper
 are obtained by refining and extending the methods of \cite{adegsurj}.

\subsection{Our Techniques}
\label{s:overview}
In order to describe our techniques, it is helpful to explain the process
by which we discovered the tight $\tilde{\Theta}(n^{3/4})$ 
lower and upper bounds for Surjectivity (cf. \Cref{thm:introsurj}).
It has previously been observed \cite{thaler, bt13, adegsurj} that 
optimal dual polynomials for a function $f$
tend to be tailored (in a sense that can be made precise via complementary slackness)
to showing optimality of some specific approximation technique for $f$. Hence, 
constructing a
dual polynomial for $f$ can provide a strong hint as to how to construct an optimal approximation for $f$,
and vice versa.

\para{Upper Bound for Surjectivity.}
In \cite{adegsurj}, Bun and Thaler  constructed a dual polynomial 
witnessing a suboptimal bound of $\tilde{\Omega}(n^{2/3})$ for $\mathsf{SURJ}$. Even though
this dual polynomial is suboptimal, it still provided a major clue as to what an optimal 
approximation for $\mathsf{SURJ}$ should look like: 
 it curiously
 ignored all inputs failing to satisfy the following condition.
\begin{condition} \label{cond}
Every range item has frequency at most $T$, for a specific threshold $T=O(N^{1/3}) \ll N$.
\end{condition}
\noindent This 
suggested that an optimal approximation for $\mathsf{SURJ}$ should treat inputs satisfying  \Cref{cond}  differently than other inputs,
leading us to the following multi-phase construction (for clarity and brevity,
this overview is simplified).
The first phase constructs a polynomial $p$ of degree $O(n^{3/4})$
approximating $\mathsf{SURJ}$ on all inputs satisfying \Cref{cond}. However, $p$
may be \emph{exponentially large} on other inputs. 
The second phase constructs a polynomial $q$ of degree $O(n^{3/4})$ that is \emph{exponentially small} on inputs $x$ that 
do \emph{not} satisfy \Cref{cond} (in particular, $q(x) \ll 1/p(x)$ for such $x$), and is close to 1 otherwise. The product $p \cdot q$ still
approximates $\mathsf{SURJ}$ on inputs satisfying \Cref{cond},
and is \emph{exponentially small} on all other inputs. Notice that 
 $\deg(p \cdot q) \leq \deg(p) + \deg(q) = O(n^{3/4})$.
Combining the above with an additional averaging step (the details of which we omit from this introduction)
yields an approximation to $\mathsf{SURJ}$  that is accurate on \emph{all} inputs.

\para{Lower Bound for Surjectivity.}
With the $O(n^{3/4})$ upper bound in hand, we were able
to identify the fundamental bottleneck
preventing further improvement of the upper bound. 
This suggested 
a way to refine the techniques of \cite{adegsurj} to
prove a matching lower bound. Once the tight lower bound
for $\SURJ$ was established, we were able to identify additional refinements
to analyze the other functions that we consider. We now describe this in more detail.

Bun and Thaler's \cite{adegsurj} (suboptimal) lower bound analysis for $\SURJ$ proceeds in two stages.
In the first stage, proving a lower bound for $\SURJ$ (on $N$ input list items and $R$ range items) is reduced to the problem
of proving a lower bound for the \emph{block-composed} function $\AND_R \circ \OR_N$,\footnote{When
it is not clear from context, we use subscripts to denote the number of variables on which a function is defined.} under the promise that the input has Hamming weight at most $N$.\footnote{Note that a reduction the other direction is straightforward: 
to approximate $\SURJ$, it suffices to approximate 
$\AND_R \circ \OR_N$ on inputs of Hamming weight exactly $N$. This is because
$\SURJ$ can be expressed as an $\AND_R$ (over all range items $r\in [R]$) of the $\OR_N$ (over all input bits $i\in[N]$) of ``Is input $x_i$ equal to $r$''? Each predicate of the form in quotes is computed exactly by a polynomial of degree $\log R$, since
it depends on only $\log R$ of the inputs, and exactly $N$ of these predicates (one for each $i \in [N]$)
evaluate to TRUE.}
In this paper, we use this stage of their analysis unmodified.

The second stage proves an $\tilde{\Omega}(R^{2/3})$ lower bound for the latter problem by leveraging much of the machinery developed to analyze the approximate degree of block-composed functions \cite{bt13, sherstovandor, razborovsherstov}. 
To describe this machinery, we require the following notion.
A dual polynomial that witnesses the fact that $\adeg_{\eps}(f_n) \geq d$ is a 
function $\psi \colon \bits^n \to \bits$ satisfying three properties:
\begin{itemize}
\item $\sum_{x \in \bits^n}\psi(x) \cdot f(x) > \eps$. If $\psi$ satisfies this condition, it is said to be \emph{well-correlated} with $f$.
\item  $\sum_{x \in \bits^n} |\psi(x)| = 1$. If $\psi$ satisfies this condition, it is said to have $\ell_1$-norm equal to 1.
\item For all polynomials $p \colon \bits^n \to \R$ of degree less than $d$, we have $\sum_{x \in \bits^n} p(x) \cdot \psi(x) = 0$. 
If $\psi$ satisfies this condition, it is said to have \emph{pure high degree} at least $d$.
\end{itemize}

\label{s:overviewoftheproofs}

In more detail, the second stage of the analysis from \cite{adegsurj} itself proceeds in two steps. First, the authors consider a dual witness $\psi$ for the high approximate degree of $\AND_R \circ \OR_N$ that was constructed in prior work \cite{bt13}. $\psi$ is constructed by taking dual witnesses $\phi$ and $\gamma$ for the high approximate degrees of $\AND_R$ and $\OR_N$ individually, and ``combining'' them in a specific
way \cite{shizhu, sherstovhalfspaces1, lee} to obtain a dual witness for the high approximate degree of their block-composition $\AND_R \circ \OR_N$. 

 Unfortunately, $\psi$ only  witnesses a lower bound for $\AND_R \circ \OR_N$ \emph{without} the promise that the Hamming weight of the input is at most $N$. To address this issue, it is enough to ``post-process'' $\psi$ so that it no longer ``exploits''
any inputs of Hamming weight larger than $N$ (formally, $\psi(x)$ should equal zero
for any inputs in $\bits^{R \cdot N}$ of Hamming weight more than $N$). The authors accomplish this
by observing that $\psi$ ``almost ignores'' all such inputs (i.e., it places exponentially little total mass on all such inputs), and hence it is possible
to perturb $\psi$ to make it \emph{completely} ignore all such inputs.

Key to this step is the fact that the ``inner'' dual witness $\gamma$ for the high approximate degree of the $\OR_N$ function satisfies a ``Hamming weight decay'' condition:
\begin{equation} \label{decaykey} |\gamma(x)| \cdot \binom{N}{|x|}\leq \frac{1}{\poly(|x|)},\end{equation} for a suitable polynomial function. 

To improve the lower bound for $\SURJ$ from $\tilde{\Omega}(n^{2/3})$ 
to the optimal $\tilde{\Omega}(n^{3/4})$, we observe that $\gamma$ in fact
satisfies a much stronger decay condition: while the inverse-polynomial decay property of Equation \eqref{decaykey} is tight for small Hamming weights $|x|$, $|\gamma(x)|$ actually
decays \emph{exponentially quickly} once $|x|$ is larger than a certain threshold $t$. 
This observation is enough to obtain the tight $\tilde{\Omega}(n^{3/4})$ lower bound for $\SURJ$. 

For intuition, it is worth mentioning that a primal formulation of the dual decay condition that we exploit shows that any low-degree polynomial $p$ that is an accurate
approximation to $\OR_N$ on low Hamming weight inputs requires large degree,
\emph{even if $|p(x)|$ is allowed to be exponentially large for inputs of Hamming weight
more than $t$.}\footnote{We do not formally describe this primal formulation of the dual decay condition, because it is not necessary to prove any of the results in this paper.} This is precisely the bottleneck that prevents us from improving our
upper bound for $\SURJ$ to $o(N^{3/4})$. In this sense, our dual witness
is intuitively tailored to showing optimality of the techniques used in our upper bound.

\para{Other Lower Bounds.}
To obtain the lower bound for $k$-distinctness, the first stage
of the analysis of \cite{adegsurj} reduces to a question about the approximate degree
of the block composed function $\OR_R \circ \mathsf{THR}^k_N$, under the promise that the input has Hamming weight at most $N$.
Here $\mathsf{THR}^k_N \colon \bits^N \to \bits$ denotes the function that evaluates to $-1$ if and only if the Hamming weight of its input is at least $k$.
By constructing a suitable dual witness for $\mathsf{THR}^k_N$, and combining it with a dual witness for $\OR_N$ via similar techniques as in our construction for $\SURJ$, we are able to prove our $\Omega(n^{3/4-1/(2k)})$ lower bound for $k$-distinctness. (This description glosses over several significant technical issues that
must be dealt with to ensure that the combined dual witness is well-correlated with $\OR_R \circ \mathsf{THR}^k_N$).\footnote{Specifically, our analysis requires the dual witness $\gamma$
for $\mathsf{THR}^k_N$ to be very well-correlated with $\mathsf{THR}^k_N$ in a certain one-sided sense (roughly, we need the probability distribution $|\gamma|$ to have the property that,
conditioned on $\gamma$ outputting a negative value, the input to $\gamma$
is in $\left(\mathsf{THR}^k_N\right)^{-1}(-1)$ with probability at least $1-1/(3R)$). 
This property was not required in the analysis for $\SURJ$, which is why
our lower bound for $\SURJ$ is larger by a factor of $n^{1/(2k)}$ than
our lower bound for $k$-distinctness. This seemingly technical issue 
is at least partially intrinsic:
a polynomial loss compared to the $\Omega(n^{3/4})$ lower
bound for $\SURJ$ is unavoidable, owing to Belovs' $n^{3/4 - \Omega(1)}$ upper bound \cite{belovs} for $k$-distinctness.}

\medskip

Recall that our lower bounds for $k$-junta testing, $\SDU$, and entropy approximation
are derived as consequences of our lower bound for image size testing. 
The connection between image size testing and junta testing was established by Ambainis et al. \cite{juntatesting}.
The reason that the image testing lower bound implies lower bounds for $\SDU$
is the following. Consider any distribution $p$ over $[R]$ such that all probabilities $p_i$ are integer multiples of $1/N$ for some $N=O(R)$. Then if $p$ has full support, $p$ is guaranteed 
to be somewhat close to uniform, while if $p$ has small support, $p$
must be very far from uniform.
We obtain our lower bound for entropy approximation using
a simple reduction from $\SDU$ due to Vadhan \cite{vadhan}.

\medskip 
To obtain our lower bound for Image Size Testing, we observe that the first stage
of the analysis of \cite{adegsurj} reduces to a question about the approximate degree
of the block composed function $\GAPAND_R \circ \OR_N$, under the promise that the input has Hamming weight at most $N$.
Here, $\GAPAND_R$
is the promise function that outputs $-1$ if all of its inputs equal $-1$, outputs $+1$ if fewer than $\gamma \cdot R$ of its inputs are $-1$, and is undefined otherwise.

Roughly speaking,  we obtain the desired $\tilde{\Omega}(n^{1/2})$ lower bound by combining a dual witness for $\GAPAND_R \circ \OR_N$ from prior work \cite{bt14} with the same techniques as in our construction for $\SURJ$. However, additional technical refinements to the analysis of \cite{adegsurj} are required to obtain our results. In 
particular, the analysis of \cite{adegsurj} only provides a lower bound
for $\SURJ$ if $N = \Omega( R \cdot \log^2(R))$. But in order to infer 
our lower bound for
$\SDU$ and entropy approximation (as well as $k$-junta testing for $\eps=\Omega(1)$), it is essential that the lower bound hold
for $N=O(R)$. This is because a distribution with full support is guaranteed 
to be $\Omega(1)$-close to uniform if all probabilities are integer multiples
of $1/N$ with $N=O(R)$, but this is not the case otherwise. (Consider, e.g., a distribution
that places mass $1-1/\log^2(R)$ on a single range item, and spreads out the
remaining mass evenly over all other range items). Refining the methods of \cite{adegsurj} to yield lower bounds even when $N=O(R)$ requires a significantly
more delicate analysis than in \cite{adegsurj}.

A diagram indicating how we obtain our results and the relationships that we establish between problems is given in \Cref{fig:reductions}. 

\begin{figure}
    \centering
    \begin{minipage}{0.9\textwidth}
        \includegraphics[clip,width=0.9\textwidth]{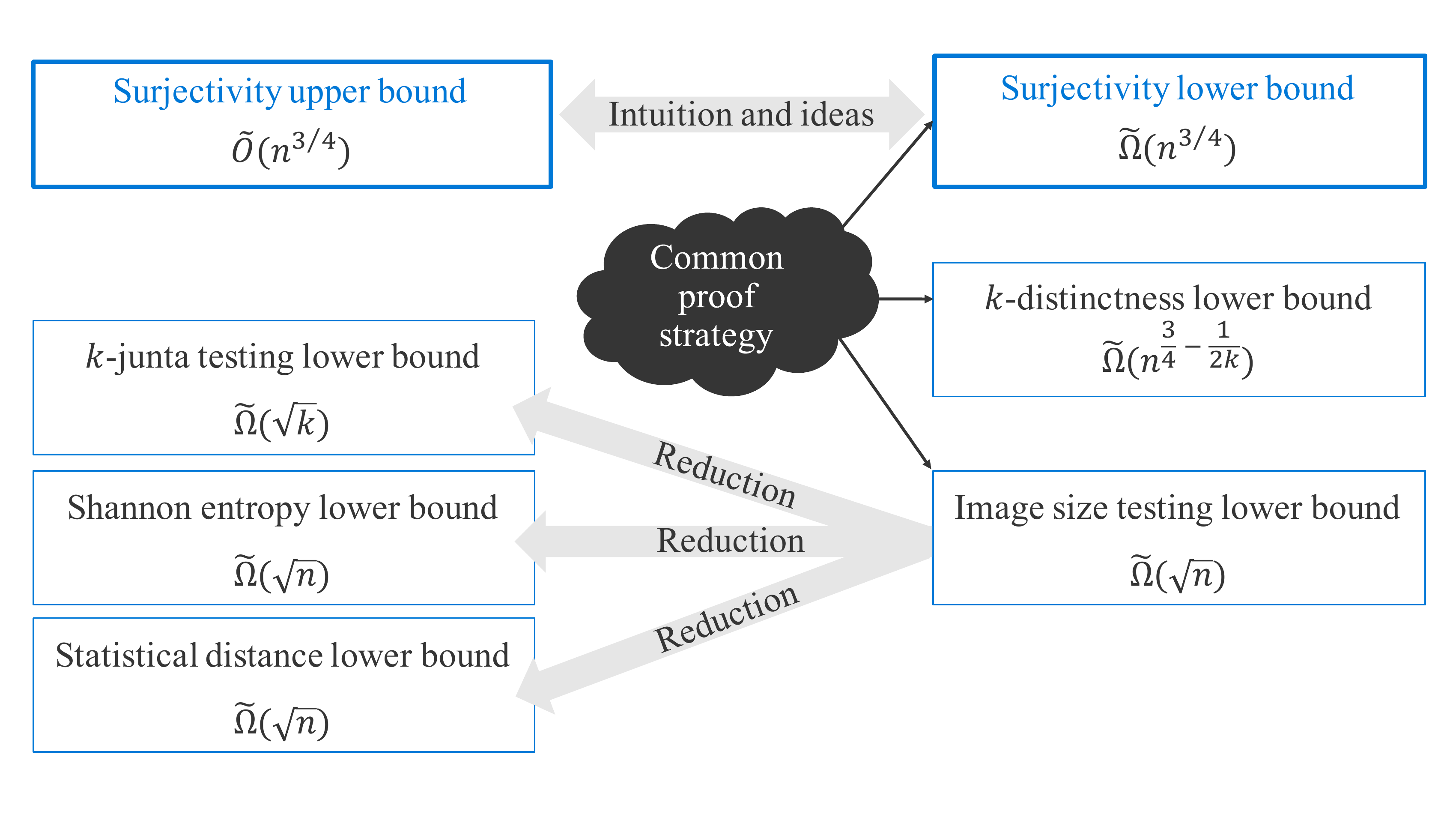} 
        \caption{Diagram of reductions and relationships amongst our results. \label{fig:reductions}}
        \end{minipage}
\end{figure}

\subsection{Outline for the Rest of the Paper}
\Cref{s:prelims} covers preliminary definitions and lemmas.
\Cref{s:upper} presents the $\tilde{O}(n^{3/4})$ upper bound
for $\SURJ$, while \Cref{s:surjlower} presents the matching $\tilde{\Omega}(n^{3/4})$
lower bound. 
\Cref{s:kdist} gives the $\tilde{\Omega}(n^{3/4-1/(2k)})$ lower 
bound for $k$-distinctness. \Cref{s:final}
presents the lower bound
for Image Size Testing, and its implications for 
junta testing, $\SDU$, and Shannon entropy approximation.
\Cref{s:conclusion} concludes by briefly describing
some additional consequences of our results, as well as a number
of open questions and directions for future work.


\section{Preliminaries}
\label{s:prelims}

\subsection{Notation}

For a natural number $N$, let $[N] = \{1, 2, \dots, N\}$ and $[N]_0 = \{0, 1, 2, \dots, N\}$. All logarithms are taken in base $2$ unless otherwise noted.
As is standard, we say that a function $f(n)$ is in $\tilde{O}(h(n))$ if there exists a constant $k$ such that $f(n)$ is in $O(h(n) \cdot \log^k(h(n)))$. 

We will frequently work with Boolean functions under the promise that their inputs have low Hamming weight. To this end, we introduce the following notation for the set of low-Hamming weight inputs.

\begin{definition} \label{def:ht} 
For $1 \le T \le n$, let $H^n_{\leq T}$ denote the subset of $\bits^n$
consisting of all inputs Hamming weight at most $T$. We use 
$|x|$ to denote the Hamming weight of an input $x \in \bits^n$,
so $H^n_{\leq T} = \{x \in \bits^n \colon |x| \leq T\}$. 
\end{definition}

\subsection{Two Variants of Approximate Degree and Their
Dual Formulations} 
There are two natural notions of approximate degree for promise problems
(i.e., for functions defined on a strict subset $\mathcal{X}$ of  $\bits^n$). 
One notion requires an approximating polynomial $p$ to be bounded
in absolute value even on inputs in $\bits^n \setminus \mathcal{X}$. 
The other places no restrictions on $p$ outside of the promise $\mathcal{X}$. 
In this work, we make use of \emph{both} notions. Hence, we must introduce 
some (non-standard) notation
to distinguish the two. 

\begin{definition}[Approximate Degree With Boundedness Outside of the Promise Required]
Let $\eps > 0$ and $\cX \subseteq \bits^n$. The $\eps$-\emph{approximate degree} of a Boolean function $f : \cX \to \bits$, denoted $\adeg_\eps(f)$, is the least degree of a real polynomial $p : \cX \to \R$ such that $|p(x) - f(x)| \le \eps$ for all $x \in \cX$ and $|p(x)| \le 1 + \eps$ for all $x \in \bits^n \setminus \cX$. We use the term \emph{approximate degree} without qualification to refer to $\adeg(f) = \adeg_{1/3}(f)$.
\end{definition}

The following standard dual formulation of this first variant of approximate degree
can be found in, e.g., \cite{bttoc}.

\begin{proposition} \label{prop:duality}
Let $\genericdomain \subseteq \bits^n$, and let $f : \genericdomain \to \{-1,1\}$.
Then $\adeg_{\eps}(f) \geq d$ if and only if there exists a function 
$\psi \colon \bits^n \to \R$ satisfying the following properties.
\begin{align} \label{eq:corr} &\sum_{x \in \genericdomain}\psi(x) \cdot f(x) - \sum_{x \in \bits^n \setminus \genericdomain} |\psi(x)| > \eps, \\
 \label{eq:unitnorm} &\sum_{x \in \bits^n} |\psi(x)| = 1, \text{ and }  \\
 \label{eq:phd} &\text{ For every polynomial } p \colon \bits^n \to \R \text{ of degree less than } d, \sum_{x \in \bits^n} p(x) \cdot \psi(x) = 0. 
\end{align}
\end{proposition}

We will refer to functions $\psi \colon \bits^n \to \reals$ as \emph{dual polynomials}.
We refer to $\sum_{x \in \bits^n} |\psi(x)|$ as the $\ell_1$-norm of $\psi$, and denote this quantity by $\|\psi\|_1$. If $\psi$ satisfies Equation \eqref{eq:phd}, it is said to have \emph{pure high degree} at least $d$.

\medskip Given a function $\psi \colon \bits^n \to \R$, and a (possibly partial)
function $f \colon \genericdomain \to \bits$, where $\genericdomain \subseteq \bits^n$,
we let $\langle f, \psi \rangle := \sum_{x \in \genericdomain} f(x) \cdot \psi(x) - \sum_{x \in \bits^n \setminus \genericdomain} |\psi(x)|$, and refer to this as the correlation of $f$ and $\psi$. So Condition \eqref{eq:corr}
is equivalent to requiring $\psi$ and $f$ to have correlation great than $\eps$.

\begin{definition}[Approximate Degree With Unboundedness Permitted Outside of the Promise]
Let $\eps > 0$ and $\cX$ be a finite set. The $\eps$-\emph{unbounded approximate degree} of a Boolean function $f : \cX \to \bits$, denoted $\ubdeg_\eps(f)$, is the least degree of a real polynomial $p : \cX \to \R$ such that $|p(x) - f(x)| \le \eps$ for all $x \in \cX$ (if $\cX$ is a strict subset of a larger domain, then no constraints are placed on $p(x)$ for $x \not\in \cX$). We use the term \emph{unbounded approximate degree} without qualification to refer to $\ubdeg(f) = \ubdeg_{1/3}(f)$.
\end{definition}

The following standard dual formulation of this second variant of approximate degree
can be found in, e.g., \cite{patmat}. A dual polynomial $\psi : \bits^n \to \bits$ witnessing the fact that $\ubdeg_\eps(f) \ge d$ is the same as a dual witness for $\adeg_\eps(f) \ge d$, but with the additional requirement that $\psi(x) = 0$ outside of $\genericdomain$.

\begin{proposition} \label{prop:ub-duality}
Let $\genericdomain \subseteq \bits^n$, and let $f : \genericdomain \to \{-1,1\}$.
Then $\ubdeg_{\eps}(f) \geq d$ if and only if there exists a function 
$\psi \colon \bits^n \to \R$ satisfying the following properties.
\begin{align} \label{eq:ub-zero} &\psi(x) = 0 \text{ for all } x \not\in \genericdomain, \\
\label{eq:ub-corr} &\sum_{x \in \genericdomain}\psi(x) \cdot f(x) > \eps, \\
 \label{eq:ub-unitnorm} &\sum_{x \in \bits^n} |\psi(x)| = 1, \text{ and }  \\
 \label{eq:ub-phd} &\text{ For every polynomial } p \colon \bits^n \to \R \text{ of degree less than } d, \sum_{x \in \bits^n} p(x) \cdot \psi(x) = 0. 
\end{align}

\end{proposition}

Observe that $\adeg(f)$ and $\ubdeg(f)$ coincide for total functions $f$. To avoid notational
clutter, when referring to the approximate degree of total functions, we will
use the shorter notation $\adeg(f)$.

\subsection{Basic Facts about Polynomial Approximations}

The seminal work of Nisan and Szegedy~\cite{nisanszegedy} gave tight bounds on the approximate degree of the $\AND_n$ and $\OR_n$ functions.

\begin{lemma}
\label{andornor}
For any constant $\eps \in (0, 1)$, the functions $\AND$ and $\OR$ on $n$ bits have $\eps$-approximate degree $\Theta(n^{1/2})$,
and the same holds for their negations.
\end{lemma}

Approximate degree is invariant under negating the inputs or output of a function, and hence the result for $\AND$ implies the result for $\mathsf{NAND}$, $\OR$, etc.

The following lemma, which forms the basis of the well-known symmetrization argument, is due to Minsky and Papert \cite{mp}. 
\begin{lemma} \label{lem:mp}
Let $p \colon \bits^n \to \bits$ be an arbitrary polynomial and let $[n]_0$
denote the set $\{0, 1, \dots, n\}$. Then
there is a univariate polynomial $q \colon \R \to \R$ of degree at most $\deg(p)$
such that 
\[q(t) = \frac{1}{\binom{n}{ t}} \sum_{x \in \bits^n \colon |x|=t} p(x)\]
for all $t \in [n]_0$.
\end{lemma}

\subsection{Functions of Interest}
\label{sec:functions}
We give formal definitions of the Surjectivity and $k$-distinctness we consider in this work, as well as several variations that will be helpful in proving our lower bounds. The formal definitions of the other functions we study, including Image Size Testing, $\SDU$, and the Shannon Entropy functions, are deferred to \Cref{s:final}.

\subsubsection{Surjectivity}

\begin{definition}\label{def:surj}
For $N \ge R$, we define $\SURJ_{N, R} : [R]^N \to \bits$ by $\SURJ_{N, R}(s_1, \dots, s_N) = -1$ iff for every $j \in [R]$, there exists an $i$ such that $s_i = j$.
\end{definition}

When $N$ and $R$ are clear from context, we will often refer to the function $\SURJ$ without the explicit dependence on these parameters. It will sometimes be convenient to think of the input to $\SURJ_{N, R}$ as a function
mapping $\bits^n \to \bits$ rather than $[R]^N \to \bits$. 
When needed, we assume that $R$ is a power of 2 and an element of $[R]$ is encoded in binary using $\log R$ bits.
In this case we will view \Surjectivity\ as a function on $n=N\log R$ bits, i.e.,  $\SURJ: \bits^{n} \to \bits$.

For technical reasons, when proving lower bounds, it will be more convenient to work with a variant of $\SURJ$ where the range $[R]$ is augmented by a ``dummy element'' $0$ that is simply ignored by the function. That is, while any of the items $s_1, \dots, s_N$ may take the dummy value $0$, the presence of a $0$ in the input is not required for the input to be deemed surjective. We denote this variant of Surjectivity by $\dumSURJ$. More formally:

\begin{definition}\label{def:dumsurj}
For $N \ge R$, we define $\dumSURJ_{N, R} : [R]_0^N \to \bits$ by $\dumSURJ_{N, R}(s_1, \dots, s_N) = -1$ iff for every $j \in [R]$, there exists an $i$ such that $s_i = j$.
\end{definition}

The following simple reduction shows that a lower bound on the approximate degree of $\dumSURJ$ implies a lower bound for $\SURJ$ itself.

\begin{proposition} \label{prop:dumsurj-reduction}
Let $\eps > 0$ and $N \ge R$. Then
\[ \adeg_\eps(\dumSURJ_{N, R}) \le \adeg_\eps(\SURJ_{N + 1, R + 1}) \cdot \log(R+1).\]
\end{proposition}

\begin{proof}
Let $p : \bits^{(N+1) \cdot \log(R+1)} \to \bits$ be a polynomial of degree $d$ that $\eps$-approximates $\SURJ_{N + 1, R+1}$. We will use $p$ to construct a polynomial of degree $d$ that $\eps$-approximates $\dumSURJ_{N, R}$. Recall that an input to $\dumSURJ_{N, R}$ takes the form $(s_1, \dots, s_N)$ where each $s_i$ is the binary representation of a number in $[R]_0$. Define the transformation $T : [R]_0 \to [R + 1]$ by
\[T(s) = \begin{cases}
R + 1 & \text{ if } s = 0 \\
s & \text{ otherwise}.
\end{cases}\]
Note that as a mapping between binary representations, the function $T$ is exactly computed by a vector of polynomials of degree at most $\log(R+1)$. For every $(s_1, \dots, s_N) \in [R]_0^N$, observe that
\[\dumSURJ_{N, R}(s_1, \dots, s_N) = \SURJ_{N + 1, R+1}(T(s_1), \dots, T(s_N), R+1).\]
Hence, the polynomial
\[p(T(s_1), \dots, T(s_N), R+1)\]
is a polynomial of degree $d \cdot \log(R+1)$ that $\eps$-approximates $\dumSURJ_{N, R}$.
\end{proof}

\subsubsection{$k$-Distinctness}

\begin{definition}
For integers $k, N, R$ with $k \le N$, define the function $\DIST^k_{N, R} : [R]^N \to \bits$ by $\DIST^k_{N, R}(s_1, \dots, s_N) = -1$ iff there exist $r \in [R]$ and distinct indices $i_1, \dots, i_k$ such that $s_{i_1} = \dots = s_{i_k} = r$.
\end{definition}

As with Surjectivity, it will be convenient to work with a variant of $k$-distinctness where $[R]$ is augmented with a dummy item:

\begin{definition}
For integers $k, N, R$ with $k \le N$, define the function $\dumDIST^k_{N, R} : [R]_0^N \to \bits$ by $\dumDIST^k_{N, R}(s_1, \dots, s_N) = -1$ iff there exist $r \in [R]$ and distinct indices $i_1, \dots, i_k$ such that $s_{i_1} = \dots = s_{i_k} = r$.
\end{definition}

For $k \ge 2$, a lower bound on the approximate degree of $\dumDIST^k$ implies a lower bound on the approximate degree of $\DIST^k$. The restriction that $k \ge 2$ is essential, because the function $\DIST^1$ is the constant function that evaluates to $\mathsf{TRUE}$ on \emph{any} input (since at least one range item must always have frequency at least one), whereas $\dumDIST^1$ contains $\OR_N$ as a subfunction, and hence has approximate degree at least $\Omega(\sqrt{N})$.

\begin{proposition} \label{prop:dumdist-reduction}
Let $\eps > 0$, $N, R \in \N$, and $k \ge 2$. Then
\[ \deg_\eps(\dumDIST^k_{N, R}) \le \deg_\eps(\DIST^k_{N, R+N}) \cdot \log(R+1).\]
\end{proposition}

\begin{proof}
The proof is similar to that of \Cref{prop:dumsurj-reduction}, but uses a slightly more involved reduction. Let $p : \bits^{2N\cdot \log(R+N)} \to \bits$ be a polynomial of degree $d$ that $\eps$-approximates $\DIST^k_{N, R+N}$. We will use $p$ to construct a polynomial of degree $d$ that $\eps$-approximates $\dumDIST^k_{N, R}$. For each $i = 1, \dots, R$, define a transformation $T_i : [R]_0 \to [R+N]$ by
\[T_i(s) = \begin{cases}
R + i & \text{ if } s = 0 \\
s & \text{ otherwise}.
\end{cases}\]
As a mapping between binary representations, the function $T=(T_1, \dots, T_N)$ is exactly computed by a vector of polynomials of degree at most $\log(R+1)$. For every $(s_1, \dots, s_N) \in [R]_0^N$, observe that
\[\dumDIST^k_{N, R}(s_1, \dots, s_N) = \DIST^k_{N, R+N}(T_1(s_1), \dots, T_N(s_N)).\]
Hence, the polynomial
\[p(T(s_1), \dots, T(s_N))\]
is a polynomial of degree $d \cdot \log(R+1)$ that $\eps$-approximates $\dumSURJ_{N, R}$.
\end{proof}

\subsection{Connecting Symmetric Properties and Block Composed Functions}
\label{sec:step1}

An important ingredient in~\cite{adegsurj} is the relationship between the approximate degree of a property of a list of numbers (such as $\SURJ$) and the approximate degree of a simpler block composed function, defined as follows.

\begin{definition}
For functions $f : Y^n \to Z$ and $g: X \to Y$, define the \emph{block composition} $f \circ g : X^n \to Z$ by $(f\circ g)(x_1, \dots, x_n) = f(g(x_1), \dots, g(x_n))$, for all $x_1, \dots, x_n \in X$.
\end{definition}

Fix $R, N \in \N$, let $f : \bits^R \to \bits$ and let $g : \bits^N \to \bits$. Suppose $g$ is a \emph{symmetric} function, in the sense that for any $x \in \bits^N$ and any permutation $\sigma : [N] \to [N]$, we have
\[g(x_1, \dots, x_N) = g(x_{\sigma(1)}, \dots, x_{\sigma(N)}).\]
Equivalently, the value of $g$ on any input $x$ depends only on its Hamming weight $|x|$.

The functions $f$ and $g$ give rise to two functions. The first, which we denote by $F^{\operatorname{prop}} : [R]_0^N \to \bits$, is a certain property of a list of numbers $s_1, \dots, s_N \in [R]_0$. The second, which we denote by $F^{\le N} : H_{\le N}^{N \cdot R} \to \bits$, is the block composition of $f$ and $g$ restricted to inputs of Hamming weight at most $N$. Formally, these functions are defined as:
\begin{align*}
&F^{\operatorname{prop}}(s_1, \dots, s_N) = f(g(\id[s_1 = 1], \dots, \id[s_N = 1]), \dots, g(\id[s_1=R], \dots,\id[s_N=R])) \\
&F^{\le N}(x_1, \dots, x_R) = \begin{cases} f(g(x_1), \dots, g(x_R)) & \text{ if } x_1, \dots, x_R \in \bits^N,  |x_1| + \dots + |x_R| \le N.\\
\text{undefined} & \text{otherwise}.
\end{cases}
\end{align*}

The following proposition from ~\cite{adegsurj}, which in turn relies heavily on a clever symmetrization argument due to due to Ambainis~\cite{ambainissmallrange}, relates the approximate degrees of the two functions $F^{\operatorname{prop}}$ and $F^{\le N}$.

\begin{theorem} \label{thm:main-reduction}
Let $f : \bits^R \to \bits$ be any function and let $g : \bits^N \to \bits$ be a symmetric function. Then for $F^{\operatorname{prop}}$ and $F^{\le N}$ defined above, and for any $\eps > 0$, we have
\[\deg_\eps(F^{\operatorname{prop}}) \ge \ubdeg_\eps(F^{\le N}).\]
\end{theorem}

In the case where $f = \AND_R$ and $g = \OR_N$, the function $F^{\operatorname{prop}}(s_1, \dots, s_N)$ is the Surjectivity function augmented with a dummy item, $\dumSURJ_{N, R}(s_1, \dots, s_N)$. Hence,

\begin{corollary} \label{cor:surj-reduction}
Let $N, R \in \N$. Then for any $\eps > 0$,
\[\adeg_\eps(\dumSURJ_{N, R}) \ge \ubdeg_\eps(F^{\le N})\]
where $F^{\le N} : H^{N\cdot R}_{\le N} \to \bits$ is the restriction of $\AND_R \circ \OR_N$ to $H^{N\cdot R}_{\le N}$.
\end{corollary}

\begin{definition}
For integers $k, N$ with $k \le N$, define the function $\THR^k_N : \bits^N \to \bits$ by $\THR^k_N(x) = -1$ iff $|x| \ge k$.
\end{definition}

If we let $f = \OR_R$ and $g = \THR_N^k$, then the function $F^{\operatorname{prop}}$ is the dummy augmented $k$-distinctness function $\dumDIST_{N, R}^k$.

\begin{corollary} \label{cor:dist-reduction}
Let $N, R \in \N$. Then for any $\eps > 0$,
\[\adeg_\eps(\dumDIST^k_{N, R}) \ge \ubdeg_\eps(G^{\le N})\]
where $G^{\le N} : H^{N\cdot R}_{\le N} \to \bits$ is the restriction of $\OR_R \circ \THR^k_N$ to $H^{N\cdot R}_{\le N}$.
\end{corollary}

\subsection{The Dual Block Method}
This section collects definitions and preliminary results on the dual block method \cite{shizhu, lee, sherstovhalfspaces1}
for constructing dual witnesses for a block composed function $F \circ f$ by combining dual witnesses
for $F$ and $f$ respectively.
\label{sec:ls}
\begin{definition}
Let $\Psi : \bits^M \to \R$ and $\psi : \bits^m \to \R$ be functions that are not identically zero. Let $x=(x_1, \dots, x_M) \in \left(\bits^{m}\right)^M$. Define the \emph{dual block composition} of $\Psi$ and $\psi$, denoted $\Psi \ls \psi : (\bits^m)^M \to \R$, by
\[(\Psi \ls \psi)(x_1, \dots, x_M) = 2^M \cdot \Psi(\dots, \sgn\left(\psi(x_i)\right), \dots) \cdot \prod_{i = 1}^M |\psi(x_i)|.\]
\end{definition}

\begin{proposition}[\!\!\cite{sherstovhalfspaces1,adegsurj}] \label{prop:ls}
The dual block composition satisfies the following properties:
\paragraph{\emph{Preservation of $\ell_1$-norm}:} If $\|\Psi\|_1 = 1$, $\|\psi\|_1 = 1$, and $\langle \psi, \mathbf{1}\rangle = 0$, then
\begin{equation}\|\Psi \ls \psi\|_1 = 1. \label{eqn:ls-norm}\end{equation}

\paragraph{\emph{Multiplicativity of pure high degree}:} If $\langle \Psi, P \rangle = 0$ for every polynomial $P \colon \bits^M \to \bits$ of degree less than  $D$, and $\langle \psi, p \rangle = 0$ for every polynomial $p \colon \bits^m \to \bits$ of degree less than $d$, then
for every polynomial $q \colon \bits^{m \cdot M} \to \bits$,
\begin{equation}\deg q < D \cdot d \implies \langle \Psi \ls \psi, q \rangle = 0. \label{eqn:ls-phd} \end{equation}

\paragraph{\emph{Associativity}:} For every $\zeta: \bits^{m_\zeta} \to \R$, $\varphi: \bits^{m_\varphi} \to \R$, and $\psi: \bits^{m_\psi} \to \R$, we have
\begin{equation} (\zeta \ls \varphi) \ls \psi = \zeta \ls (\varphi \ls \psi). \label{eqn:ls-assoc} \end{equation}
\end{proposition}

\subsection{A Refinement of a Technical Lemma from Prior Work}
 The following technical proposition
 refines techniques of Bun and Thaler \cite{adegsurj}. 
This proposition is useful for ``zeroing out'' the mass that 
a dual polynomial $\xi$
places on inputs of high Hamming weight,
if $\xi$ is obtained via the dual-block method.

\begin{definition}
Let $M \in \N$ and $\alpha, \beta > 0$. A function $\omega: [M]_0 \to \R$ satisfies the $(\alpha, \beta)$-decay condition if
\begin{align}
&\sum_{t = 0}^M \omega(t) = 0, \label{eqn:omega-balanced}\\
&\sum_{t = 0}^M |\omega(t)| = 1, \label{eqn:omega-norm}\\
&|\omega(t)| \le \alpha \exp (-\beta t) / t^2 \quad \forall t = 1, 2, \dots, M. \label{eqn:omega-decay}
\end{align}
\end{definition}

\begin{restatable}{proposition}{btzeroing} \label{prop:btzeroing}
Let $R \in \N$ be sufficiently large, and let $\Phi : \bits^R \to \R$ with $\|\Phi\|_1 = 1$. For $M \le R$, let $\omega : [M]_0 \to \R$ satisfy the $(\alpha, \beta)$-decay condition with parameters $1 \leq \alpha \le R^2, \beta \in (4\ln^2 R/ \sqrt{\alpha} R, 1)$.

Let $N = \lceil 20 \sqrt{\alpha} \rceil R$, and define $\psi: \bits^N \to \R$ by $\psi(x) = \omega(|x|) / \binom{N}{|x|}$. If $D < N$ is such that
\begin{equation}
\text{ For every polynomial } p \text{ with } \deg p < D, \text{ we have } \langle \Phi \ls \psi, p \rangle = 0, \label{eqn:prelim-phd}
\end{equation}
then there exist $\Delta \ge \beta\sqrt{\alpha} R / 4\ln^2 R$ and a function $\zeta: (\bits^N)^R \to \R$ such that
\begin{align}& \text{ For every polynomial } p \text{ with } \deg p < \min\{D, \Delta\}, \text{ we have } \langle \zeta, p \rangle = 0, \label{eqn:final-phd} \\
&\|\zeta - \Phi \ls \psi\|_1 \le \frac{2}{9}, \label{eqn:final-corr} \\
&\|\zeta\|_1 = 1, \label{eqn:final-norm}\\
&\zeta \text{ is supported on } H^{N \cdot R}_{\le N}. \label{eqn:final-support}
\end{align}
\end{restatable}

The key refinement of \Cref{prop:btzeroing} relative to the analysis of Bun and Thaler
is that \Cref{prop:btzeroing} applies when $N=\Theta(R)$ (assuming $\alpha=O(1))$. In contrast, the techniques of Bun and Thaler required $N=\Omega(R \cdot \log^2 R)$. As indicated in \Cref{s:overviewoftheproofs}, this refinement will be essential in obtaining our lower bounds for $\SDU$, entropy approximation, and junta testing for constant proximity parameter.

The proof of \Cref{prop:btzeroing} occurs in two steps. First, in \Cref{prop:prelim-mass}, we show that $\xi = \Phi \ls \psi$ places an exponentially small amount of mass on inputs outside of $H_{\le N}^{N \cdot R}$. Second, in \Cref{prop:correction}, we construct a correction object $\nu$ that zeroes out the mass $\xi$ places outside of $H_{\le N}^{N \cdot R}$ without decreasing its pure high degree. Combining $\xi$ with $\nu$ yields the desired object $\zeta$.

\begin{proposition} \label{prop:prelim-mass}
Let $\Phi : \bits^R \to \R$ and $\psi : \bits^n \to \R$ satisfy the conditions of \Cref{prop:btzeroing}. Then for sufficiently large $R$, there exists $\Delta \ge \beta \sqrt{\alpha} R / 4\ln^2 R$ such that, for $N = \lceil 20 \sqrt{\alpha} \rceil R$,
\begin{equation}
\sum_{x \notin H_{\le N}^{N \cdot R}} |(\Phi \ls \psi)(x)| \le (2NR)^{-2\Delta}. \label{eqn:prelim-mass}
\end{equation}
\end{proposition}

\begin{proof}
 Recall that $\psi(x) = \omega(|x|) / \binom{N}{|x|}$ where $\omega \colon [M]_{0} \to \R$. By Equation~\eqref{eqn:omega-balanced} we may write $\omega = \omega_{+1} - \omega_{-1}$  where $\omega_{+1}$ and $\omega_{-1}$ are non-negative
functions satisfying
\begin{equation} \label{eq:sweeteq} \sum_{t=0}^k \omega_{+1}(t) = \sum_{t=0}^k \omega_{-1}(t) = 1/2.\end{equation}

By the definition of dual block composition, we have
\[(\Phi \ls \psi)(x_1, \dots, x_R) = 2^R \cdot \Phi(\dots, \sgn\left(\psi(x_i)\right), \dots) \cdot \prod_{i = 1}^R |\psi(x_i)|.\]
Consequently,
\begin{align}
\notag \sum_{x \notin H_{\le N}^{N \cdot R}} |(\Phi \ls \psi)(x)| &= 2^R \sum_{z \in \bits^R} |\Phi(z)| \left( \sum_{\substack{(x_1, \dots, x_R) \notin  H_{\le N}^{N \cdot R} \text{ s.t.} \\ \notag \sgn\left( \psi(x_1)\right) = z_1, \dots, \sgn\left(\psi(x_R)\right) = z_R}} \prod_{i = 1}^R |\psi(x_i)| \right) \\
&= 2^R \sum_{z \in \bits^R} |\Phi(z)| \left( \sum_{(x_1, \dots, x_R) \notin H_{\le N}^{N \cdot R}} \prod_{i = 1}^R \frac{\omega_{z_i}(|x_i|)}{\binom{N}{ |x_i|}} \right). \label{finallineman}
\end{align}

Observe that for any $(t_1, \dots, t_R) \in [M]_0^R$, the number of inputs $(x_1, \dots, x_R) \in \left(\bits^N\right)^R$
such that $|x_i| = t_i$ for all $i \in [R]$ is exactly $\prod_{i=1}^{R} \binom{N}{ t_i}$. 
Hence, defining $$P = \{(t_1, \dots, t_R) \in [M]_0^R : t_1 + \dots + t_R > N\},$$ we
may rewrite Expression
\eqref{finallineman} as
\[2^R \sum_{z \in \bits^R} |\Phi(z)| \left( \sum_{(t_1, \dots, t_R) \in P} \prod_{i = 1}^R \omega_{z_i}(t_i) \right).\]
To control this quantity, we appeal to the following combinatorial lemma,
whose proof we delay to \Cref{s:proofofcomb} (this lemma
is a substantial refinement of \cite[Lemma 32]{adegsurj}).

\begin{restatable}{lemma}{combinatorial} \label{lem:combinatorial}
Let $\alpha \le R^2$, let $\beta \in (4\ln^2 R/ \sqrt{\alpha} R, 1)$, and let $R$ be sufficiently large. Let $N = \lceil 20 \sqrt{\alpha}\rceil R$. Let $\eta_i :[M]_0 \to \R$, for $i = 1, \dots R$, be a sequence of non-negative functions where for every $i$,
\begin{align}&\sum_{r = 0}^M \eta_i(r) \le 1 / 2     \label{eqn:eta-bound} \\
&\eta_i(r) \le \alpha \exp(-\beta r) / r^2  \qquad \forall r = 1, \dots, M.     \label{eqn:eta-decay} 
\end{align}
Let $P = \{\vec{t} = (t_1, \dots, t_R) \in [M]_0^R : t_1 + \dots + t_R > N\}$. Then
\[\sum_{\vec{t} \in P} \prod_{i = 1}^R  \eta_i(t_i) \le 2^{-R} \cdot (2NR)^{-2\Delta}\]
where $\Delta \ge \beta\sqrt{\alpha} R / 4\ln^2 R$.
\end{restatable}

Observe that the functions $\omega_{z_i}$ satisfy Condition \eqref{eqn:eta-bound} (cf. Equation \eqref{eq:sweeteq}) and Condition \eqref{eqn:eta-decay} (cf. Property~\eqref{eqn:omega-decay}). We complete the proof by bounding
\begin{align*}
2^R \sum_{z \in \bits^R} |\Phi(z)| \left( \sum_{\vec{t} \in P} \prod_{i = 1}^R \omega_{z_i}(t_i) \right) &\le 2^R \sum_{z \in \bits^R} |\Phi(z)| \cdot \left(2^{-R} \cdot (2NR)^{-2\Delta}\right)\\
&=(2NR)^{-2 \Delta} .
\end{align*}
Here, the equality appeals to the condition that $\|\Phi\|_1 = 1$.
\end{proof}

Our final dual witness $\zeta$ is obtained by modifying $\xi = \Phi \ls \psi$ to zero out all of the mass it places on inputs of total Hamming weight larger than $N$. The following proposition captures the conditions under which this postprocessing step can be done.

\begin{proposition}[{\cite[Proposition 34]{adegsurj}, building on~\cite{razborovsherstov}}] \label{prop:correction}
Let $N \ge R > D$ and let $\xi : (\bits^N)^R \to \R$ be any function such that
\[\sum_{x \notin H_{\le N}^{N \cdot R}} |\xi(x)| \le (2NR)^{-2D}.\]
Then there exists a function $\nu : (\bits^N)^R \to \R$ such that
\begin{align*}
  &\text{For all polynomials } p \colon (\bits^N)^R  \to \R\text{, } \deg p < D \implies \langle \nu, p \rangle = 0    \\
   &\|\nu\|_1 \le 1/10     \\
   &|x| > N \implies \nu(x) = \xi(x).    
\end{align*}
\end{proposition}

\begin{remark}
\Cref{prop:correction} is framed exclusively in the language 
of dual polynomials: it states that if a dual polynomial $\xi$ places very little mass on inputs
of Hamming weight more than $N$,
then there exists another dual polynomial $\nu$ satisfying certain useful properties
(we ultimately use $\nu$ to zero out the mass that $\xi$ places on inputs of Hamming weight more than $N$).
\Cref{prop:correction} also has a clean \emph{primal} formulation. Roughly speaking,
it is equivalent to a bound on the growth rate of any polynomial of degree at most $D$ that is bounded at
all inputs of Hamming weight at most $D$. We direct the interested
reader to \cite{violalecturenotes} for details of this primal formulation.
\end{remark}

We are now ready to combine \Cref{prop:prelim-mass} and \Cref{prop:correction} to complete the proof of \Cref{prop:btzeroing}.

\begin{proof}[Proof of \Cref{prop:btzeroing}]
Let $\xi = \Phi \ls \psi$. By \Cref{prop:prelim-mass}, we have
\[\sum_{x \notin H_{\le N}^{N \cdot R}} |(\Phi \ls \psi)(x)| \le (2 N R)^{-2\Delta},\]
where $\Delta \ge \beta \sqrt{\alpha} R /\ln^2 R$. By \Cref{prop:correction}, there exists a function $\nu : (\bits^N)^R \to \R$ such that
\begin{align}
  &\text{For all polynomials } p \colon (\bits^N)^R  \to \R\text{, } \deg p < \min\{D, \Delta\} \implies \langle \nu, p \rangle = 0    \label{eqn:correction-phd} \\
   &\|\nu\|_1 \le 1/10     \label{eqn:correction-norm}  \\
   &|x| > N \implies \nu(x) = \xi(x).     \label{eqn:correction-correct}  
\end{align}

Observe that $\|\xi - \nu\|_1 > 0$, as
$\|\xi\|_1 = 1$ (cf. Equation \eqref{eqn:ls-norm}) and $\|\nu\|_1 \leq 1/10$ (cf. Inequality \eqref{eqn:correction-norm}).
Define the function
\[\zeta(x) = \frac{\xi(x) - \nu(x)}{\|\xi - \nu\|_1}.\]
Since $\nu(x) = \xi(x)$ whenever $|x| > N$ (cf. Equation~\eqref{eqn:correction-correct}), the function $\zeta$ is supported on the set $H_{\le N}^{N \cdot R}$, establishing~\eqref{eqn:final-support}. We establish~\eqref{eqn:final-corr} by computing
\begin{align*}
\|\zeta - \xi\|_1 &= \sum_{x \in (\bits^N)^R} \left|  \frac{\xi(x) - \nu(x)}{\|\xi - \nu\|_1} - \xi(x)\right| \\
&\le \sum_{x \in (\bits^N)^R} \left( \frac{1}{\|\xi - \nu\|_1} - 1\right)|\xi(x)| + \frac{1}{\|\xi - \nu\|_1}|\nu(x)| \\
&\le \left( \frac{1}{\|\xi\|_1 - \|\nu\|_1} - 1 \right) \cdot \|\xi\|_1 + \frac{1}{\|\xi\|_1 - \|\nu\|_1} \cdot \|\nu\|_1 \\
&\le \left(\frac{1}{1 - 1/10} - 1 \right)+ \frac{1 /10 }{1 - 1/10} \\
&= \frac{2}{9}.
\end{align*}

Equation~\eqref{eqn:final-norm} is immediate from the definition of $\zeta$. Finally,~\eqref{eqn:final-phd} follows from~\eqref{eqn:prelim-phd}, \eqref{eqn:correction-phd}, and linearity.
\end{proof}

\subsubsection{Proof of \texorpdfstring{\Cref{lem:combinatorial}}{Lemma \ref*{lem:combinatorial}}}
\label{s:proofofcomb}
Our final task is to prove \Cref{lem:combinatorial}, restated here for the reader's convenience.

\combinatorial*

\begin{lemma} \label{lem:binomial}
Let $k, n \in \N$ with $k \le n$. Then 
$\binom{n}{k} \le \left(\frac{en}{k}\right)^k$.
\end{lemma}

\begin{lemma} \label{lem:inv-exp-tail}
Let $m \in \N$ and $\eta > 0$. Then
\[\sum_{r = m}^\infty e^{-\eta r} = \frac{1}{1 - e^{-\eta}} \cdot e^{-\eta m}.\]
\end{lemma}

\begin{proof}
We calculate
\[\sum_{r = m}^\infty e^{-\eta r} = e^{-\eta m}  \cdot \sum_{r = 0}^\infty e^{-\eta r} = e^{-\eta m} \cdot \frac{1}{1 - e^{-\eta}}.\qedhere
\]
\end{proof}

\begin{lemma} \label{lem:inv-sqrt}
Let $m \in N$. Then
\[\sum_{r = 1}^m \frac{1}{\sqrt{r}} \le 2 \sqrt{m}.\]
\end{lemma}

\begin{proof}
Using the fact that the function $1/\sqrt{t}$ is decreasing, we may estimate
\[
  \sum_{r = 1}^m \frac{1}{\sqrt{r}} \le \int_{0}^m \frac{1}{\sqrt{t}} \ dt = 2\sqrt{m}.\qedhere
\]
\end{proof}

\begin{proof}

Let $T = \lfloor \frac{N}{2M} \rfloor$.  For each $s \in \{T, \dots, R\}$, let 
\[C(s) = \max \left\{ \frac{N}{6 s \ln R}, \frac{N}{6\sqrt{R \cdot s}} \right\}.\]

We begin with a simple, but important, structural observation about the set $P$. Let $t = (t_1, \dots, t_R) \in [M]_0^R$ be a sequence such that $t_1 + \dots + t_R > N$. Then  we claim that there exists an $s \in \{T \dots, R\}$ such that $t_i \ge C(s)$ for at least $s$ indices $i \in [R]$. To see this, assume without loss of generality that the entries of $\vec{t}$ are sorted so that $t_1 \ge t_2 \ge \dots \ge t_R$. Then there must exist an $s \geq T$ such that $t_s \ge C(s)$. Otherwise, because no $t_i$ can exceed $M$, we would have:
\begin{align*}
t_1 + \dots + t_R  &< T \cdot M + \sum_{s = 1}^R C(s) \\
&\le \frac{N}{2} + \frac{N}{6 \ln R} \sum_{s = 1}^{\lfloor R / \ln^2 R \rfloor} \frac{1}{s} + \frac{N}{6\sqrt{R}} \sum_{s = \lceil R / \ln^2 R \rceil}^R \frac{1}{\sqrt{s}}\\
&\le \frac{N}{2}+ \frac{N}{6} + \frac{N}{3} \\
&= N.
\end{align*}
where the last inequality follows from \Cref{lem:inv-sqrt} and the fact that $\sum_{s = 1}^m s^{-1} \le \ln m + 1$ (and $R$ is sufficiently large). 
Since the entries of $\vec{t}$ are sorted, the preceding values $t_1, \dots, t_{s-1} \geq C(s)$ as well.

For each subset $S \subseteq [R]$, define
\[P_S = \{\vec{t} \in P : t_i \ge C(|S|)\text{ for all indices }i \in S\}.\]
The observations above guarantee that for every $\vec{t}=(t_1, \dots, t_R) \in P$, 
there exists some set $S$ of size at least $s \in \{T, \dots, R\}$ such that 
$t_i \geq C(s)$ for all $i \in S$. Hence
\begin{align*} \hspace{-10mm}
\sum_{\vec{t} \in P} \prod_{i = 1}^R \eta_i(t_i)  & \leq \sum_{s = T}^{R} \sum_{S \subseteq [R] : |S| = s} \sum_{\vec{t} \in P_S} \prod_{i = 1}^R \eta_i(t_i) \\
& \leq \sum_{s = T}^{R} \binom{R}{s} \left( \max_{i \in \{1, \dots, R\}} \left( \sum_{r=\lceil C(s) \rceil}^M \eta_i(r)\right) \right)^s  \left( \max_{i \in \{1, \dots, R\}} \left( \sum_{r = 0}^M \eta_i(r) \right)\right)^{R-s} \hspace{-3.3em}&\hspace{-3em}\\
&\leq  2^{-R}\sum_{s = T}^{R} \binom{R}{s} \left( \sum_{r = \lceil C(s) \rceil}^M 2\alpha \exp(-\beta r ) r^{-2} \right)^s & \hspace{-28mm} \text{by Properties } \eqref{eqn:eta-bound} \text{ and } \eqref{eqn:eta-decay}\\
&\le 2^{-R} \sum_{s = T}^{R} \left( \frac{Re}{s} \right)^s \left(\frac{2\alpha}{(C(s))^2} \right)^{s} \cdot \sum_{r = \lceil C(s) \rceil}^\infty \exp\left(-\beta r s\right) & \hspace{-28mm} \text{by \Cref{lem:binomial}}\\
&\le 2^{-R} \sum_{s = T}^{R} \left( \frac{Re}{s} \right)^s \left(\frac{2\alpha}{(C(s))^2} \right)^{s} \cdot \frac{1}{1 - e^{-\beta s}} \cdot  \exp\left(-\beta s C(s)\right) & \hspace{-28mm} \text{by \Cref{lem:inv-exp-tail}}\\
&\le 2^{-R} \sum_{s = T}^{R} \left( \frac{Re}{s} \right)^s \left(\frac{72\alpha R s}{N^2} \right)^{s} \cdot \frac{2}{\beta} \cdot  \exp\left(-\frac{\beta N}{6\ln R}\right) & \text{by definition of $C(s)$}\\
&\qquad\qquad\text{ and since } \frac{1}{1 - e^{-\beta s}} \le \frac{1}{1 - (1-\beta/2)} = \frac{2}{\beta} \text{ for } s \ge 1, \beta \in (0, 1) \hspace{-10em}&\hspace{-10em} \\
&\le 2^{-R}  \sum_{s = T}^{R} 2^{-s} \cdot \exp\left(-\frac{3 \beta \sqrt{\alpha} R }{ \ln R} + \ln(2/\beta)\right)&\hspace{-28mm}\text{setting } N = \lceil 20 \sqrt{\alpha}\rceil R \\
&\le 2^{-R} \cdot (2 N R)^{-2\Delta},
\end{align*}
where
\[\Delta = \frac{1}{2\ln(2NR)} \cdot \left(\frac{3\beta \sqrt{\alpha} R}{\ln R} - \ln(2/\beta) \right) \ge \frac{\beta \sqrt{\alpha} R}{4\ln^2 R}\]
holds for sufficiently large $R$ by the restrictions placed on $\alpha$ and $\beta$ in the statement of the lemma.
\end{proof}

 
\section{Upper Bound for Surjectivity}
\label{s:upper}
The goal of this section is to prove that the approximate degree of the \Surjectivity\ function (\Cref{def:surj}) is $\tO(N^{3/4})$.
\begin{theorem}
\label{thm:surj_upper} \label{thm:upper}
For any $R \in \mathbb{N}$, we have $\adeg(\SURJ_{N,R}) = \tO(N^{3/4})$.
\end{theorem}

For the remainder of the section, we focus on proving \Cref{thm:upper} in the case that $R=\Theta(N)$. This is without loss of generality by the following argument. 
If $R > N$, then 
\Surjectivity\ is identically false, and hence has (exact) degree 0. 
And if $R=o(N)$, then we can reduce to the case $R=\Theta(N)$ as follows. Let $N'=2N$ and $R' = R + N$;
clearly $R' = \Theta(N)$.
Given an input $x$ to $\SURJ_{N, R}$, obtain an input $x'$ to $\SURJ_{N', R'}$ by appending range elements $R+1, \dots, R+N$ to $x$. 
This construction guarantees that $\SURJ_{N, R}(x) = \SURJ_{N', R'}(x')$. It follows that an approximation of degree $\tilde{O}(N^{3/4})$ for $\SURJ_{N', R'}$
implies an approximation of the same degree for $\SURJ_{N, R}$.

\medskip \noindent \textbf{Section Outline.} This section is structured as follows. 
\Cref{specialnotation} introduces some notation that is specific to this section.
\Cref{s:warmup} provides some intuition for the construction of the polynomial approximation for \Surjectivity\ using the simpler function $\NOR$ as a warmup example. 
\Cref{s:terminology} introduces some terminology that is useful for providing intuitive descriptions of our final polynomial construction using the language of algorithms. Finally, in \Cref{s:formal} we formally apply our strategy to \Surjectivity\ in order to prove \Cref{thm:surj_upper}.

\subsection{Notation}
\label{specialnotation}
In this section, we make a few departures from the notation used in the introduction and later sections in order to more easily convey the algorithmic intuition behind our polynomial constructions. First, we will consider Boolean functions $f : \B^n \to \B$, where $1$ corresponds to logical $\mathsf{TRUE}$ and $0$ corresponds to logical $\mathsf{FALSE}$. For such a Boolean function, we say that a polynomial $p : \B^n \to \R$ is an $\eps$-approximating polynomial for $f$ if $p(x) \in [0, \eps]$ when $f(x) = 0$ and $p(x) \in [1-\eps, 1]$ when $f(x) = 1$. 
For such a polynomial $p$, it will be useful to think of $p(x)$ as representing 
the probability that a randomized or quantum algorithm accepts when run on input $x$.

By extension, throughout this section we will use $\adeg_\eps(f)$ to denote the least degree of a real polynomial $p$ that $\eps$-approximates $f$ in the sense described above, and we will write $\adeg(f) = \adeg_{1/3}(f)$. Given an input $x \in \B^n$,
we will use $|x|$ to denote its Hamming weight (i.e., $|x| = \sum_i x_i$).

\subsection{Warmup: Approximating \texorpdfstring{$\NOR$}{NOR}}
\label{s:warmup}
To convey the intuition behind our polynomial construction for \Surjectivity, we start by considering a much simpler function as an illustrative example. 
Consider the negation of the $\OR$ function on $n$ bits, $\NOR \colon \B^n\to\B$. 
We will give a novel construction of an approximating polynomial for $\NOR$ of degree $\tilde{O}(\sqrt{n})$. 
Of course, this is not terribly interesting since it is already known that $\adeg(\NOR) = \Theta(\sqrt{n})$ (cf. \Cref{andornor}). But this construction highlights the main idea behind the more involved constructions to follow.

In many models of computation, such as deterministic or randomized query complexity, the $\NOR$ function remains just as hard if we restrict to inputs with $|x|=0$ or $|x|=1$ (where $|x|$ denotes the Hamming weight of $x\in\B^n$). 
This is fairly intuitive, since distinguishing these two types of inputs essentially requires finding a single $1$ among $n$ possible locations. 
The fact that these inputs represent the ``hard case'' is true for approximate degree as well:
any polynomial that uniformly approximates $\NOR$ to error $1/3$ on Hamming weights 0 and 1, and merely remains bounded in $[0,1]$ on the rest of the hypercube, has degree $\Omega(\sqrt{n})$ \cite{nisanszegedy}. 

However, if we remove the boundedness constraint on inputs of Hamming weight larger than 1, then there is a polynomial of degree \emph{one} that exactly equals the $\NOR$ function on Hamming weights 0 and 1: namely, the polynomial  $1-\sum_{i} x_i$. That is, if we view $\NOR$ as a promise function mapping ${H}_{\leq 1}^n$ to $\bits$, then its approximate degree is $\Theta(n^{1/2})$, but its \emph{unbounded}
approximate degree is just $1$. 

More generally, suppose that we only want the polynomial to approximate the $\NOR$ function on all inputs of Hamming weight up to $T\leq n$, and we place no restrictions on the polynomial when evaluated at inputs of Hamming weight larger than $T$. This can be achieved by a polynomial of degree $O(\sqrt{T})$ (see \Cref{lem:orT} for a proof). 

Let us refer to the set of low-Hamming weight inputs as $\P$, i.e.,
\begin{equation}
    \P = H^n_{\le T} = \{x\in\B^n: 0\leq |x| \leq T\}. \label{pdefnor}
\end{equation}

The above discussion shows that we can construct a polynomial $p$ of degree $O(\sqrt{T})$ that tightly approximates $\NOR$ on inputs $x\in\P$, though $|p(x)|$ may be exponentially large for $x \notin \P$. 
On the other hand, distinguishing inputs with $|x|=0$ from inputs with $x \notin \P$ also seems ``easy''. 
For example, a randomized algorithm that simply samples $\Theta(n/T)$ bits and declares $|x|=0$ if and only if it does not see a $1$
is correct with high probability. 
Analogously, we can construct a low-degree polynomial $\tilde{q}$ which distinguishes between $|x|=0$ or $x \notin \P$ (and
is bounded in $[0, 1]$ for \emph{all} inputs in $\B^n$):
it is not hard to show (via an explicit construction involving Chebyshev polynomials, or by appeal to a quantum algorithm called quantum counting \cite{quantumcounting})  that there exists a polynomial  $\tilde{q}$ of degree $O(\sqrt{n/T})$ that accomplishes this.

To summarize the above discussion, we can construct polynomials $p$ and $\tilde{q}$, of degree $O(\sqrt{T})$ and $O(\sqrt{n/T})$ respectively, with the following properties:
\begin{align}
p(x) \in \begin{cases} [9/10,1] &\mathrm{if}~ |x| = 0 \\
[0,1/10] & \mathrm{if}~ 1\leq |x| \leq T \\ 
\mathbb{R} & \mathrm{if}~ x \notin \P \end{cases}
\qquad \qquad
\tilde{q}(x) \in \begin{cases} [9/10,1] &\mathrm{if}~ |x| = 0 \\
[0,1] & \mathrm{if}~ 1\leq |x| \leq T \\ 
[0,1/10] & \mathrm{if}~ x \notin \P \end{cases}
\end{align}

Now consider the polynomial $p(x) \cdot \tilde{q}(x)$. This polynomial approximates $\NOR$ on $|x|=0$, since its value is in $[0.81,1]$. 
It also approximates $\NOR$ on $1\leq|x|\leq T$, since its value is in $[0,1/10]$. 
However, when $x \notin \P$, although $\tilde{q}(x)$ is small, we cannot ensure that the product $p(x) \cdot q(x)$ is small,
since we have no control over $p(x)$ for such $x$.

To fix this, we will construct a new polynomial $q$ that behaves like $\tilde{q}$ for $x \in \P$ and is \emph{extremely} small when $x \notin \P$
(in particular $0 \leq q(x) \ll 1/|p(x)|$ for such $x$). 
To understand how small we need $q$ to be, we need to determine just how large $p(x)$ can be on inputs with $x \notin \P$. 

To understand the behavior of $p(x)$ outside of $\P$, we can either analyze the behavior of an explicit polynomial of our choice for the $\NOR$ function or we can appeal to a general result about the growth of polynomials that are bounded in a region (see \Cref{lem:orT}). In either case, we get that there exists an upper bound $M = \exp(O(\sqrt{T} \log n))$ such that for all $x \notin \P$, $|p(x)|\leq M$.

This leads us to design a new polynomial $q$, which has the same behavior as $\tilde{q}$ for all $x \in \P$, but is at most $1/(3M$) when $x \notin \P$. 
We can construct the polynomial $q$ from $\tilde{q}$ by applying standard error reduction to reduce the approximation error of $\tilde{q}$ to $\epsilon=1/(3M)$, which increases its the degree by a multiplicative factor of $O(\log(3M)) = O(\sqrt{T} \log n)$. Thus, $\deg(q) = O(\sqrt{n/T} \cdot \sqrt{T} \log n) = \tO(\sqrt{n})$.

In summary, we have now constructed polynomials $p$ and $q$ with the following behavior:
\begin{align}
p(x) \in \begin{cases} [9/10,1] &\mathrm{if}~ |x| = 0 \\
[0,1/10] & \mathrm{if}~ 1\leq |x| \leq T \\ 
[0,M] & \mathrm{if}~ x \notin \P \end{cases}
\qquad 
q(x) \in \begin{cases} [1-1/(3M),1] &\mathrm{if}~ |x| = 0 \\
[0,1] & \mathrm{if}~ 1\leq |x| \leq T \\ 
[0,1/(3M)] & \mathrm{if}~ x \notin \P \end{cases}
\end{align}

\begin{figure}
    \centering
    \begin{minipage}{0.48\textwidth}
        \includegraphics[clip,width=\textwidth]{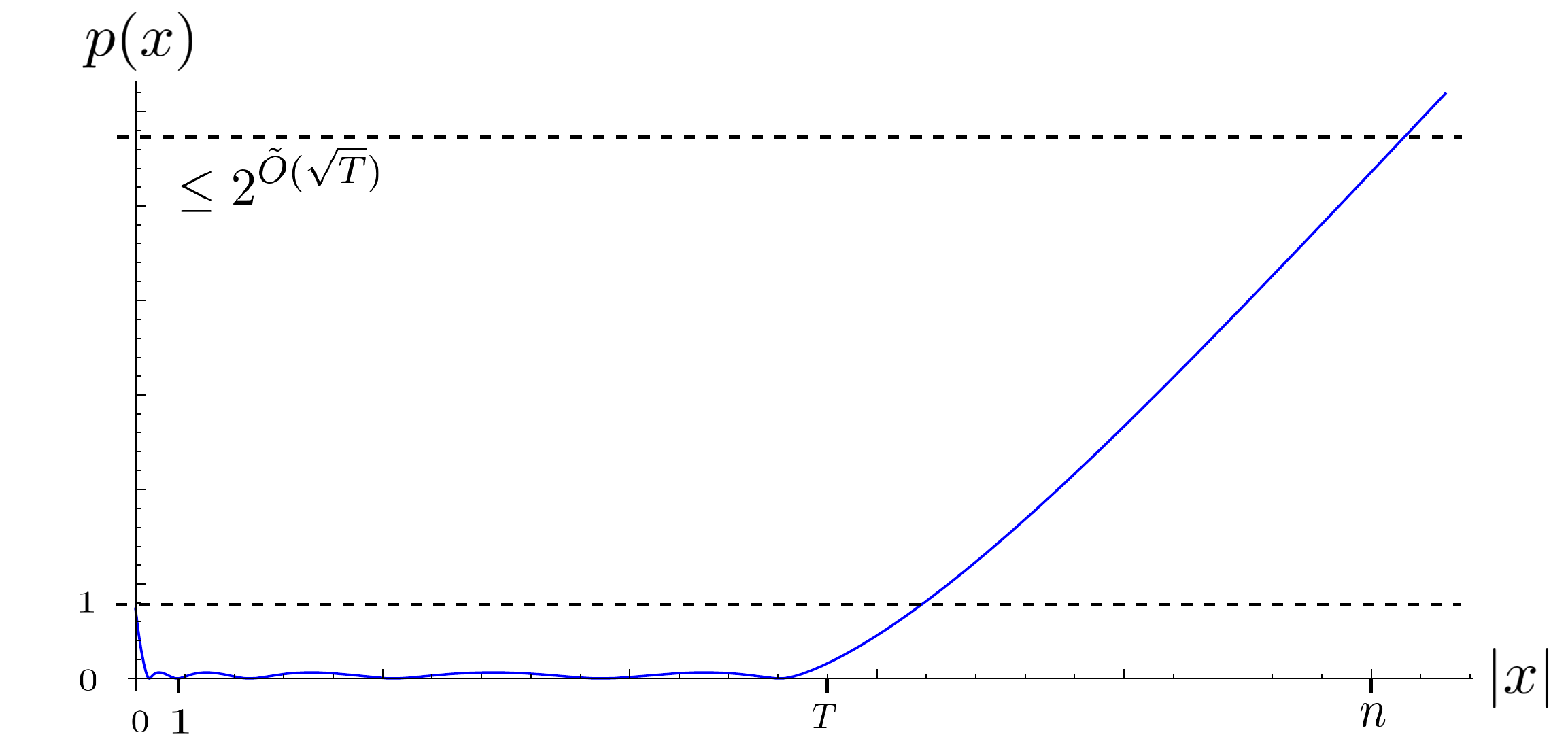} 
        \caption{Caricature of the polynomial $p(x)$\label{fig:p}}
    \end{minipage}\hfill
    \begin{minipage}{0.48\textwidth}
        \includegraphics[clip,width=\textwidth]{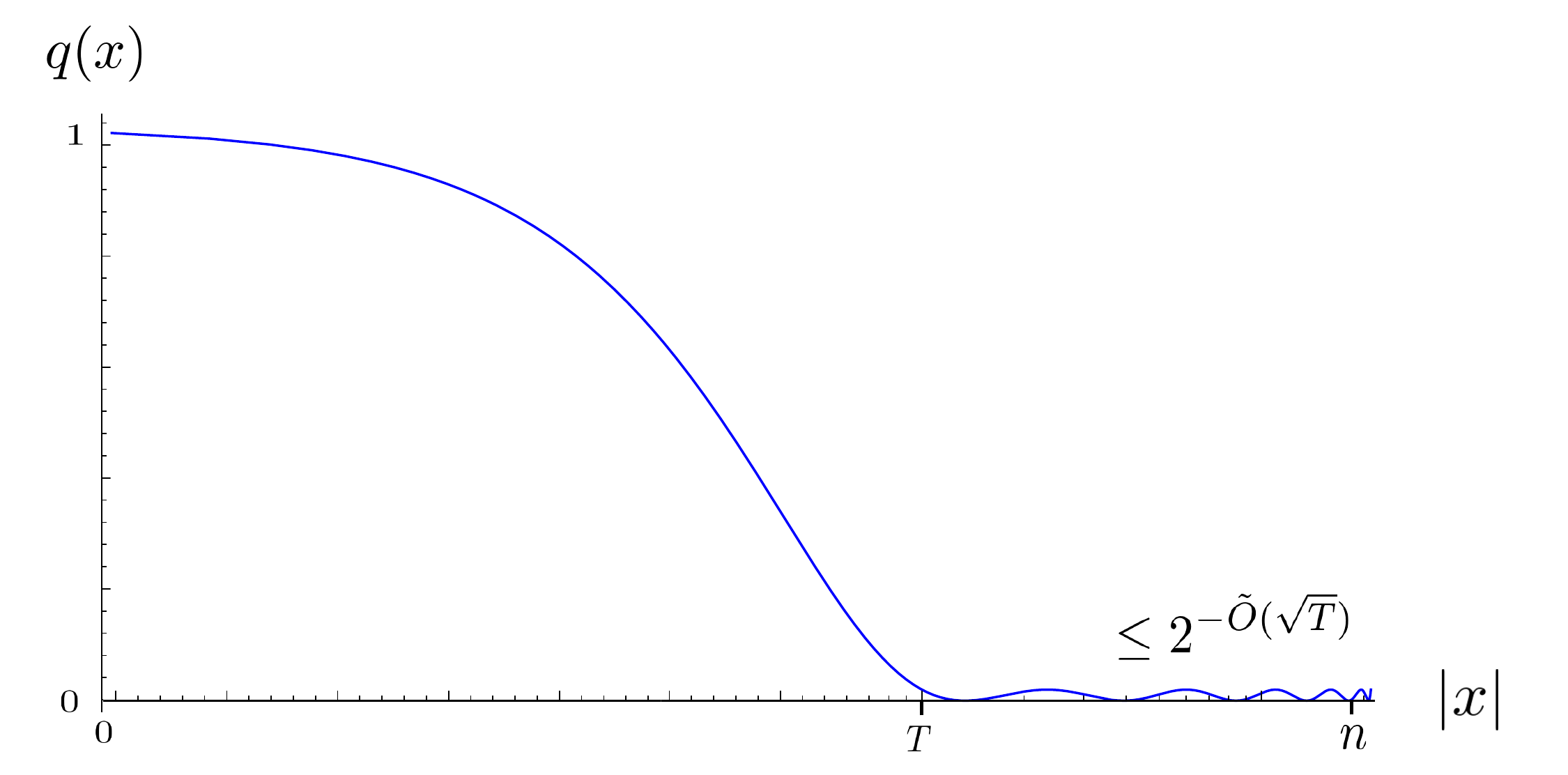} 
        \caption{Caricature of the polynomial $q(x)$\label{fig:q}}
    \end{minipage}
\end{figure}

Caricatures of these polynomials are depicted in \Cref{fig:p} and \Cref{fig:q}. It is now easy to see that the product $r(x)=p(x)\cdot q(x)$ is a (1/3)-error approximation to $\NOR$ for $\emph{all}$ $x \in \B^n$. 
The degree of the constructed polynomial is $\deg(r) \leq \deg(p) + \deg(q) = \tO(\sqrt{T}+\sqrt{n}) = \tO(\sqrt{n})$. 

Thus our constructed polynomial, $r$, has degree $\tO(\sqrt{n})$ and approximates $\NOR$ to error $1/3$, which is optimal by \Cref{andornor}.

\label{s:strategy}

\subsection{Informal Terminology: Polynomials as Algorithms}
\label{s:terminology}
Before moving on to \Surjectivity, we briefly introduce some terminology that will allow us convey the intuition of more involved constructions by reasoning about polynomials as if they were algorithms. 

Consider three Boolean functions $p_1:\B^n\to\B$, $p_2:\B^n\to\B$ and $p_3:\B^n\to\B$, and suppose there are deterministic algorithms $A_1$, $A_2$, and $A_3$ that compute these Boolean functions exactly. Now it makes sense to say ``run algorithm $A_1$ in input $x$; if it accepts then output $A_2(x)$, and if it rejects, then output $A_3(x)$.'' The Boolean function computed by this is $A_2(x)$ if $A_1(x)=1$ and $A_3(x)$ if $A_1(x)=0$. Observe that the following polynomial composition of $p_1, p_2, p_3$ computes the same Boolean function: $p_1(x)p_2(x) + (1-p_1(x))p_3(x)$.

We would like to use this terminology to discuss combining polynomials more generally. For polynomials $p$ that approximate Boolean functions by outputting a value in $[0,1]$ on any input $x$, we can imagine $p(x)$ as representing the probability that a randomized (or quantum) algorithm accepts on input $x$, and the same interpretation goes through. We will also use the same terminology for polynomials that output values greater than $1$, in which case we cannot interpret the output as a probability, but expressions like $p_1(x)p_2(x) + (1-p_1(x))p_3(x)$ still make sense.

For example, in the previous section we had two polynomials $p$ and $q$ and we constructed the polynomial $r(x)=p(x) \cdot q(x)$ by multiplying the two polynomials together. 
Another way to think of this is that $r$ is the polynomial obtained when we ``run'' the polynomial $q$ and output $p$ if it accepts and output $0$ if it rejects. 
This yields the polynomial $q(x)p(x) + (1-q(x)) \cdot 0 = q(x)p(x)$. 
We would like to informally describe polynomial constructions using this language, which will be especially useful when the constructions become more involved. So for example, we would informally describe the polynomial we constructed for the $\NOR$ function as follows (\Cref{alg:NOR}): 

\begin{algorithm}
\caption{An informal description of the polynomial approximation for $\NOR$}\label{alg:NOR}
\begin{algorithmic}[1]
\vspace{0.25em}
\State Using $q$, check if $|x|>T$ (with exponentially small probability of error if indeed $|x| > T$). If so, halt and output 0. 
\State Using $p$, compute $\NOR$ under the promise that $0\leq |x| \leq T$ and output the result.
\end{algorithmic}
\end{algorithm}

\subsection{Approximating \Surjectivity}
\label{s:formal}
We now construct a polynomial that approximates \Surjectivity\ using the strategy described above. Recall that $\SURJ : [R]^N \to \B$ is defined by $\SURJ(x_1, \dots, x_N) = 1$ if and only if for all $r \in [R]$ there exists an $i \in [N]$ such that $x_i = r$.

In this section, we will use the notation
\begin{equation}
    \fr_r(x) = |\{i \in [N]: x_i=r\}|
\end{equation} 
to denote the number of times the range element $r$ appears in the input $x$. Note that the \Surjectivity\ function evaluates to $1$ if and only if $\fr_r(x)\geq 1$ for all $r\in[R]$.

Finally, we will also need to consider a generalized version of \Surjectivity, $\Rcal$-\Surjectivity\ for some set $\Rcal \subseteq [R]$, which we denote $\SURJ_{\Rcal}$. 
As with \Surjectivity, 
\begin{equation}
    \SURJ_{\Rcal}:[R]^N \to \B,
\end{equation}
and $\SURJ_{\Rcal}(x_1, \ldots, x_N) = 1$ if and only if for all $r \in \Rcal$ there exists an $i\in [N]$ such that $x_i = r$. 
In other words, $\SURJ_\Rcal(x)=1$ if and only if for all $r\in\Rcal$ we have $\fr_r(x)\geq 1$.
Note that $\SURJ_{[R]} = \SURJ$. 
Our construction will actually show more generally that $\adeg(\SURJ_{\Rcal}) = \tO(N^{3/4})$ for all $\Rcal \subseteq [R]$.

\subsubsection{Approximating \Surjectivity\ on the Hard Inputs} \label{sec:surj-hard-inputs}

To implement the strategy described in \Cref{s:strategy} in the context of \Surjectivity, we first choose a set $\P$ of inputs that we consider to be ``hard''. 
Since \Surjectivity\ can be phrased as asking whether all range items appear at least once in the input, it is natural to consider the hard inputs to be those that have few occurrences of each range item.  Intuitively, this is because on such inputs, evidence that any given range item $r$ appears at least once in the input is hard to find. 

To this end, we define $\mathcal{P}$ as the set of inputs for which every range item appears at most $T$ times, for a parameter $T$ to be chosen later:
\begin{equation}
    \mathcal{P}=\{x:\forall r\in [R],~\fr_r(x)\leq T\}. \label{pdefsurj}
\end{equation}
More generally, when we consider $\Rcal$-\Surjectivity, we denote the set of hard inputs $\P_\Rcal$ and define it as
\begin{equation}
    \P_\Rcal=\{x:\forall r\in \Rcal,~\fr_r(x)\leq T\}.
\end{equation}

In this section, we will construct a polynomial $p_{\mathcal{R}}$ that approximates $\SURJ_\Rcal$ on $\P_\Rcal$ to bounded error, but may be exponentially large outside of $\P_{\Rcal}$. 
The value of $T$ in the definition of $\P_{\Rcal}$ will be chosen later;
for now we only assume that our choice will satisfy $T=N^{\Theta(1)}$, which simplifies some expressions since we have $\log T = \Theta(\log N)$.

\medskip \noindent \textbf{Overview of the Construction of $p_{\mathcal{R}}$.} To construct a polynomial that works on the hard inputs, $\P_\Rcal$, we first express $\SURJ_\Rcal$ as 
\begin{equation}
    \SURJ_\Rcal(x) = \bigwedge_{r\in\Rcal} \bigvee_{i\in[N]} \id[x_i=r], \label{eq:feck}
\end{equation}
where $\id[x_i=r]$ is the indicator function that takes value $1$ when $x_i=r$ and $0$ otherwise. 
Observe that for any fixed $r \in \Rcal$, the function $\id[x_i=r]$ depends on only $\log R$ bits
of $x$ and hence can be exactly computed by a polynomial of degree at most $\log R$. 

Since our goal in this section is to construct a polynomial $p_{\mathcal{R}}$ that approximates $\SURJ_\Rcal$ on all inputs in $\P_{\Rcal}$ and may be exponentially large outside of $\P_{\Rcal}$, we can assume that each inner $\OR$ gate 
in Equation \eqref{eq:feck} is fed an input of Hamming weight at most $T$. 
Hence, our approach will be to first construct a low-degree polynomial $q$ that approximates $\AND_{|\Rcal|} \circ \OR_{N}$ for inputs in $\left(H^N_{\le T}\right)^{|\mathcal{R}|}$.
We then obtain the polynomial $p_{\mathcal{R}}$ that approximates $\SURJ_{\Rcal}$ at all inputs in $\P_{\Rcal}$
by composing $q$ with the indicator functions $\{\id[x_i=r]\}_{i \in [N], r \in \mathcal{R}}$. Notice that the degree of $p_{\mathcal{R}}$ is at most $\deg(q) \cdot \log R$.

To construct $q$, our approach is as follows. First, we 
construct a polynomial $V_T$ of degree $O(T^{1/2} \log N)$ that approximates $\OR_N$ to error $O(1/N)$ at all inputs of Hamming weight at most $T$.
(However, $V_T(x)$ may be as large as $\exp(T^{1/2} \log N)$ for inputs $x$ of larger Hamming weight). 
Invoking \Cref{andornor}, we let $w$ be a polynomial of degree $\Theta(N^{1/2})$ that approximates $\AND_{|\Rcal|}$ to error $1/20$,
and finally we define $q := w \circ V_T$. 
A simple and elegant analysis of Buhrman et al. \cite{BuhrmanNRW07} (cf. \Cref{lem:naive}) allows
us to argue that $q$ indeed approximates $\AND_{|\Rcal|} \circ \OR_{N}$
on $\left(H^N_{\le T}\right)^{|\mathcal{R}|}$.

\medskip \noindent \textbf{Preliminary Lemmas.}
Before formally defining and analyzing $p_{\mathcal{R}}$, we record a few important facts about Chebyshev polynomials that will be useful throughout the remainder of this section.

\begin{lemma}[Properties of Chebyshev Polynomials]\label{lem:cheb}
Let $d \in \mathbb{N}$.
\begin{enumerate}[(1)]
\item There exists a polynomial $T_d : \R \to \R$ (the Chebyshev polynomial of degree $d$) such that $T_d(x) \in [-1, 1]$ for all $x \in [-1, 1]$ and $T_d(1 + \mu) \ge \frac{1}{2} \exp(d\sqrt{\mu})$ for all $\mu \in (0, 1)$.
\item For any polynomial $p : \R \to \R$ of degree $d$ with $|p(x)| \le 1$ for all $x \in [-1, 1]$, we have that for any $x$ with $|x| > 1$,
\begin{equation}
    |p(x)| \leq |T_d(x)| \leq (2|x|)^d, \label{eq:swell}
\end{equation}
where $T_d$ is the Chebyshev polynomial of degree $d$.
\end{enumerate}
\end{lemma}

\begin{proof}
To establish property (1), we use the fact that for $\mu > 0$, the value of the degree $d$ Chebyshev polynomial can be written~\cite[\S 3.2 Problem 8(f)]{cheney} 
 as
\begin{align*}
T_d(1 + \mu) &= \cosh ( d \operatorname{arcosh} ( 1 + \mu) ) \\
& \ge \cosh ( d \sqrt{\mu} ) & \text{ for } \mu \le 1\\
& \ge \frac{1}{2} \cdot \exp(d \sqrt{\mu}).
\end{align*}
Property (2) appears as~\cite[3.2 Problem 19]{cheney}.
\end{proof}

We are now ready to establish
the existence of a low-degree polynomial that approximates $\OR$ at all inputs of low Hamming weight.

\begin{lemma}[Approximating $\OR$ on Inputs of Low Hamming Weight]\label{lem:orT}
Let $\eps \in (0, 1)$ and $1 \le T \le N$. There is a polynomial $V_{T, \eps}:\B^N\to \R$ of degree $O(\sqrt{T}\log (1/\eps))$ such that
\begin{align}
V_{T, \eps}(x) \in \begin{cases} [0,\eps] &\mathrm{if}~ |x| = 0 \\
[1-\eps,1] & \mathrm{if}~ 1\leq |x| \leq T \\ 
[-a, a] \text{ for some } a \in {\exp\left(O\left(\sqrt{T} \cdot \log N \cdot \log\left(1/\eps\right)\right)\right)} & \mathrm{if}~ |x| > T \end{cases}.
\end{align}
\end{lemma}
\begin{proof}
Choose $d = O(\sqrt{T} \log(1/\eps))$ so that $M := T_d(1 + 1/T) + 1\ge 2/\eps$ (as guaranteed by Property 1 of \Cref{lem:cheb}). Define $V_{T, \eps}$ by the following affine transformation of $T_d$:
\[V_{T, \eps}(x) = \left(1  - \frac{1}{M}\right) - \frac{1}{M} \cdot T_d\left(1 + \frac{1}{T} - \frac{|x|}{T}\right).\]
 Then for $|x| = 0$, we have 
\[V_{T, \eps}(x) =  \left(1  - \frac{1}{M}\right) - \frac{1}{M} \cdot T_d\left(1 + \frac{1}{T}\right) = 0.\]
If $1 \le |x| \le T$, then $1 + 1/T - |x| / T \in [-1, 1]$, so $V_{T, \eps(x)} \in [1-\eps, 1]$. Finally, if $T +1 \le |x| \le N$, then
\[|V_{T, \eps}(x)| \le 1 + \frac{1}{M} \cdot T_d(N / T) \le \exp(O(\sqrt{T} \cdot\log N \cdot\log(1/\eps)),\]
where the final inequality holds by Equation \eqref{eq:swell}.
\end{proof}

The following lemma shows that if $p$ and $q$ are approximating polynomials for Boolean functions $f$ and $g$, respectively, then the block composition $p \circ q$ approximates $f \circ g$, with a blowup in error that is proportional to the number of variables over which $f$ is defined. The proof is due to Buhrman et al.~\cite{BuhrmanNRW07}, but our formulation is slightly different so we give the proof for completeness.

\begin{lemma} \label{lem:naive}
Let $f : \B^n \to \B$ and $g : X \to \B$ be Boolean functions, where $X \subseteq \{0, 1\}^m$ for some $m$. Let $p : \B^n \to [0, 1]$ be an $\eps$-approximating polynomial for $f$, and let $q : X \to [0, 1]$ be a $\delta$-approximating polynomial for $g$. Then the block composition $p \circ q : X^n \to \R$ is an $(\eps + \delta n)$-approximating polynomial for $f \circ g : X^n \to \B$.
\end{lemma}

\begin{proof}
Fix any input $x = (x_1, \dots, x_n) \in X^n$, and let $y = (g(x_1), \dots, g(x_n)) \in \B^n$. Let $z \in [0, 1]^n$ be defined by $z = (q(x_1), \dots, q(x_n))$. Since $p$ is an $\eps$-approximating polynomial for $f$, by the triangle inequality, it suffices to show that $|p(y) - p(z)| \le \delta n$.

Let $Z$ be a random variable on $\B^n$ where each $Z_i = 1$ independently with probability $z_i$. Then due to the multilinearity of $p$, we have
\[p(z) = \E[p(Z)] = \Pr[Z =y] \cdot p(y) + \Pr[Z \ne y] \cdot \E[p(Z) | Z \ne y].\]
Since $q$ is a $\delta$-approximating polynomial for $g$, we have $|y_i - z_i| \le \delta$ for every $i \in [n]$.  Hence,
\[\Pr[Z = y] \ge \left(1 - \delta \right)^n \ge 1 - \delta n.\]
Because $p$ is bounded in $[0, 1]$, this implies
\[p(z) \ge (1 - \delta n) \cdot p(y) + 0 \ge p(y) - \delta n\]
and
\[p(z) \le 1 \cdot p(y) + \delta n \cdot 1 = p(y) + \delta n,\]
completing the proof.
\end{proof}

\medskip \noindent \textbf{Formal Definition of $p_{\mathcal{R}}$.} Let $w$ be a $(1/20)$-approximating polynomial for $\AND_{|\mathcal{R}|}$ of degree $O(|\mathcal{R}|^{1/2})$
whose existence is guaranteed by \Cref{andornor}. We may assume that $w(x) \in [0, 1]$ for all $x \in \B^n$ (we will exploit this assumption in the proof of \Cref{lem:promiseSURJ} below, as it allows
us to apply \Cref{lem:real-bound} below to $w$). Let $q := w \circ V_{T, 1/(20n)}$.
Finally, let us define $p_{\mathcal{R}}$ to be the composition of $q$ with the indicator functions $\{\id[x_i=r]\}_{i \in [N], r \in \mathcal{R}}$.
For example, if $\mathcal{R}=\{1, \dots, |\mathcal{R}|\}$,
then \begin{equation*} p_{\mathcal{R}} = q\left(\id\left[x_1 = 1\right], \id\left[x_2=1\right], \dots,
\id\left[x_N=1\right], \id\left[x_1=2\right], \dots, \id\left[x_N=\left|\mathcal{R}\right|\right]\right).
\end{equation*}

Observe that $$\deg(p_{\mathcal{R}}) \leq \deg(w) \cdot \deg(V_{T, 1/(20n)}) \cdot \deg(\id[x_i=r]) \leq O\left(|\mathcal{R}|^{1/2} \cdot T^{1/2} \log n \cdot \log R\right) \leq \tilde{O}\left(\sqrt{NT}\right).$$

\medskip \noindent \textbf{Showing $p_{\mathcal{R}}$ Approximates $\SURJ_{\mathcal{R}}$ on $P_{\mathcal{R}}$.}
\Cref{lem:naive} implies that:
\begin{equation} |q(x) - \AND_{|\mathcal{R}|} \circ \OR(x)| \leq 1/10 \text{ for all } x \in \left(H^N_{\le T}\right)^{|\mathcal{R}|}. \label{imdoneguys} \end{equation} 
An immediate consequence is the following lemma.

\begin{lemma} \label{approxlemma} $|p_{\mathcal{R}}(x) - \SURJ_{\Rcal}(x)| \leq 1/10$
for all $x \in \P_\Rcal$. \end{lemma}

\medskip \noindent \textbf{Bounding $p_{\mathcal{R}}$ Outside of $P_{\mathcal{R}}$.}
For an input $x\in [R]^N$ outside of $\P_\Rcal$, let $b_\Rcal(x)$ be the number of range items that appear more than $T$ times, i.e., 
\begin{equation}\label{eq:b}
    b_\Rcal(x) = |\{r\in\Rcal:\fr_r(x)>T\}|.
\end{equation}
We claim that $|p_{\mathcal{R}}(x) | \leq {\exp\left(b_\Rcal(x) \cdot \tO(\sqrt{T})\right)}$. This bound relies on the following elementary lemma.

\begin{lemma} \label{lem:real-bound}
Let $p : \R^n \to \R$ be a multilinear polynomial with $p(x) \in [0, 1]$ for all $x \in \B^n$. Then for $x \in \R^n$, we have
\[|p(x)| \le \prod_{i = 1}^n (|1 - x_i| + |x_i|).\]
\end{lemma}

\begin{proof}
We prove the lemma by induction on the number of variables $n$. If $n = 0$, then $p$ is a constant in the interval $[0, 1]$ and the claim is true.

Now suppose the claim is true for $n-1$ variables, and let $p : \R^n \to \R$ with $p(x) \in [0, 1]$ for $x \in \B^n$. We begin by decomposing
\[p(x) = (1-x_n) \cdot q_0(x_1, \dots, x_{n-1}) + x_n \cdot q_1(x_1, \dots, x_{n-1})\]
where $q_0$ and $q_1$ are themselves multilinear polynomials. Since $p(x) \in [0, 1]$ for all $x \in \B^n$, this is in particular true when $x_n = 0$. Hence $q_0(x') \in [0, 1]$ for all $x' \in \B^{n-1}$. Similarly, setting $x_n = 1$ reveals that $q_1(x') \in [0, 1]$ for all $x' \in \B^{n-1}$. Now for any $x \in \R^n$ with $x' = (x_1, \dots, x_{n-1})$ we have
\begin{align*}
|p(x)| &= \left| (1 - x_n) \cdot q_0(x') + x_n \cdot q_1(x')\right| \\
&\le |1 - x_n| \cdot |q_0(x')| + |x_n| \cdot |q_1(x')| \\
&\le (|1 - x_n| + |x_n|) \cdot \prod_{i = 1}^{n-1} (|1 - x_i| + |x_i|)
\end{align*}
where the final inequality uses the inductive hypothesis. This proves the claim.
\end{proof}

\begin{lemma}\label{lem:promiseSURJ} \label{lem:boundonprcal}
There exists a function $a(x) =  \exp\left(b_\Rcal\left(x\right)\cdot \tO\left(\sqrt{T}\right)\right)$ such that for any $\Rcal \subseteq [R]$, the polynomial $p_\Rcal:[R]^N\to \R$ has degree $\tO(\sqrt{NT})$ and satisfies:
\begin{align*}
p_\Rcal(x) \in \begin{cases} [0,1/10] &\mathrm{if}~ x\in\P_\Rcal ~\mathrm{and}~ \SURJ_\Rcal(x)=0\\
[9/10,1] &\mathrm{if}~ x\in\P_\Rcal ~\mathrm{and}~ \SURJ_\Rcal(x)=1\\
[-a(x), a(x)]& \mathrm{if}~ x \notin \P_\Rcal.\end{cases}
\end{align*}
\end{lemma}

\begin{proof}
The first two cases are an immediate consequence of \Cref{approxlemma}.

To upper bound the value of $|p_\Rcal(x)|$ for $x\notin \P_\Rcal$, we exploit the structure of $p_\Rcal$ as a multilinear polynomial $w$ of degree $O(\sqrt{N})$ over the variables $z_1,\ldots,z_{|\Rcal|}$, where each $z_r$ is the output of the $r^\mathrm{th}$ polynomial from \Cref{lem:orT}. That is, 
$$z_r = V_{T, 1/(20n)}\left(\id[x_1 = r], \dots, \id[x_N = r]\right).$$ 
If $b_\Rcal(x)$ range items appear greater than $T$ times, this means that up to $b_\Rcal(x)$ of the variables $z_r$ might take values outside $[0,1]$. 
However, by \Cref{lem:orT}, each of these $b_\Rcal(x)$ variables is still at most ${\exp(\tO(\sqrt{T}))}$. By \Cref{lem:real-bound},
\[|w(z)| \le \prod_{r \in \Rcal} \left(|1 - z_r| + |z_r|\right) \le \exp\left(b_\Rcal(x) \cdot \tO\left(\sqrt{T}\right)\right),\]
since each $z_r $ that is in $[0, 1]$ contributes a factor of exactly $1$ to the product, whereas each of the remaining $b_\Rcal(x)$ variables contributes a factor of at most ${\exp(\tO(\sqrt{T}))}$ to the product.
\end{proof}

\subsubsection{Controlling The Easy Inputs} \label{sec:surj-easy-inputs}

\noindent \textbf{Intuition.} Unlike the example of the $\NOR$ function, where all the inputs outside the ``hard'' set $\P$ (cf. Equation \eqref{pdefnor}) were in $\NOR^{-1}(1)$, for \Surjectivity\ there are both $0$- and $1$-inputs outside of the hard set $\P$  (cf. Equation \eqref{pdefsurj}) . 
So our remaining task is not simply a matter of constructing a polynomial that detects 
if the input is outside of $\P$, as it was in the case of the $\NOR$ function.

However, we will show that, for \Surjectivity, the inputs outside of $\P$ are easy to handle in a different sense---they are easy because we can construct a reduction from inputs outside of $\P$ to inputs in $\P$. 
To gain some intuition, consider the task of designing a randomized algorithm for \Surjectivity\ where we want to reduce the general case to the case where there is a set $\Rcal \subseteq [R]$ such that all range items $r\in \Rcal$ have $\fr_r(x)\leq T$.
To do this, we can simply sample a large number of elements $x_i$ and remove from consideration all  range items appearing at least once in the sample, because we know these range items all appear in the input at least once.
After this step, we have a new set $\Rcal\subseteq[R]$ consisting of all range items that have not been seen in the sampling stage.
Every $r\in\Rcal$ is likely to have $\fr_r(x)\leq T$ because elements that appeared too frequently would have (with high probability) been observed in the sampling stage.
Thus it now suffices to solve $\SURJ_{\Rcal}$ on the input under the assumption that all range items $r \in \Rcal$ have $\fr_r(x)\leq T$.

Informally, the above discussion states that we want to construct the polynomial described in \Cref{alg:SURJalmost}.

\begin{algorithm}[H]
\caption{Informal description of the polynomial approximation for $\SURJ$}\label{alg:SURJalmost}
\begin{algorithmic}[1]
\vspace{0.25em}
\State Sample $S=\tTheta(N^{3/4})$ items and remove all range items seen from $[R]$. Let the remaining set be $\Rcal \subseteq [R]$.
\State Solve $\SURJ_{\Rcal}$ under the promise that all $r\in\Rcal$ have $\fr_r(x)\leq T$, where $T=\tTheta(\sqrt{N})$.
\end{algorithmic}
\end{algorithm}

We now have to construct a polynomial that represents this algorithmic idea. 
We have already constructed an (unbounded approximating) polynomial for $\SURJ_\Rcal$ under the promise $\P_\Rcal$ in the previous section, so the second step of this construction is done.

For Step 1 of \Cref{alg:SURJalmost}, we need to construct a polynomial to represent the idea of sampling input elements and evaluating a polynomial that \emph{depends on} the sampled elements. 
To build up to this, consider a deterministic algorithm that queries a subset $\Scal \subseteq [N]$ of input elements, checks if the sampled string equals another fixed string $y$ and outputs $1$ if true and $0$ if false. If we denote the input $x\in[R]^N$ restricted to the subset $\Scal \subseteq [N]$ as $x_\Scal$, then this algorithm outputs $1$ if and only if $x_\Scal = y$. Interpreting the input as an element of $\B^{N \log R}$
rather than $[R]^N$, it is easy to see that a deterministic query algorithm querying $|\Scal|\log R$ bits of $x$ can solve this problem.
Consequently, there is a polynomial of degree $|\Scal|\log R$ that outputs $1$ if $x_\Scal = y$ and outputs $0$ otherwise.
For any fixed $\Scal \subseteq [N]$ and $y\in [R]^{|\Scal|}$, we denote this polynomial by $\id_y(x_\Scal)$.

Now we can construct a polynomial which samples $S$ elements of the input and then outputs a bit depending on the elements seen. Let $\Scal \subseteq [N]$ be a subset of indices with $|\Scal|=S$. 
Let $f(\Scal,x_\Scal)$ be an arbitrary Boolean function that tells us whether to output $0$ or $1$ on seeing the sample $(\Scal,x_\Scal)$.
Then the following polynomial samples a random set $\Scal$ of size $S$, reads the input $x$ restricted to the set $\Scal$, and outputs the bit $f(\Scal,x_\Scal)$:
\begin{equation}
    \frac{1}{\binom{N}{S}} \sum_{\Scal \in \binom{[N]}{S}} \sum_{y \in [R]^S} \id_y(x_\Scal) f(\Scal,y)
\end{equation}
Note that for any function $f$, this is a polynomial of degree $S\log R$  in the variables $x_1,\ldots,x_N$, because $\id_y(x_\Scal)$ is a multilinear polynomial over those variables, and $f(\Scal,y)$ is simply a hard-coded bit that does not depend on $x$. The value of this polynomial on an input $x$ equals
\begin{equation}
    \frac{1}{\binom{N}{S}} \sum_{\Scal \in \binom{[N]}{S}} \sum_{y \in [R]^S} \id_y(x_\Scal) f(\Scal,y)  
    = \frac{1}{\binom{N}{S}} \sum_{\Scal \in \binom{[N]}{S}} f(\Scal,x_\Scal).
\end{equation}

We can now generalize this construction to allow for the possibility that the function $f$ is itself a polynomial. 
This is what we need to implement \Cref{alg:SURJalmost}, in which we sample a random set $\Scal \subseteq [N]$ of size $S$, query $x_{\mathcal{S}}$ (using $S \log R$ queries), and then run a polynomial that depends on the results.
Specifically, if we sample the set $\Scal \subseteq [N]$, and learn that $x_{\mathcal{S}}$ equals the string $y\in[R]^{|\Scal|}$, then we want to run the polynomial $p_\Rcal(y)$ from \Cref{lem:promiseSURJ} for the set $\Rcal(y)$ of all elements in $[R]$ that do not appear in $y$, i.e., 
\begin{equation} \label{nomoreeqs}
    \Rcal(y) = [R] \setminus \{r:\exists i~y_i = r\}.
\end{equation}

\medskip
\noindent \textbf{Formal Description of The Approximation To \Surjectivity.} Using the tools from \Cref{sec:surj-hard-inputs} and the above discussion, we can now construct a polynomial that corresponds to the informal description in \Cref{alg:SURJalmost}:
\begin{equation}
    r(x) = \frac{1}{\binom{N}{S}} \sum_{\Scal \in \binom{[N]}{S}} \sum_{y \in [R]^S} \id_y(x_\Scal) p_{\Rcal(y)}(x). 
\end{equation}
This is a polynomial of degree $S\log R + \max_{y}\{\deg(p_{\Rcal(y)})\} = \tO(S + \sqrt{NT}) = \tO(N^{3/4})$, using $S=\tTheta(N^{3/4})$ and $T=\tTheta(\sqrt{N})$ (where the factors hidden by the $\tTheta$ notation will be chosen later). The value of the polynomial on input $x$ is
\begin{equation} \label{fleck}
    r(x) = \frac{1}{\binom{N}{S}} \sum_{\Scal \in \binom{[N]}{S}} p_{\Rcal(x_\Scal)}(x). 
\end{equation}
where recall that $\Rcal(x_{\Scal})$ is as defined in Equation \eqref{nomoreeqs}.
The right hand side of Equation \eqref{fleck} is precisely the expected value of the polynomial $p_{\Rcal(x_\Scal)}(x)$ when $\Scal$ is a uniformly random set of size $S$. We will now show that $r(x)$ is an approximating polynomial for \Surjectivity, i.e., that for all $x\in[R]^N$, $|\SURJ(x)-r(x)|\leq 1/3$.

Recalling that $b_{\Rcal(x_\Scal)}$ is the number of range items $r \in \Rcal$ that appear in $x$ greater than $T$ times (cf. Equation \eqref{eq:b}),
we compute the value of the polynomial $r(x)$ on an input $x$:
\begin{align}
    r(x) &= \frac{1}{\binom{N}{S}} \sum_{\Scal \in \binom{[N]}{S}} p_{\Rcal(x_\Scal)}(x) \\
    &= \frac{1}{\binom{N}{S}} \left(\sum_{\Scal \in \binom{[N]}{S}} \sum_{b=0}^{N/T} \id[b_{\Rcal(x_\Scal)}(x)=b] \cdot p_{\Rcal(x_\Scal)}(x) \right)\\
    &= \frac{1}{\binom{N}{S}} \sum_{\Scal \in \binom{[N]}{S}} \id[b_{\Rcal(x_\Scal)}(x)=0] \cdot p_{\Rcal(x_\Scal)}(x) \label{eq:t1}\\
    &\qquad\qquad\qquad + \sum_{b=1}^{N/T} \frac{1}{\binom{N}{S}} \sum_{\Scal \in \binom{[N]}{S}}  \id[b_{\Rcal(x_\Scal)}(x)=b]  \cdot p_{\Rcal(x_\Scal)}(x),
    \label{eq:t2}
\end{align}
where we split up the sum into the $b=0$ \eqref{eq:t1} and $b\geq 1$ \eqref{eq:t2}  cases.
We first show that the $b\geq1$ term \eqref{eq:t2} has essentially no contribution to the final value. 

When $b_{\Rcal(x_\Scal)}(x) \geq 1$, we know by \Cref{lem:boundonprcal} that the magnitude of the polynomial $|p_\Rcal(x)|$ is at most $\exp\left(b_{\Rcal(x_\Scal)}(x) \cdot \tO\left(\sqrt{T}\right)\right)$.
Thus the magnitude of the term in Expression \eqref{eq:t2} is at most
\begin{align}
    &\sum_{b=1}^{N/T} \frac{1}{\binom{N}{S}} \sum_{\Scal \in \binom{[N]}{S}}  \id[b_{\Rcal(x_\Scal)}(x)=b] \cdot \exp(b \cdot \tO(\sqrt{T})) \\
    =&\sum_{b=1}^{N/T} \Pr_{\Scal}[b_{\Rcal(x_\Scal)}(x)=b] \cdot \exp\left(b \cdot \tO\left(\sqrt{T}\right)\right).
\end{align}

We now need to compute the value of $\Pr_{\Scal}[b_{\Rcal(x_\Scal)}(x)=b]$ for each $b$. Intuitively, this roughly corresponds to the probability that we sample $S=\tTheta(N^{3/4})$ elements from the input and miss all $bT$ elements that correspond to the $T = \tTheta(\sqrt{N})$ copies of the $b$ range items that appear at least $T$ times. 
If we were to sample $N / (bT)$ elements, the probability of seeing none of the $b T$ elements would be $\Theta(1)$. 
Since we are sampling $S=\tTheta(N^{3/4}) \gg N / (bT)$ elements, the probability of not seeing one of the $bT$ elements is ${\exp(-bST/N)} = {\exp(-b \cdot \tOmega(N^{1/4}))}$. We make this heuristic calculation formal in the following lemma:

\begin{lemma} \label{lem:prob-b}
Let $S, T \le N$. Let $\Scal \subset [N]$ be a random subset of size $S$. Then for every $x \in [R]^N$ and every $b \ge 1$,
\[\Pr_{\Scal}[b_{\Rcal(x_\Scal)}(x) \ge b] \leq {\exp(-b \cdot (ST/N - \log N))},\]
recalling that $b_{\Rcal(x_\Scal)}$ is the number of range items $r \in \Rcal$ that appear in $x$ more than $T$ times, but do not appear in $x_\Scal$.
\end{lemma}

\begin{proof}
To calculate the probability of interest, we begin by analyzing a simplified experiment. Suppose there are $b$ range items $r_1, \dots, r_b$ each appearing greater than $T$ times in $x$. We compute the probability that \emph{none} of the items $r_1, \dots, r_b$ appear in the sample $x_\Scal$. If $T \ge N - S$, then this probability is zero. Otherwise, by direct calculation, we have for each $i = 1, \dots, b$:
\[\Pr_{\Scal}[r_i \notin x_\Scal | r_1 \notin x_\Scal, \dots, r_{i-1} \notin x_\Scal] \le \prod_{j = 0}^{S-1} \left(1 - \frac{T}{N-j} \right) \le \left( 1 - \frac{T}{N}\right)^S \le \exp \left(-\frac{ST}{N} \right).\]
Hence the probability that none of $r_1, \dots, r_b$ appear is
\[\Pr_{\Scal}[r_1 \notin x_\Scal \land \dots \land r_b \notin x_\Scal] = \prod_{i  = 1}^b \Pr_{\Scal}[r_i \notin x_\Scal | r_1 \notin x_\Scal, \dots, r_{i-1} \notin x_\Scal] \le \exp \left(-\frac{bST}{N} \right).\]

Now fix an arbitrary $x$, and let $r_1, \dots, r_{k}$ be the range items appearing greater than $T$ times in $x$. Observe that $k \le N/T$. We estimate the probability of interest by taking a union bound over all subsets of $r_1, \dots, r_k$ of size $b$:
\begin{align*}
\Pr_{\Scal}[b_{\Rcal(x_\Scal)}(x) \ge b] &\le \sum_{Y \subseteq [k] : |Y| = b} \Pr_{\Scal}[r_{i_j} \notin x_\Scal \quad \forall j \in Y] \\
& \le \binom{k}{b} \exp \left(-\frac{bST}{N}\right) \\
&\le \exp \left( -\frac{bST}{N} + b \log N\right).
\end{align*}
\end{proof}

From \Cref{lem:prob-b} we know that $\Pr_{\Scal}[b_{\Rcal(x_\Scal)}(x)=b] \leq {\exp(-b \cdot \tOmega(ST/N))}$.
Hence the $b\geq1$ term \eqref{eq:t2} is at most
\begin{align}
    \sum_{b=1}^{N/T} {\exp(-b \cdot \tOmega(ST/N))} \cdot {\exp(b \cdot \tO(\sqrt{T}))} = o(1),
\end{align}
by choosing $T = \tTheta(\sqrt{N})$ and $S = \tTheta(N^{3/4})$ appropriately. This shows that the term in \eqref{eq:t2} does not significantly influence the value of the polynomial for any input $x$.

Thus we have for any input $x$,
\begin{align}
    r(x) = \frac{1}{\binom{N}{S}} \sum_{\Scal \in \binom{[N]}{S}} \id[b_{\Rcal(x_\Scal)}(x)=0]  \cdot p_{\Rcal(x_\Scal)}(x) + o(1).
\end{align}

By \Cref{lem:promiseSURJ}, we know that if $b_{\Rcal(x_\Scal)}(x)=0$ (and hence $x \in \P_{\Rcal(x_\Scal)}$), then $p_{\Rcal(x_\Scal)}(x) $ is a $(1/10)$-approximation to $\SURJ(x)$. Applying \Cref{lem:prob-b} once again shows that $\Pr_{\Scal}[b_{\Rcal(x_\Scal)}(x) \ge 1] \le \exp(- \tOmega(N^{1/4}))$, so $\Pr_{\Scal}[b_{\Rcal(x_\Scal)}(x) = 0] = 1 - o(1)$. 
Hence, for all $x \in [R]^N$, $|r(x) - \SURJ(x)| \le 1/10 + o(1) \le 1/3$. 

This completes
the proof of \Cref{thm:upper}.


\section{Lower Bound for Surjectivity}
\label{s:surjlower}
The goal of this section is to show the following improved lower bound on the approximate degree of the Surjectivity function.

\begin{theorem} \label{thm:surj}
For some $N = O(R)$, the $(1/3)$-approximate degree of $\SURJ_{N, R}$ is $\tOmega(R^{3/4})$.
\end{theorem}

To prove \Cref{thm:surj}, we combine the following 
theorem with the reductions of \Cref{prop:dumsurj-reduction} and \Cref{cor:surj-reduction}.

\begin{theorem}   \label{thm:main-amp}
Let $N = c \cdot R$ for a sufficiently large constant $c>0$. Let $F^{\le N} \colon H^{N \cdot R}_{\le N} \to \bits$ equal $\AND_{R} \circ \OR_N$ restricted to inputs in $H^{N \cdot R}_{\le N} = \{x \in \bits^{N \cdot R} : |x| \le N\}$. Then $\ubdeg(F^{\le N}) \ge \tOmega(R^{3/4}).$
\end{theorem}

The proof of \Cref{thm:main-amp} entails using dual witnesses for the high approximate degree of $\AND_R$ and $\OR_N$ to construct a dual witness for the higher approximate degree of $F^{\le N}$. As indicated in \Cref{s:overview}, the construction is essentially the same as in~\cite{adegsurj}, except that we observe that a dual witness for $\OR$ constructed and used in prior works satisfies an exponentially stronger decay condition than has been previously realized. 

The construction can be thought of as consisting of three steps:
\paragraph{Step 1.} We begin by constructing a dual witness $\psi$ for the fact that the unbounded approximate degree of the $\OR_N$ function is $\Omega\left(\sqrt{T}\right)$ even when promised that the input has Hamming weight at most $T = \Theta(\sqrt{R})$. The dual witness $\psi$ is a small variant of the one in~\cite{adegsurj}, but we give a more careful analysis of its tail decay. In particular, we make use of the fact that for all $t \ge 1$, the $\ell_1$ weight that $\psi$ places on the $t$'th layer of the Hamming cube is upper bounded by $\exp(-\Omega(t/\sqrt{T})) / t^2$.

\paragraph{Step 2.} We combine $\psi$ with a dual witness $\Phi$ for $\AND_R$ to obtain a preliminary dual witness $\Phi \ls \psi$ for $F = \AND_R \circ \OR_N$. The dual witness $\Phi \ls \psi$ shows that $F$ has approximate degree $\Omega(\sqrt{R} \cdot \sqrt{T}) = \Omega(R^{3/4})$. However, $\Phi \ls \psi$ places weight on inputs of Hamming weight larger than $N$, and hence does not give an \emph{unbounded} approximate degree lower bound for the promise variant $F^{\le N}$.

\paragraph{Step 3.} Using \Cref{prop:btzeroing} we zero out the mass that $\Phi \ls \psi$ places on inputs of  Hamming weight larger than $N$, while maintaining its pure high degree and correlation with $F^{\le N}$. This yields the final desired \emph{unbounded} approximate degree dual witness $\zeta$ for $F^{\le N}$, as per \Cref{prop:ub-duality}.

\subsection{Step 1: A Dual Witness for \texorpdfstring{$\OR_N$}{OR}}

We begin by constructing a univariate function which captures the properties we need of our inner dual witness for $\OR_N$. The construction slightly modifies the dual polynomial for $\OR_N$ given by \v{S}palek \cite{spalek}. We provide a careful 
analysis of its decay as a function of the input's Hamming weight.

\begin{proposition} \label{prop:or-sym-dual}
Let $T \in \N$ and $1/T \le \delta \le 1/2$. There exist constants $c_1, c_2 \in (0, 1) $ and a function $\omega : [T]_0 \to \R$ such that
\begin{align}
&\omega(0) - \sum_{t = 1}^T \omega(t) \ge 1- \delta  \label{eqn:or-sym-corr}\\
&\sum_{t = 0}^T |\omega(t)| = 1 \\
& \text{For all univariate polynomials } q \colon \R \to \R\text{, } \deg q < c_1\sqrt{\delta T} \implies \sum_{t=0}^T \omega(t) \cdot q(t) = 0 \label{eqn:or-sym-phd}\\
& |\omega(t)| \le \frac{170 \exp(-c_2 t \sqrt{\delta} /  \sqrt{T})}{\delta \cdot t^2} \qquad \forall t = 1, \dots, T.\label{eqn:or-sym-decay}\end{align}
\end{proposition}

\begin{proof}[Proof of \Cref{prop:or-sym-dual}]
By renormalizing, it suffices to construction a function $\omega : [T]_0 \to \R$ such that
\begin{align} &\omega(0) - \sum_{t = 1}^T \omega(t) \ge (1 - \delta)\|\omega\|_1  \label{eqn:or-sym-corr-ref}\\
& \text{For all univariate polynomials } q \colon \R \to \R\text{, } \deg q < c_1\sqrt{\delta T} \implies \sum_{t=0}^T \omega(t) \cdot q(t) = 0 \label{eqn:or-sym-phd-ref}\\
& |\omega(t)| \le \frac{170 \exp(-c_2 t \sqrt{\delta} /\sqrt{T})\|\omega\|_1}{\delta \cdot t^2} \qquad \forall t = 1, \dots, T.\label{eqn:or-sym-decay-ref}\end{align}
 Let $c = \lceil 8 / \delta \rceil$ below. We will freely use the fact that since $\delta \le 1/2$, we have $c \le c^2 / (c-1) \le 10 /\delta$. Let $m = \lfloor \sqrt{T/2c} \rfloor$ and define the set
\[S = \{1, c\} \cup \{2ci^2 : 0 \le i \le m\}.\]
Note that $|S| \ge  c_1 \sqrt{\delta T}$ for some absolute constant $c_1 > 0$. Define the function
\[\omega(t) = \frac{(-1)^{t+(T-m)}}{T!} \binom{T}{t} \prod_{r \in [T]_0 \setminus S} (t - r).\]
Property~\eqref{eqn:or-sym-phd-ref} follows from the following well-known combinatorial identity.
\begin{fact}[e.g., \cite{concrete} Equation (5.23) or \cite{osnewbounds}] \label{fact:combinatorial}
Let $T \in \N$, and let $p$ be a polynomial of degree less than $T$. Then
\[\sum_{t = 0}^T (-1)^t \binom{T}{t} p(t) = 0.\]
\end{fact}

Expanding out the binomial coefficient in the definition of $\omega$ reveals that
\[|\omega(t)| = \begin{cases}
 \prod\limits_{r \in S \setminus \{t\}} \frac{1}{|t - r|} & \text{ for } t \in S, \\
0 & \text{ otherwise.}
\end{cases}\]

We now use this characterization to establish the improved tail decay property \eqref{eqn:or-sym-decay-ref}. This clearly holds for $t = 1$ with $c_2 = 1/10$, since $|\omega(1)| \le \|\omega\|_1$ and
\[\frac{170}{\delta} \exp(-c_2\sqrt{\delta} / \sqrt{T}) \ge 340 \cdot e^{-\sqrt{2} / 10} > 1.\]

For $t = c$, we have
\begin{align} \label{eqn:omega-2}
\frac{|\omega(c)|}{\omega(0)} &= \frac{c\prod_{i = 1}^m (2ci^2)}{c(c-1)\prod_{i = 1}^m (2ci^2 - c)} \nonumber \\
&= \frac{1}{c-1}\left(\prod_{i = 1}^m \frac{i^2 - 1/2}{i^2}\right)^{-1} \nonumber \\
&\le \frac{1}{c-1} \left(1 - \sum_{i=1}^m \frac{1}{2i^2}\right)^{-1} \nonumber \\
&\le \frac{1}{c-1} \left(1 - \frac{\pi^2}{12}\right)^{-1} \le \frac{6}{c - 1}
\end{align}
where the first inequality follows from the fact that $\prod_{i = 1}^m (1 - a_i) \ge 1 - \sum_{i = 1}^m a_i$ for $a_i \in (0, 1)$. Now note that
\begin{align*}
\frac{|\omega(c)|}{\|\omega\|_1} &\le \frac{|\omega(c)|}{\omega(0)} \\
&\le \frac{6}{c - 1} \\
&\le \frac{60}{\delta \cdot c^2} & \text{ since } 10/\delta \ge c^2 / (c-1) \\
& \le \frac{170}{\delta \cdot c^2} \cdot e^{-c\sqrt{\delta}/10\sqrt{T}},
\end{align*}
since $\delta \ge 1/T$ and hence $e^{-c\sqrt{\delta}/10\sqrt{T}} \ge e^{-1}$. Thus~\eqref{eqn:or-sym-decay-ref} holds for $t = c$, recalling that $c_2 = 1/10$.

For $t = 2cj^2$ with $j \ge 1$, we get
\begin{align*}
\frac{|\omega(t)|}{\omega(0)} &= \frac{c\prod_{i = 1}^m (2ci^2)}{(2cj^2 - 1)(2cj^2-c)\prod_{i \in [m]_0 \setminus \{j\}} |2ci^2 - 2cj^2|} \\
&= \frac{c(m!)^2}{(4c^2j^4 - (2c^2 + 2c)j^2 + c)\prod_{i \in [m]_0 \setminus \{j\}} (i+j)|i-j|} \\
&= \frac{c}{4c^2j^4 - (2c^2 + 2c)j^2 + c} \cdot \frac{(m!)^2}{(m+j)!(m-j)!}.
\end{align*}
For $j \ge 1$, the first factor is bounded by
\[\frac{c}{4c^2j^4 - (2c^2 + 2c)j^2 + c} \le \frac{3c}{(2cj^2)^2},\]
using the fact that $c \ge 2$. We control the second factor by
\begin{align*}
\frac{(m!)^2}{(m+j)!(m-j)!} &= \frac{m}{m+j} \cdot \frac{m-1}{m+j-1} \cdot \ldots \cdot \frac{m-j+1}{m+1} \\
&\le \left(\frac{m}{m+j} \right)^j \\
&\le  \left(1 - \frac{j}{2m} \right)^j \\
&\le e^{-j^2/2m},
\end{align*}
where the last inequality uses the fact that $1 - x \le e^{-x}$ for all $x$.
Since
\[\frac{|\omega(t)|}{\|\omega\|_1} \le \frac{|\omega(2cj^2)|}{\omega(0)} \le \frac{3c}{(2cj^2)^2} \cdot e^{-2cj^2/(4cm)} \le \frac{170}{\delta \cdot t^2} \cdot e^{-t \sqrt{\delta} / 10 \sqrt{T}},\]
this establishes~\eqref{eqn:or-sym-decay-ref}.

What remains is to perform the correlation calculation to establish~\eqref{eqn:or-sym-corr-ref}. For $t = 1$, we observe
\[\frac{|\omega(1)|}{\omega(0)} = \frac{c\prod_{i = 1}^m (2ci^2)}{(c-1)\prod_{i = 1}^m (2ci^2 - 1)} \ge \prod_{i=1}^m \frac{i^2}{i^2 - 1/(2c)} \ge 1.\]
Next, we observe that the total contribution of $t > c$ to $\|\omega\|_1/\omega(0)$ is at most
\begin{equation} \label{eqn:omega-tail}
\sum_{t > c} \frac{|\omega(t)|}{\omega(0)} \le \sum_{j=1}^{m} \frac{3c}{(2cj^2)^2} < \sum_{j=1}^\infty \frac{3}{4cj^4} = \frac{\pi^4}{120 c}.
\end{equation}

Next, we calculate
\begin{align}
\omega(0) - \sum_{t = 1}^T \omega(t) &\ge \omega(0) - \omega(1) - \left( \sum_{t = c}^T |\omega(t)|\right) \nonumber \\
&\ge \omega(0) - \omega(1) - \left( \omega(c) +  \omega(0) \cdot \frac{\pi^4}{120c}\right) & \text{by } \eqref{eqn:omega-tail} \nonumber \\
&\ge - \omega(1) + \omega(0) \left(1 - \frac{6}{c-1} - \frac{\pi^4}{120c}\right) & \text{by } \eqref{eqn:omega-2} \nonumber \\
&\ge - \omega(1)  + (1 - \delta)\omega(0) & \text{by our choice of } c \ge 8/\delta. \label{eqn:fudgethis}
\end{align}
On the other hand,
\begin{align}
\|\omega\|_1 &\le \omega(0) - \omega(1) + \omega(2) + \omega(0) \cdot \frac{\pi^4}{120c} & \text{by } \eqref{eqn:omega-tail} \nonumber \\
&\le -\omega(1) + \omega(0) \left(1 + \frac{6}{c - 1} + \frac{\pi^4}{120c} \right) &\text{by } \eqref{eqn:omega-2}\nonumber \\
&\le -\omega(1) + (1 + \delta)\omega(0) & \text{since } c \ge 8/\delta. \label{eqn:fudgethat}
\end{align}
Combining~\eqref{eqn:fudgethis} and~\eqref{eqn:fudgethat}, and using the fact that $-\omega(1) \ge \omega(0)$ shows that
\[\frac{\omega(0) - \sum_{t = 1}^k \omega(t)}{\|\omega\|_1} \ge \frac{-\omega(1) + (1-\delta) \omega(0)}{-\omega(1) + (1+\delta) \omega(0)} \ge \frac{2-\delta}{2 + \delta} \ge 1 - \delta.\]
This establishes~\eqref{eqn:or-sym-corr-ref}, completing the proof.
\end{proof}

The following construction of a dual polynomial for $\OR_N$, with $N \ge T$, is an immediate consequence of Minsky-Papert symmetrization (\Cref{lem:mp}), combined with \Cref{prop:or-sym-dual}.

\begin{proposition} \label{prop:or-dual}
Let $T, N \in \N$ with $T \le N$, and let $\delta > 1/T$. Define $\omega$ as in \Cref{prop:or-sym-dual}. Define the function $\psi: \bits^N \to \bits$ by $\psi(x) = \omega(|x|) / \binom{N}{ |x|}$ for $x \in H^N_{\le T}$ and $\psi(x) = 0$ otherwise. Then
\begin{align} &\langle \psi, \OR_N \rangle \ge 1-\delta  \label{eqn:or-corr} \\
& \|\psi\|_1 = 1   \label{eqn:or-norm}\\
& \text{For any polynomial } p \colon \bits^N \to \R\text{, } \deg p < c_1\sqrt{\delta T} \implies \langle \psi, p \rangle = 0 \label{eqn:or-phd}
\end{align}
\end{proposition}

\subsection{Step 2: Constructing a Preliminary Dual Witness for \texorpdfstring{$\AND_R \circ \OR_N$}{AND o OR}} \label{sec:prelim-dual}

The following proposition, when combined with \Cref{prop:ls}, shows that there is a function $\Phi : \bits^R \to \bits$ such that the dual block composition $\Phi \ls \psi$ is a good dual polynomial for $\AND_R \circ \OR_N$. In the next section, we will modify $\Phi \ls \psi$ to zero out the weight it places outside $H_{\le N}^{N \cdot R}$.

\begin{restatable}{proposition}{andor} \label{prop:and-or}
Let $\OR_N : \bits^N \to \bits$ \ and \ $\AND_R \colon \bits^R \to \bits$. \ Let  $\psi\colon \bits^N \to \bits$ be a function such that $\|\psi\|_1 = 1$ and $\langle \psi, \OR_N \rangle \ge 19/20$. Then there exists a function $\Phi \colon \bits^R \to \bits$ with pure high degree $\Omega(\sqrt{R})$ and $\|\Phi\|_1 = 1$ such that
\[ \langle \Phi \ls \psi, \AND_R \circ \OR_N \rangle > 2/3.\]
\end{restatable}

The proof of \Cref{prop:and-or} is implicit in the results of ~\cite{bt13, sherstovandor}.

\subsection{Step 3: Constructing the Final Dual Witness} \label{sec:final-dual}

\label{sec:parametersetting}

\begin{proposition} \label{prop:surj-dual}
Let $R$ be sufficiently large. There exist $N = O(R)$, $D = \tOmega(N^{3/4})$, and  $\zeta : (\bits^N)^R \to \R$ such that
\begin{align} \label{eq:surj-zero} &\zeta(x) = 0 \text{ for all } x \not\in H_{\le N}^{N \cdot R}, \\
\label{eq:surj-corr} &\sum_{x \in H_{\le N}^{N \cdot R}}\zeta(x) \cdot (\AND_{R} \circ \OR_N)(x) > 1/3, \\
 \label{eq:surj-unitnorm} &\|\zeta\|_1 = 1, \text{ and }  \\
 \label{eq:surj-phd} &\text{ For every polynomial } p \colon (\bits^N)^R\to \R \text{ of degree less than } D, \text{ we have } \langle p, \zeta \rangle = 0. 
\end{align}
\end{proposition}

\begin{proof}

We start by fixing choices of several key parameters:
\begin{itemize}
\item $d= \Theta(\sqrt{R})$ is the pure high degree of the dual witness $\Phi$ for $\AND_R$ in \Cref{prop:and-or},
\item $T = \lfloor(R/d)^{1/2}\rfloor^2 = \Theta(\sqrt{R})$,
\item $\hat{D} = c_1 \sqrt{T} \cdot d =\Theta(R^{3/4})$, where $c_1$ is the constant from \Cref{prop:or-sym-dual},
\item $\delta =1/20$
\item $\alpha = 170/\delta = 3400$,
\item $\beta = c_2 \cdot \sqrt{\delta} / \sqrt{T} = \Theta(1/R^{1/4})$, where $c_2$ is the constant from \Cref{prop:or-sym-dual},
\item $N = \lceil 20 \sqrt{\alpha} \rceil R = 693 R$.
\end{itemize}

Let $\psi : \bits^N \to \bits$ be the function constructed in \Cref{prop:or-dual} with $\delta:=1/20$. Let $\Phi : \bits^R \to \bits$ be the function constructed in \Cref{prop:and-or}, and define $\xi = \Phi \ls \psi$. Then by \Cref{prop:ls}, $\xi$ satisfies the following properties:
\begin{align}& \langle \xi, \AND_R \circ \OR_N \rangle > 2/3, \\
& \| \xi \|_1 = 1, \\ 
& \text{ For every polynomial } p \text{ of degree less than } D, \text{ we have }\langle \xi, p \rangle = 0. 
\end{align}
Recall that $\psi$ was obtained by symmetrizing the function $\omega$ constructed in \Cref{prop:or-sym-dual}.
\Cref{prop:btzeroing} guarantees that for some $\Delta \ge \beta \sqrt{\alpha} R / 4\ln^2 R = \tOmega(R^{3/4})$, the function $\xi$ can be modified to produce a function $\zeta : (\bits^N)^R \to \R$ such that
\begin{align*}
&\zeta(x) = 0 \text{ for all } x \not\in H_{\le N}^{N \cdot R}, \\
&\langle \zeta, \AND_R \circ \OR_N \rangle \ge \langle \xi, \AND_R \circ \OR_N \rangle - \|\zeta - \xi\|_1 \ge 2/3 - 2/9 > 1/3, \\
&\|\zeta\|_1 = 1, \\
& \text{ For every polynomial } p \text{ of degree less than } \min\{\hat{D}, \Delta\}, \text{ we have }\langle \zeta, p \rangle = 0. 
\end{align*}
Observing that
\[D =  \min\{\hat{D}, \Delta\} = \tOmega(R^{3/4})\]
shows that the function $\zeta$ satisfies the conditions necessary to prove \Cref{prop:surj-dual}.
\end{proof}

\Cref{thm:main-amp} follows by combining \Cref{prop:surj-dual} with the dual characterization of unbounded approximate degree given in \Cref{prop:ub-duality}. By \Cref{cor:surj-reduction}, we conclude that $\adeg(\dumSURJ_{N, R})=\tilde{\Omega}(R^{3/4})$. \Cref{thm:surj} follows by \Cref{prop:dumsurj-reduction}.
 

\section{Lower Bound For \texorpdfstring{$k$}{k}-Distinctness}
\label{s:kdist}

Our goal is to prove the following lower bound on the approximate degree of the $k$-distinctness function.

\begin{theorem} \label{thm:dist}
For $k \ge 2$ and some $N = O_k(R)$, the $(1/3)$-approximate degree of $\DIST^k_{N, R}$ is $\tilde{\Omega}_k(R^{3/4 - 1/(2k)})$.
The same lower bound holds for the quantum query complexity of $\DIST^k_{N, R}$.
\end{theorem}
Here, the notation $O_k$ hides factors depending only on $k$, and $\tilde{\Omega}_k$ hides factors logarithmic in $R$ and factors depending only on $k$.

\Cref{thm:dist} is a consequence of applying the reductions of \Cref{prop:dumdist-reduction} and \Cref{cor:dist-reduction} to the  following, which is the ultimate goal of this section.

\begin{theorem} \label{thm:dist-promise}
Let $\promcomp : H_{\le N}^{N\cdot R} \to \bits$ equal $\OR_R \circ \THR^k_N$ restricted to inputs in $H_{\le N}^{N \cdot R}$. Then for some $N = O_k(R)$, we have $\ubdeg(\promcomp) \ge \tilde{\Omega}_k(R^{3/4 - 1/(2k)})$.
\end{theorem}

The proof of \Cref{thm:dist-promise} will follow the same basic outline as the proof of \Cref{thm:main-amp}. We will construct a dual polynomial for $\promcomp$ via the following three steps:

\paragraph{Step 1.} We first construct a dual witness $\psi$ showing that the unbounded-approximate degree of the $\THR^k_N$ function is $\Omega_k\left(\sqrt{T N^{-1/k}}\right)$, even when promised that the input has Hamming weight at most $T = \Theta_k(\sqrt{R})$. Moreover, this dual witness satisfies additional properties
that are exploited in Step 2 below.

\paragraph{Step 2.} We combine $\psi$ with a dual witness $\Phi$ for $\OR_R$ to obtain a preliminary dual witness $\Phi \ls \psi$ for $\OR_R \circ \THR_N^k$. The dual witness $\Phi \ls \psi$ shows that $\OR_R \circ \THR^k_N$ has approximate degree $\Omega(\sqrt{R} \cdot \sqrt{T}) = \Omega_k(R^{3/4 - 1/(2k)})$. However, $\Phi \ls \psi$ places weight on inputs of Hamming weight larger than $N$, and hence does not give an \emph{unbounded} approximate degree lower bound for the promise variant $\promcomp$.

\paragraph{Step 3.} Using \Cref{prop:btzeroing} we zero out the mass that $\Phi \ls \psi$ places on inputs of  Hamming weight larger than $N$, while maintaining its pure high degree and correlation with $\promcomp$. This yields the final desired dual witness $\zeta$ for $\promcomp$.

\paragraph{Additional Notation.}

For functions $f : \cX \to \bits$ and $\psi : \cX \to \mathbb{R}$, define the error regions
\begin{align*}
E_+(\psi, f) = \{x \in \cX : \psi(x) > 0, f(x) = -1\} \\
E_-(\psi, f) = \{x \in \cX : \psi(x) < 0, f(x) = +1\}.
\end{align*}

These are the regions where $\psi$ disagrees in sign with $f$. 
We refer to $E_+$ as the set of ``false positive'' errors made by $\psi$,
and $E_-$ as the set of ``false negative'' errors.

\subsection{Step 1: A Dual Witness for \texorpdfstring{$\THR^k_N$}{THR}}

We begin by constructing a univariate version of our dual witness for $\THR^k_N$. Properties~\eqref{eqn:thr-sym-pos-corr} and~\eqref{eqn:thr-sym-neg-corr} below amount to more refined conditions on the correlation between $\omega$ and the (symmetrized) $\THR^k_N$ function. These properties will be needed in order to execute Step 2 of the construction in \Cref{sec:k-dist-step2}.

\begin{proposition} \label{prop:thr-sym-dual}
Let $k, T, N \in \N$ with $k \le T$. There exist constants $c_1, c_2 \in (0, 1] $ and a function $\omega : \{0, 1, \dots, T\} \to \R$ such that
\begin{align} & \sum_{\omega(t) > 0, t \ge k} |\omega(t)| \le \frac{1}{48N}  \label{eqn:thr-sym-pos-corr} \\
& \sum_{\omega(t) < 0, t < k} |\omega(t)| \le \left(\frac{1}{2} - \frac{2}{4^k}\right)\label{eqn:thr-sym-neg-corr} \\
& \|\omega\|_1 := \sum_{t = 0}^{T} |\omega(t)| = 1 \label{eqn:thr-sym-norm} \\
& \text{For all univariate polynomials } q \colon \R \to \R\text{, } \notag \\
& \deg p < c_1\sqrt{k^{-1} \cdot T \cdot N^{-1/k}} \implies \sum_{t=0}^T \omega(t) \cdot q(t) = 0 \label{eqn:thr-sym-phd} \\
&  |\omega(t)| \le \frac{(2k)^k \exp(-c_2 t / \sqrt{k \cdot T \cdot N^{1/k}})}{t^2} \qquad \forall t = 1, 2, \dots, T.\label{eqn:thr-sym-decay}
\end{align}
\end{proposition}

\begin{proof}
If $k = 1$, then the function defined by $\omega(0) = \frac{1}{2}$ and $\omega(1) = -\frac{1}{2}$ satisfies the conditions of the proposition for $c_1 = c_2 = 1$. In what follows, we treat the complementary case where $k \ge 2$.

Let $E_+ := \{t \colon \omega(t) > 0, t \ge k\}$, and $E_{-} := \{t\colon \omega(t) < 0, t < k\}$. 
By normalizing, it suffices to construct a function $\omega : [T]_0 \to \R$ such that

\begin{align} & \sum_{t \in E_+} |\omega(t)| \le \frac{1}{48 N} \cdot \|\omega\|_1  \label{eqn:thr-sym-pos-corr-ref} \\
& \sum_{t \in E_-} |\omega(t)| \le \left(\frac{1}{2} - \frac{2}{4^k}\right) \cdot \|\omega\|_1  \label{eqn:thr-sym-neg-corr-ref} \\
& \text{For all univariate polynomials } q \colon \R \to \R\text{, } \nonumber \\
&\deg p < c_1\sqrt{k^{-1} \cdot T \cdot N^{-1/k}} \implies \sum_{t=0}^k \omega(t) \cdot q(t) = 0 \label{eqn:thr-sym-phd-ref} \\
&  |\omega(t)| \le \frac{(2k)^k \exp(-c_2 t / \sqrt{k \cdot T \cdot N^{1/k}})\|\omega\|_1}{t^2} \qquad \forall t = 1, 2, \dots, T.\label{eqn:thr-sym-decay-ref}
\end{align}

Let $c = 2k\lceil N^{1/k}\rceil$, and let $m = \lfloor \sqrt{T/c} \rfloor$. 
Define the set
\[S = \{1, 2, \dots, k\} \cup \{ci^2 : 0 \le i \le m\}.\]
Note that $|S| = \Omega(k^{-1/2}T^{1/2}N^{-1/(2k)})$. Define the polynomial
\[\omega(t) = \frac{(-1)^{t+(T-m)+1}}{T!} \binom{T}{t} \prod_{r \in [T]_0 \setminus S} (t - r).\]

The signs are chosen so that $\omega(k) < 0$. 
 It is immediate from \Cref{fact:combinatorial} that $\omega$ satisfies~\eqref{eqn:thr-sym-phd-ref} for $c_1 = 1/\sqrt{2}$.
We now show that~\eqref{eqn:thr-sym-decay-ref} holds. For $t = 1, \dots, k$, we have
\[\frac{(2k)^k \exp(-c_2 t / \sqrt{k \cdot T \cdot N^{1/k}})}{t^2} \ge \frac{(2k)^k \exp(-c_2 \sqrt{k})}{k^2} \ge 1\]
as long as $c_2 \le 1/2$ and $k \ge 2$. Since $|\omega(t)| \le \|\omega\|_1$, the bound holds for $t = 1, \dots, k$.

 For $t = cj^2$ with $j \ge 1$, we expand out the binomial coefficient in the definition of $\omega$ to obtain
\[|\omega(t)| = \begin{cases}
 \prod\limits_{r \in S\setminus \{t\}} \frac{1}{|t - r|} & \text{ for } t \in S, \\
0 & \text{ otherwise.}
\end{cases}\]
For $t \in \{0, 1, \dots, k\}$, we observe that
\begin{equation} \label{eqn:head-bound}
\frac{|\omega(t)|}{|\omega(k)|} = \frac{k! \cdot \prod_{i = 1}^m (c i^2 - k)}{t! \cdot (k - t)! \cdot \prod_{i = 1}^m (c i^2 - t)} \le \binom{k}{t}.
\end{equation}
Meanwhile, for $t = cj^2$ with $j \ge 1$, we get
\begin{align*}
\frac{|\omega(t)|}{|\omega(k)|} &= \frac{k! \cdot \prod_{i = 1}^m (ci^2 - k)}{\prod_{i = 1}^k(cj^2 - i) \cdot \prod_{i \in [m]_0 \setminus \{j\}} |ci^2 - cj^2|} \\
&\le \frac{k! \cdot \prod_{i = 1}^m ci^2}{(cj^2 - k)^{k} \cdot \prod_{i \in [m]_0 \setminus \{j\}} c(i+j)|i-j|} \\
&= \frac{k!}{j(cj^2 - k)^{k}} \cdot \frac{(m!)^2}{(m+j)!(m-j)!}.
\end{align*}
The first factor is bounded above by
\[\frac{k!}{(c - k)^k j^{2k+1}}.\] 
As long as $c \ge 2k$ and $k \ge 2$,
this expression is at most
\[\frac{k^k}{(c/2)^kj^4} = \frac{(2k)^k}{c^k \cdot j^4}.\]

We control the second factor by
\begin{align*}
\frac{(m!)^2}{(m+j)!(m-j)!} &= \frac{m}{m+j} \cdot \frac{m-1}{m+j-1} \cdot \ldots \cdot \frac{m-j+1}{m+1} \\
&\le \left(\frac{m}{m+j} \right)^j \\
&\le  \left(1 - \frac{j}{2m} \right)^j \\
&\le e^{-j^2/2m},
\end{align*}
where the last inequality uses the fact that $1 - x \le e^{-x}$ for all $x$. Hence,
\begin{equation} \label{eqn:tail-bound}
\frac{|\omega(cj^2)|}{|\omega(k)|} \le \frac{(2k)^k}{c^k \cdot j^4} \cdot e^{-j^2/2m}.
\end{equation}
This immediately yields
\[\frac{|\omega(cj^2)|}{\|\omega\|_1} \le \frac{|\omega(cj^2)|}{|\omega(k)|} \le \frac{(2k)^k}{(cj^2)^2} \cdot e^{-cj^2/(2cm)},\]
 which establishes~\eqref{eqn:thr-sym-decay-ref} for all $t = cj^2 > k$. 
 
 Moreover, by \eqref{eqn:tail-bound}
\begin{equation} \label{eqn:tail-sum}
\sum_{t > k} |\omega(t)| \le |\omega(k)| \cdot \sum_{j = 1}^m \frac{(2k)^k}{c^k \cdot j^4} \cdot e^{-j^2/2m} \le \frac{(2k)^k}{c^k} \cdot |\omega(k)| \cdot  \sum_{j=1}^m \frac{1}{j^4} \le \frac{|\omega(k)|}{48 N}.
\end{equation}
Hence, since $\omega(k) < 0$,
\[\sum_{t \in E_+} |\omega(t)| \le \sum_{t > k} |\omega(t)| \le \frac{|\omega(k)|}{48 N} \le \frac{\|\omega\|_1}{48N},\]
which gives~\eqref{eqn:thr-sym-pos-corr-ref}.

Finally, to establish~\eqref{eqn:thr-sym-neg-corr-ref}, we combine ~\eqref{eqn:head-bound} and~\eqref{eqn:tail-sum} to obtain
\begin{equation} \label{eqn:mass-on-k}
\frac{\|\omega\|_1}{|\omega(k)|} \le \sum_{t = 0}^k \binom{k}{t} + \frac{1}{48 N} < 2^k + 1 < \frac{1}{2} \cdot 4^k.
\end{equation}
We calculate
\begin{align*}
\frac{\|\omega\|_1}{2} - \sum_{t \in E_-} |\omega(t)| &=  \sum_{t : \omega(t) < 0} (-\omega(t)) - \sum_{t \in E_-} (-\omega(t)) & \text{since $\langle \omega, \mathbf{1} \rangle = 0$}\\
&= \sum_{t \colon \omega(t) < 0, t \geq k} (-\omega(t)) \\
&\ge -\omega(k).
\end{align*}
Rearranging and applying the bound~\eqref{eqn:mass-on-k},
\[\sum_{t \in E_{-}} |\omega(t)| \le \left(\frac{1}{2} + \frac{\omega(k)}{\|\omega\|_1}\right) \cdot \|\omega\|_1 \le \left(\frac{1}{2} - 2 \cdot 4^{-k}\right) \cdot \|\omega\|_1.\]
\end{proof}

Applying Minsky-Papert symmetrization (\Cref{lem:mp}) to ensure that the resulting function has the appropriate pure high degree, \Cref{prop:thr-sym-dual} yields a dual polynomial for $\THR^k_N$.

\begin{proposition} \label{prop:thr-dual}
Let $k, T, N\in \N$ with $k \le T \le N$. Define $\psi : \bits^N \to \R$ by $\psi(x) = \omega(|x|) / \binom{N}{ |x|}$ for $x \in H_{\le T}^N$ and $\psi(x) = 0$ otherwise, where $\omega$ is as constructed in \Cref{prop:thr-sym-dual}. Then
\begin{align}
&\sum_{x \in E_+(\psi, \THR^k_N)} |\psi(x)| \le \frac{1}{48  N}  \label{eqn:thr-pos-corr}\\
&\sum_{x \in E_-(\psi, \THR^k_N)} |\psi(x)| \le \frac{1}{2} - \frac{2}{4^k}  \label{eqn:thr-neg-corr} \\
&\|\psi\|_1 = 1   \label{eqn:thr-norm} \\
&\text{For any polynomial } p \colon \bits^N \to \R\text{, } \deg p < c_1 \sqrt{k^{-1} \cdot T \cdot N^{-1/k}} \implies \langle \psi, p \rangle = 0 \label{eqn:thr-phd} \\
&\sum_{|x| = t} |\psi(x)| \le (2k)^k \exp(-c_2 t / \sqrt{k \cdot T \cdot N^{1/k}})/t^2 \qquad \forall t = 1, 2, \dots, N.\label{eqn:thr-decay}
\end{align}
\end{proposition}

\subsection{Step 2: A Preliminary Dual Witness for \texorpdfstring{$\OR_R \circ \THR^k_N$}{OR o THR}}
\label{sec:k-dist-step2}

\subsubsection{Refined Amplification Lemmas}

The dual witness $\Phi$ for $\OR_R$
that we construct will itself be obtained as a dual block composition $\rho \ls \varphi$. Each constituent dual polynomial will play a distinct role in showing that $\Phi \ls \psi$ is a good dual polynomial for $\OR_R \circ \THR^k_N$. The first function $\varphi$ is an ``error amplifier'' in the sense that $\varphi \ls \psi$ is much better correlated with $\OR_R \circ \THR^k_N$ than $\psi$ is with $\THR^k_N$. The second function $\rho$, on the other hand, is a ``degree amplifier'' in that it serves to increase the pure high degree of $\varphi \ls \psi$.

While the amplification results we need are new, they are relatively straightforward extensions of similar results in~\cite{bt13, sherstovandor, bt14}. Proofs appear below for completeness.

\paragraph{Amplifying Error.} The following proposition shows that
if $\psi$ is a dual witness for the high approximate degree of a Boolean function $f$, 
then there is a dual witness of the form $\varphi \ls \psi$ for $\OR_M \circ f$
such that (a) $\varphi \ls \psi$ may make slightly more false positive errors
than $\psi$ (by at most a factor of $M$), and (b) $\varphi \ls \psi$
makes significantly fewer false-negative errors than $\psi$ (exponentially
smaller in $M$). 

\begin{restatable}{proposition}{erroramp} \label{prop:error-amp}
Let $f : \bits^m \to \bits$ and $\psi : \bits^m \to \R$ be functions such that
\begin{align} &\sum_{x \in E_+(\psi, f)} |\psi(x)| \le \delta^+ \label{eqn:error-base-pos} \\
& \sum_{x \in E_-(\psi, f)} |\psi(x)| \le \delta^- \label{eqn:error-base-neg} \\
& \|\psi\|_1 = 1.\label{eqn:error-base-norm}
\end{align}
For every $M \in \N$, there exists a function $\varphi : \bits^M \to \bits$ with $\|\varphi\|_1 = 1$ and pure high degree $1$ such that
\begin{align} & \sum_{x \in E_+(\varphi \ls \psi, \OR_M \circ f)} |(\varphi \ls \psi)(x)| \le M \cdot \delta^+ \label{eqn:error-amp-pos} \\
& \sum_{x \in E_-(\varphi \ls \psi, \OR_M \circ f)} |(\varphi \ls \psi)(x)| \le  \frac{1}{2} \cdot (2\delta^-)^M.\label{eqn:error-amp-neg}
\end{align}
\end{restatable}

\begin{proof} 
Let $\varphi: \{-1, 1\}^M \rightarrow \{-1, 1\}$ be defined such that $\varphi(\mathbf{1})=1/2$, $\varphi(-\mathbf{1})=-1/2$, and $\varphi(x)=0$ for all other $x$. 
Notice that 
\begin{equation}  \sum_{(x_1, \dots, x_{M}) \in \{-1, 1\}^M} \varphi(x_1, \dots, x_{M}) = 0 \label{eq:balanced} \end{equation}
so $\Psi$ has pure high degree $1$, and that $\|\Psi\|_1 = 1$.

We now prove that (\Cref{eqn:error-amp-pos}) holds. Let $\lambda$ be the distribution on $\{-1, 1\}^{m}$ given by $\lambda(x) = |\psi(x)|$, and let $\lambda^{\otimes M}$ be the product distribution on $(\{-1, 1\}^m)^M$ given by $\lambda^{\otimes M}(x_1, \dots, x_M) = \prod_{i=1}^M |\psi(x_i)|$.
Since $\psi$ is orthogonal to the constant polynomial, it has expected value 0, and hence the string $(\dots, \sgn(\psi(x_i)), \dots)$ is distributed uniformly in $\{-1, 1\}^{M}$
when one samples $(x_1, \dots, x_{M})$ according to $\lambda^{\otimes M}$. This allows us to write
\begin{align}
&\hspace{-2em}\sum_{(x_1, \dots, x_{M}) \in E^+(\varphi \ls \psi, \OR_M \circ f)} |(\varphi \ls \psi)(x_1, \dots, x_M)| \nonumber \\
&\hspace{-2em}= 2^{M} \mathbf{E}_{\lambda^{\otimes M}} [\varphi( \dots, \sgn(\psi(x_i)), \dots) \cdot \mathbb{I} (\varphi( \dots, \sgn(\psi(x_i)), \dots) > 0 \land \OR_M\left( \dots, f(x_i), \dots\right) = -1)] \nonumber\\
&\hspace{-2em}= \sum_{z : \varphi(z) > 0} \varphi(z) \cdot \Pr_{\lambda^{\otimes M}} [\OR_M(\dots, f(x_i), \dots) = -1 | (\dots, \sgn(\psi(x_i)), \dots) = z]. \label{eqn:error-amp-middle}
\end{align}
Observe that for any bit $b$,
\[\Pr_{x \sim \lambda} [f(x) \ne \sgn(\psi(x)) | \sgn(\psi(x)) = b] = 2\sum_{x \in A_b} |\psi(x)|,\]
where for brevity, we have written $A_{+1} = E_+(\psi, f)$ and $A_{-1} = E_-(\psi, f)$. Therefore, as noted in \cite{sherstovhalfspaces1}, for any given $z \in \{-1, 1\}^{M}$, the following two random variables are identically distributed:

\begin{itemize}
\item The string $(\dots, f(x_i), \dots)$, when one chooses $(\dots, x_i, \dots)$ from $\lambda^{\otimes M}$ conditioned on \newline $(\dots, \sgn(\psi(x_i)), \dots) = z$
\item The string $(\dots, y_iz_i, \dots)$, where $y \in \{-1, 1\}^{M}$ is a random string whose $i$th bit independently
takes on value $-1$ with probability $2 \sum_{x \in A_{z_i}} |\psi(x)|$.
\end{itemize}
Thus, Expression~\eqref{eqn:error-amp-middle} equals
\begin{equation}\sum_{z : \varphi(z) > 0} \varphi(z) \cdot \Pr_y[\OR_M(\dots, y_iz_i, \dots) = -1], \label{eqn:error-amp-middle2}\end{equation}
where $y \in \{-1, 1\}^{M}$ is a random string whose $i$th bit independently
takes on value $-1$ with probability $2 \sum_{x \in A_{z_i}} |\psi(x)|$. The only term in this sum corresponds to $z=\mathbf{1}$, which we now argue is at most $M \delta^+$.
By \eqref{eqn:error-base-pos}, each $y_i = -1$ with probability $2 \sum_{x \in A_{1}} |\psi(x)| = 2\sum_{x \in E^+(\psi, f)} |\psi(x)| \le 2\delta^+$.  Hence, for $z = \mathbf{1}$, we have 
\[\Pr_y[\OR_M(\dots, y_i, \dots) = -1] \le 2 M \delta^+\]
by a union bound. Thus Expression~\eqref{eqn:error-amp-middle2} is at most $M \delta^+$, proving \eqref{eqn:error-amp-pos}.

It now remains to prove the bound~\eqref{eqn:error-amp-neg}. By an identical argument as above, we have
\begin{equation}\sum_{(x_1, \dots, x_{M}) \in E^-(\varphi \ls \psi, \OR_M \circ f)} |(\varphi \ls \psi)(x_1, \dots, x_M)| = \sum_{z : \varphi(z) < 0} \varphi(z) \cdot \Pr_y[\OR_M(\dots, y_iz_i, \dots) = 1]. \label{eqn:error-amp-end}\end{equation}
The only term in the sum corresponds to $z=-\mathbf{1}$, which we argue takes value $\frac{1}{2} \cdot (2\delta^-)^{M}$. 
Here, each $y_i=-1$ independently with probability $\sum_{x \in A_{-1}} |\psi(x)| = 2\sum_{x \in E^-(\psi, f)} |\psi(x)| \le 2\delta^-$, and $\OR_M(\dots, -y_i, \dots)=1$ only if $y_i = -1$ for every $i$. Hence, we conclude that
\[\Pr_y[\OR_{M}\left(\dots, -y_i, \dots\right) = 1]  \le (2\delta^-)^M. \]
It follows that Expression~\eqref{eqn:error-amp-end} is at most $\frac{1}{2} \cdot (2\delta^-)^M$, establishing~\eqref{eqn:error-amp-pos}. This completes the proof.
\end{proof}

\paragraph{Amplifying Degree.}

The following proposition states that if $\psi$ is a dual polynomial for a Boolean function $f$,
then there is a dual polynomial $\rho \ls \psi$ for $\OR_M \circ f$
with significantly larger pure high degree that does not make too many
more false positive and false negative errors than does $\psi$ itself.

\begin{restatable}{proposition}{degreeamp} \label{prop:degree-amp}
Let $f : \bits^m \to \bits$ and $\psi : \bits^m \to \R$ be functions such that
\begin{align} & \sum_{x \in E_+(\psi, f)} |\psi(x)| \le \delta^+ \label{eqn:deg-base-pos} \\
& \sum_{x \in E_-(\psi, f)} |\psi(x)| \le \delta^- \label{eqn:deg-base-neg} \\
& \|\psi\|_1 = 1 \label{eqn:deg-base-norm}
\end{align}
For every $M \in \N$ there exists a function $\rho : \bits^M \to \R$ with $\|\rho\|_1 = 1$ and pure high degree $\Omega(\sqrt{M})$ such that
\begin{equation}
\langle \rho \ls \psi, \OR_M \circ f \rangle \ge \frac{9}{10} - 4  M \delta^+ - 4\delta^-. \label{eqn:deg-amp-corr}
\end{equation}
\end{restatable}

\begin{proof}

\Cref{andornor} shows that the function $\OR_M$ has $(9/10)$-approximate degree $\Omega(\sqrt{M})$. Hence, \Cref{prop:duality} guarantees the existence of a function $\rho : \bits^M \to \R$ with $\|\rho\|_1 = 1$ and pure high degree $\Omega(\sqrt{M})$ such that
\begin{equation}
\langle \rho, \OR_M \rangle \ge \frac{9}{10}.
\end{equation}
What remains is to establish the correlation bound~\eqref{eqn:deg-amp-corr}. Letting $\lambda$ denote the distribution $\lambda(x) = |\psi(x)|$ as in the proof of \Cref{prop:error-amp}, we may write
\begin{align}
& \sum_{(x_1, \dots, x_{M}) \in (\bits^m)^M} (\varphi \ls \psi) (x_1, \dots, x_M) \cdot \OR_M(\dots, f(x_i), \dots) \nonumber \\
&\qquad= 2^{M} \mathbf{E}_{\lambda^{\otimes M}} [\varphi( \dots, \sgn(\psi(x_i)), \dots) \cdot \OR_M\left( \dots, f(x_i), \dots\right)] \nonumber\\
&\qquad= \sum_{z \in \bits^M} \varphi(z) \cdot \mathbf{E}_{\lambda^{\otimes M}} [\OR_M(\dots, f(x_i), \dots) | (\dots, \sgn(\psi(x_i)), \dots) = z] \nonumber\\
&\qquad = \sum_{z \in \bits^M} \varphi(z) \cdot \mathbf{E}_{y} [\OR_M(\dots, y_i z_i, \dots)], \label{eqn:deg-amp-middle}
\end{align}
where $y \in \bits^M$ is a random string whose $i$th bit independently takes the value $-1$ with probability $2 \sum_{x \in A_{z_i}} |\psi(x)|$. (Here, we are using the abbreviated notation $A_{+1} = E_+(\psi, f)$ and $A_{-1} = E_-(\psi, f)$.) We first consider the contribution of the term corresponding to $z = \mathbf{1}$ to the sum. Here, by a union bound,
\begin{align*}
\mathbf{E}_{y} [\OR_M(\dots, y_i z_i, \dots)] &= 1 - 2 \Pr_y[\OR_M(\dots, y_i, \dots) = -1] \\
&\ge 1 - 2 M \cdot \left(2 \sum_{x \in A_{+1}} |\psi(x)|\right) \\
&\ge 1 - 4 M\delta^+.
\end{align*}
Hence, the term $z = \mathbf{1}$ contributes $\varphi(\mathbf{1}) \cdot (1 - 4 M \delta^+)$ to the sum.

Now we consider the contribution of any term corresponding to $z \ne \mathbf{1}$. Given such a $z$, let $i^*$ be an index such that $z_{i^*} = -1$. Then we have
\begin{align*}
-\mathbf{E}_{y} [\OR_M(\dots, y_i z_i, \dots)] &= 1 - 2 \Pr_y[\OR_M(\dots, y_i z_i, \dots) = 1] \\
&\ge 1 - 2 \cdot \Pr_{y_{i^*}}[y_{i^*} = -1] \\
&= 1 - 2 \cdot \left(2 \sum_{x \in A_{-1}} |\psi(x)|\right) \\
&\ge 1 - 4 \delta^-.
\end{align*}

We can now lower bound~\eqref{eqn:deg-amp-middle} by
\begin{align*}\hspace{-1.1em}
\varphi(\mathbf{1}) \cdot (1 - 4 M \delta^+) - \sum_{z \ne \mathbf{1}} \varphi(z) (1 - 4\delta^-) &\ge \sum_{z \in \bits^M} \varphi(z) \OR_M(z) - 4M\delta^+ |\varphi(\mathbf{1})| - 4\delta^- \sum_{z \ne 1} |\varphi(z)| \\
&\ge \frac{9}{10} - 4M\delta^+ - 4\delta^-.\qedhere
\end{align*}
\end{proof}

\subsubsection{Constructing a Dual Witness for $\OR_R \circ \THR^k_N$}

We now combine our amplification lemmas to construct a dual witness for $\OR_R \circ \THR^k_N$.

\begin{proposition} \label{prop:dist-prelim}
Let $k, T, N, R \in \N$ with $k \le T \le R \le N$ and $R$ divisible by $4^k$. Let $\psi : \bits^N \to \R$ be a function with $\|\psi\|_1 = 1$ and
\begin{align*}
&\sum_{x \in E_+(\psi, \THR^k_N)} |\psi(x)| \le \frac{1}{48 N} \\
&\sum_{x \in E_-(\psi, \THR^k_N)} |\psi(x)| \le \frac{1}{2} - \frac{2}{4^k}. 
\end{align*}
Then there exists a function $\Phi : \bits^R \to \R$ with $\|\Phi\|_1 = 1$ and pure high degree $\Omega(2^{-k} \sqrt{R})$ such that
\[ \langle \Phi \ls \psi, \OR_R \circ \THR^k_N \rangle \ge 2/3. \]
\end{proposition}

\begin{proof}
Using the construction of \Cref{prop:error-amp} with $m = N$, $M = 4^k$, and $f = \THR_N^k$, we first obtain a function $\varphi : \bits^{4^k} \to \R$ with $\|\varphi\|_1 = 1$ and pure high degree $1$ such that
\begin{align*}
&\sum_{x \in E_+(\varphi \ls \psi, \OR_{4^k} \circ \THR^k_N)} |(\varphi \ls \psi)(x)| \le \frac{4^k}{48 N} \\
&\sum_{x \in E_-(\varphi \ls \psi, \OR_{4^k} \circ \THR^k_N)} |(\varphi \ls \psi)(x)| \le \frac{1}{2} \cdot \left(1 - 4 \cdot 4^{-k} \right)^{4^k} \le \frac{e^{-4}}{2}.
\end{align*}
Now by the construction of \Cref{prop:degree-amp} with $m = 4^k \cdot N$, $M = R/4^{k}$ and $f = \OR_{4^k} \circ \THR_N^k$, there exists a function $\rho : \bits^{R / 4^k} \to \R$ with $\|\rho\|_1 = 1$ and pure high degree $\Omega(2^{-k}\sqrt{R})$ such that
\[ \langle \rho \ls (\varphi \ls \psi), \OR_{R/4^k} \circ (\OR_{4^k} \circ \THR^k_N) \rangle \ge \frac{9}{10} - \frac{4R}{48 N} - 2e^{-4} \ge \frac{2}{3}.\]
Let $\Phi : \bits^R \to \R$ be the dual block composition $\rho \ls \varphi$. Since the dual block composition preserves $\ell_1$-norms (\Cref{prop:ls}, Condition~\eqref{eqn:ls-norm}) and multiplies pure high degrees (\Cref{prop:ls}, Condition~\eqref{eqn:ls-phd}), the function $\Phi$ itself has $\ell_1$-norm $1$ and pure high degree $\Omega(2^{-k} \sqrt{R})$. The claim now follows from the associativity of dual block composition (\Cref{prop:ls}, Condition~\eqref{eqn:ls-assoc}) and the fact that $\OR_R = \OR_{R/4^k} \circ \OR_{4^k}$.
\end{proof}

\subsection{Step 3: Completing the Construction} 

We are now ready to apply \Cref{prop:btzeroing} to zero out the mass that the construction of \Cref{prop:dist-prelim} places on inputs outside of $H_{\le N}^{N \cdot R}$.

\begin{proposition} \label{prop:dist-dual}
Let $R$ be sufficiently large. There exist $N = O((2k)^{k/2} \cdot R)$, $D = \tOmega(2^{-k/4} k^{(k-3)/ 4} \cdot R^{3/4 - 1/(2k)})$, and  $\zeta : (\bits^N)^R \to \R$ such that
\begin{align} \label{eq:dist-zero} &\zeta(x) = 0 \text{ for all } x \not\in H_{\le N}^{N \cdot R}, \\
\label{eq:dist-corr} &\sum_{x \in H_{\le N}^{N \cdot R}}\zeta(x) \cdot (\OR_{R} \circ \THR_N^k)(x) > 1/3, \\
 \label{eq:dist-unitnorm} &\|\zeta\|_1 = 1, \text{ and }  \\
 \label{eq:dist-phd} &\text{ For every polynomial } p \colon (\bits^N)^{R}\to \R \text{ of degree less than } D, \text{ we have } \langle p, \zeta \rangle = 0. 
\end{align}
\end{proposition}

\begin{proof}
We start by fixing choices of several key parameters:
\begin{itemize}
\item $d= \Theta(2^{-k}\sqrt{R})$ is the pure high degree of the dual witness $\Phi$ for $\OR_R$ in \Cref{prop:dist-prelim},
\item $T =  \lfloor (8k)^{k/2} \sqrt{R}\rfloor$, 
\item $\alpha = (2k)^k$,
\item $N = \lceil 20 \sqrt{\alpha} \rceil R = \Theta((2k)^{k/2} R)$,
\item $\hat{D} = c_1 \sqrt{k^{-1} \cdot T \cdot N^{-1/k}} \cdot d =\Theta(2^{-k/4} k^{(k-3)/4}  \cdot R^{3/4 - 1/(2k)})$, where $c_1$ is the constant from \Cref{prop:thr-sym-dual},
\item $\beta = c_2 / \sqrt{k \cdot T \cdot N^{1/k}} = \Theta(2^{-3k/4} k^{(-k-3)/4} \cdot R^{-1/4 - 1/(2k)})$, where $c_2$ is the constant from \Cref{prop:thr-sym-dual}.
\end{itemize}

Let $\psi : \bits^N \to \bits$ be the function constructed in \Cref{prop:thr-dual}. Let $\Phi : \bits^R \to \bits$ be the function constructed in \Cref{prop:dist-prelim}, and define $\xi = \Phi \ls \psi$. Then by \Cref{prop:ls}, $\xi$ satisfies the following properties:
\begin{align}& \langle \xi, \OR_R \circ \THR_N^k \rangle > 2/3, \\
& \| \xi \|_1 = 1, \\ 
& \text{ For every polynomial } p \text{ of degree less than } \hat{D}, \text{ we have }\langle \xi, p \rangle = 0. 
\end{align}
Recall that $\psi$ was obtained by symmetrizing the function $\omega$ constructed in \Cref{prop:thr-sym-dual}.
\Cref{prop:btzeroing} guarantees that for some $\Delta \ge \beta \sqrt{\alpha} R / 4\ln^2 R = \tOmega(2^{-k/4} k^{(k-3)/4} \cdot R^{3/4 - 1/(2k)})$, the function $\xi$ can be modified to produce a function $\zeta : (\bits^N)^R \to \R$ such that
\begin{align*}
&\zeta(x) = 0 \text{ for all } x \not\in H_{\le N}^{N \cdot R}, \\
&\langle \zeta, \OR_R \circ \THR_N^k \rangle \ge \langle \xi, \OR_R \circ \THR_N^k \rangle - \|\zeta - \xi\|_1 \ge 2/3 - 2/9 > 1/3, \\
&\|\zeta\|_1 = 1, \\
& \text{ For every polynomial } p \text{ of degree less than } \min\{\hat{D}, \Delta\}, \text{ we have }\langle \zeta, p \rangle = 0. 
\end{align*}
Observing that
\[D =  \min\{\hat{D}, \Delta\} = \tOmega(2^{-k/4} k^{(k-3)/ 4} \cdot R^{3/4-1/(2k)})\]
shows that the function $\zeta$ satisfies the conditions necessary to prove \Cref{prop:dist-dual}.
\end{proof}

\Cref{thm:dist-promise} now follows by combining \Cref{prop:dist-dual} with the dual characterization of unbounded approximate degree \Cref{prop:ub-duality}. The approximate degree lower bound in \Cref{thm:dist} is then a consequence of \Cref{prop:dumdist-reduction} and \Cref{cor:dist-reduction}. The quantum query lower bound follows via the standard fact
that the $\eps$-error quantum query complexity of $f$ is lower bounded by $1/2 \cdot \adeg_{2\eps}(f)$ \cite{beals}.

 
\section{Lower Bound for Image Size Testing and Its Implications}
\label{s:final}
\subsection{Image Size Testing} 

The Image Size Testing problem ($\SE$ for short) is defined as follows. 

\begin{definition}
Given an input $s=(s_1, \dots, s_N) \in [R]_0^N$, and $i \in [R]$, let
$f_i = |\{j \colon s_j=i\}|$. The \emph{image size} of $s$
is the number of $i \in [R]$ such that $f_i > 0$. For $0 < \gamma < 1$, define:
$$\SE_{N, R}^{\gamma}(s_1, \dots, s_N) = \begin{cases} 
-1 & \text{ if the image size is } R\\
1 & \text{ if the image size is at most } \gamma \cdot R\\
\text{undefined} & \text{otherwise}.\end{cases}$$
\end{definition}

Observe that the definition of $\SE$ ignores whether or not 
the range item 0 has positive frequency, just like the functions
$\dumSURJ$ and $\dumDIST^k$. We choose
to define $\SE$ in this manner to streamline
our analysis.

The goal of this section is to prove the following lower bound.

\begin{theorem} \label{sethm}
 For some constant $c > 0$, and any constant $\gamma \in (0, 1)$, $\adeg\left(\SE_{N, R}^{\gamma}\right) \ge \tOmega(R^{1/2})$, where $N = c \cdot \gamma^{-1/2} \cdot R$.
 The same lower bound applies to the quantum query complexity of $\SE_{N, R}^{\gamma}$.
 \end{theorem}

 \begin{remark}
 It is possible to refine our analysis to show that even the \emph{unbounded} approximate degree of $\SE_{N, R}^{\gamma}$ is $\tOmega(R^{1/2})$, and that this holds even if the error parameter is $1-2^{-n^{\Omega(1)}}$. 
 However, for brevity we do not explicitly establish this stronger result. We direct the interested reader to subsequent work
 \cite{btlargeerror}, which shows that the threshold
 degree of $\SURJ_{N, R}$ is $\tOmega(R^{1/2})$ for $R \leq N/2$. The proof of that result 
 can be extended with little difficulty to show that the threshold degree of $\SE_{N, R}^{\gamma}$ is $\tOmega(R^{1/2})$.
 \end{remark}

\subsubsection{Connecting Symmetric Promise Properties and Block Compositions of Partial Functions}
For any function $f$ and symmetric function $g$, \Cref{sec:step1} described a connection between the symmetric
property
\begin{align*}
&F^{\operatorname{prop}}(s_1, \dots, s_N) = f(g(\id[s_1 = 1], \dots, \id[s_N = 1]), \dots, g(\id[s_1=R], \dots,\id[s_N=R]))
\end{align*}
and the partial function 
\begin{align*}
&F^{\le N}(x_1, \dots, x_R) = \begin{cases} f(g(x_1), \dots, g(x_R)) & \text{ if } x_1, \dots, x_R \in \bits^N,  |x_1| + \dots + |x_R| \le N,\\
\text{undefined} & \text{otherwise}.
\end{cases}
\end{align*}

For simplicity and clarity, that discussion was restricted to total functions $f$ and $g$
(in particular, this avoided having to address the possibility that $f(g(x_1), \dots, g(x_R))$ 
is undefined in the definitions of $F^{\operatorname{prop}}$ and $F^{\leq N}$).
Because $\SE$ is a partial function, we need to explain
that the same connection still holds even when $f$ is a partial function.
To do this, we need to introduce the notion of the \emph{double-promise} approximate degree of $F^{\le N}$. 

\begin{definition} \label{thisisreallyaproblem}
Let $Y \subset \bits^R$ and $f \colon Y \to \bits$,
and let $g \colon \bits^N \to \bits$ be a symmetric (total) function.
Let $G = \{x_1, \dots, x_R \colon (g(x_1), \dots, g(x_R)) \in Y\}.$
Let $F^{\le N}$ be defined as above.
Observe that $F^{\le N}$ is defined at all inputs 
in $H_{\leq N}^{N \cdot R} \cap G$.
The double-promise $\eps$-approximate degree of $F^{\le N}$,
denoted $\dpdeg(F^{\le N})$
is the least degree of a real polynomial $p$ such that:
\begin{align}
& |p(x) - F^{\le N}(x)| \leq \eps \text{ for all } x \in H_{\leq N}^{N \cdot R} \cap G.\\
& |p(x)| \leq  1+\eps \text{ for all } x \in  H_{\leq N}^{N \cdot R} \setminus G.
\end{align}
\end{definition}
Observe that in the definition above, no restriction is placed on $p(x)$
for any inputs that are not in $H_{\leq N}^{N \cdot R}$.
\medskip

Bun and Thaler's analysis from \cite{adegsurj} (cf. \Cref{thm:main-reduction}) applies to partial functions $f$ in the following way. 
\begin{theorem}[Bun and Thaler \cite{adegsurj}] \label{thm:main-reductionpartial}
Let $Y \subset \bits^R$ and let $f : Y \to \bits$ be any partial function. Let $g : \bits^N \to \bits$ be any symmetric function. Then for $F^{\operatorname{prop}}$ and $F^{\le N}$ defined above, and for any $\eps > 0$, we have
\[\adeg_\eps(F^{\operatorname{prop}}) \ge \dpdeg_\eps(F^{\le N}).\]
\end{theorem}
  
We will require the following dual formulation of $\dpdeg_{\eps}(F^{\le N})$. 

\begin{proposition} \label{prop:dualitydiedie}
Let $F^{\le N}$ and $G$ be defined as above. Then $\dpdeg_{\eps}(F^{\le N}) \geq d$ if and only if there exists a function 
$\psi \colon \bits^n \to \R$ satisfying the following properties.
\begin{align}\label{eq:ub-zerodp} &\psi(x) = 0 \text{ for all } x \not\in H_{\leq N}^{N \cdot R}, \\
\label{eq:corrdp} &\sum_{x \in H_{\leq N}^{N \cdot R} \cap G}\psi(x) \cdot F^{\le N}(x) - \sum_{x \in H_{\leq N}^{N \cdot R} \setminus G} |\psi(x)| > \eps, \\
 \label{eq:unitnormdp} &\sum_{x \in \bits^n} |\psi(x)| = 1, \text{ and }  \\
 \label{eq:phddp} &\text{ For every polynomial } p \colon \bits^n \to \R \text{ of degree less than } d, \sum_{x \in \bits^n} p(x) \cdot \psi(x) = 0. 
\end{align}
\end{proposition}

We will need to define the following partial function.

\begin{definition}
Define $\GAPAND_{R}^{\gamma} \colon H_{\leq (\gamma \cdot R)}^R \cup \{\mathbf{-1}\} \to \bits$ via:
$$\GAPAND_{R}^{\gamma}(x) = \begin{cases} 
-1 & \text { if } x_i=-1 \text{ for all } i \\
1 & \text{ if } x\in H_{\leq (\gamma \cdot R)}^R\\
\text{undefined} & \text{otherwise}.
\end{cases}$$
\end{definition}

In the case where $f = \GAPAND_{R}^{\gamma}$ and $g = \OR_N$, the function $F^{\operatorname{prop}}(s_1, \dots, s_N)$ is precisely the $\SE_{N, R}^\gamma$ function. Hence:

\begin{corollary} \label{cor:se-reduction}
Let $N, R \in \N$. Let $Y$ be the domain of $\GAPAND_R^{\gamma}$,
and let $G=\{(x_1, \dots, x_R) \in \bits^{N \cdot R} \colon (\OR_N(x_1), \dots, \OR_N(x_R)) \in Y\}.$
Then for any $\eps > 0$,
\[\adeg_\eps(\SE_{N, R}^\gamma) \ge \dpdeg_\eps(F^{\le N})\]
where $F^{\le N} : G \cap H^{N\cdot R}_{\le N} \to \bits$ is the partial function
obtained by restricting $\GAPAND_{R}^{\gamma} \circ \OR_N$ to $H^{N\cdot R}_{\le N}$.
\end{corollary}

With \Cref{cor:se-reduction} in hand, we now turn to proving 
a lower bound on $\dpdeg_{\eps}(F^{\le N})$.

\subsubsection{Completing the Proof of \Cref{sethm}} \label{sec:ughsection}
\begin{proof}
We construct a dual polynomial to witness the lower bound in \Cref{sethm}.

Define the parameters
\begin{itemize}
\item  $\delta = \gamma/4$, 
\item $\alpha = 170/\delta$, 
\item $T=N=\lceil 20 \sqrt{\alpha} \rceil R \leq 310 \cdot \gamma^{-1/2} \cdot R$,
\item $\beta = c_2 \cdot \sqrt{\delta} / \sqrt{T}$, where $c_2$ is the constant from \Cref{prop:or-sym-dual}.
\end{itemize}

Let 
$\psi$ be the dual witness for $\OR_N$ from \Cref{prop:or-dual} with $T=N$. Define $\Phi \colon \bits^R \to \reals$ as follows:
$$\Phi(x) = \begin{cases}
-1/2 & \text{ if } x=(-1, -1, \dots, -1)\\
1/2 & \text{ if } x = (1, 1, \dots, 1)\\
0 & \text{otherwise}.
\end{cases}$$

The dual block composition $\Phi \ls \psi$ is the same dual witness for $\AND_R \circ \OR_N$ which Bun and Thaler \cite{bt14} used to show that $\adeg_{\eps}(\AND_R \circ
\OR_N) = \Omega(N^{1/2})$ for $\eps = 1-2^{-R}$. This dual witness was also
used in subsequent works \cite{sherstov14, bchtv}. Curiously, we are
interested in this dual witness for a completely different reason than these
prior works. These prior works were interested in $\Phi \ls \psi$ because its correlation
with the target function $\AND_R \circ \OR_N$ is exponentially closer to 1
than is the correlation of $\psi$ with $\OR_N$. For our purposes,
it will not be essential to exploit such a strong correlation guarantee---rather, we are
interested in $\Phi \ls \psi$ because most its ``$\ell_1$-mass''
lies on inputs either with full image or tiny image (i.e.,
most of its mass lies in the domain of $\GAPAND_{R}^{\gamma} \circ \OR_N$).

As in \Cref{cor:se-reduction}, let $G$ denote the set of all inputs on
which $\GAPAND_{R}^{\gamma} \circ \OR_N$ is defined,
i.e., $$G=\{x_1, \dots, x_R \in \bits^{N \cdot R} \colon (\OR_N(x_1), \dots, \OR_N(x_R)) \in
H_{\leq \gamma \cdot R}^R \cup \{\mathbf{-1}\}\}.$$
The analysis in these prior works \cite{bt14,bchtv} implies that 
$\Phi \ls \psi$ satisfies the following three properties.
\begin{align} &\|\Phi \ls \psi \|_1=1, \label{ughnorm} \\
&\sum_{x \in G}  ( \Phi \ls \psi)(x) \cdot \left(\GAPAND_{R}^{\gamma} \circ \OR_N\right)(x) - \sum_{x \in\bits^{N \cdot R} \setminus G} |(\Phi \ls \psi)(x)|
\geq 9/10,  \label{ughcor} \\
&\text{For any polynomial } p \colon \bits^{N\cdot R} \to \R\text{, } \deg p < c_1\sqrt{\delta T} \implies \langle \Phi \ls \psi, p \rangle = 0,   \label{ughphd}
\end{align}
where $c_1$ is the constant from \Cref{prop:or-dual}. Indeed, Properties~\eqref{ughnorm} and~\eqref{ughphd} are immediate from \Cref{prop:ls} on the properties of dual block composition. For completeness, we prove that Property~\eqref{ughcor} holds in \Cref{sec:gap-and-corr} below, making use of the fact that $\delta = \gamma/4$ and taking $R$ to be sufficiently large.

\Cref{prop:btzeroing} (with $\alpha$ and $\beta$ set as above) guarantees that for some $\Delta \ge \beta \sqrt{\alpha} R / 4\ln^2 R = \tOmega(R^{1/2})$, the function $\Phi \ls \psi$ can be modified to produce a function $\zeta : (\bits^N)^R \to \R$ such that
\begin{align*}
&\zeta(x) = 0 \text{ for all } x \not\in H_{\le N}^{N \cdot R}, \\
&\langle \zeta, \GAPAND_{R}^{\gamma} \circ \OR_N \rangle \ge \langle \Phi \ls \psi, \GAPAND_{R}^{\gamma} \circ \OR_N \rangle - \|\zeta - \Phi \ls \psi\|_1 \ge 9/10 - 2/9 > 1/3, \\
&\|\zeta\|_1 = 1, \\
& \text{ For every polynomial } p \text{ of degree less than } D:=\min\{\hat{D}, \Delta\}, \text{ we have }\langle \zeta, p \rangle = 0. 
\end{align*}
Observing that
\[D =  \min\{c_1\sqrt{\delta T}, \Delta\} = \tOmega(R^{1/2})\]
shows that the function $\zeta$ satisfies the conditions necessary to prove \Cref{sethm} via \Cref{prop:dualitydiedie}.
\end{proof}

\subsubsection{Proof of Property~\eqref{ughcor}}
\label{sec:gap-and-corr}

\begin{lemma} \label{lem:ughcor}
Let $\delta > 0$, $\gamma > 2\delta$, and let $$G=\{x_1, \dots, x_R \in \bits^{N \cdot R} \colon (\OR_N(x_1), \dots, \OR_N(x_R)) \in
H_{\leq \gamma \cdot R}^R \cup \{\mathbf{-1}\}\}.$$
Define $\Phi : \bits^R \to \bits$ by $\Phi(-\mathbf{1}) = -1/2$, $\Phi(\mathbf{1}) = 1/2$ and $\Phi(z) = 0$ otherwise. Let $\psi : \bits^N \to \bits$ be any dual witness for $\OR_N$ such that $\|\psi\|_1 = 1$, $\langle \psi, \mathbf{1}\rangle = 0$, and $\langle \psi, \OR_N \rangle \ge 1 - \delta$. Then
\[\sum_{x \in G}  ( \Phi \ls \psi)(x) \cdot \left(\GAPAND_{R}^{\gamma} \circ \OR_N\right)(x) - \sum_{x \in\bits^{N \cdot R} \setminus G} |(\Phi \ls \psi)(x)|
\geq 1-\delta^R -\exp(-(\gamma - \delta) R/3).\]
\end{lemma}

The proof of \Cref{lem:ughcor} crucially relies on a special property, called \emph{one-sided error}, that is satisfied by any dual polynomial for $\OR$.

\begin{definition}
Let $f : \bits^N \to \bits$ and let $\psi : \bits^N \to \R$. We say that $\psi$ has one-sided error with respect to $f$ if for all $x \in \bits^N$,
\begin{equation}
f(x) = 1 \implies \psi(x) > 0.
\end{equation}
\end{definition}
The following lemma shows that \emph{any} dual witness for the $\OR_N$ function has one-sided error.

\begin{lemma}[Gavinsky and Sherstov \cite{comm2}] \label{lem:gs}
Let $\psi : \bits^N \to \R$ be a function with pure high degree at least $1$ such that $\langle \psi, \OR_N \rangle > 0$. Then $\psi$ has one-sided error with respect to $\OR_N$.
\end{lemma}

In particular, if $\psi$ is such a dual witness for $\OR_N$, then we have
\begin{equation} \label{eqn:mybullshitequation}
\sum_{x \in A_{+1}} |\psi(x)| \le \frac{1}{2}(1 - \langle \psi, \OR_N \rangle), \qquad \sum_{x \in A_{-1}} |\psi(x)| = 0,
\end{equation}
where the sets $A_{+1}$ and $A_{-1}$ are, respectively, the sets of false positive and false negative errors given by
\begin{align*}
&A_{+1} = E_+(\psi, \OR_N) = \{x \in \bits^N : \psi(x) > 0,  \OR_N(x) = -1\}, \\
&A_{-1} = E_-(\psi, \OR_N) = \{x \in \bits^N : \psi(x) < 0, \OR_N(x) = +1\}.
\end{align*}

\begin{proof}[Proof of \Cref{lem:ughcor}]
We begin by observing that the quantity of interest can be written as
\begin{align}
&\sum_{x \in \bits^{N \cdot R}}  ( \Phi \ls \psi)(x) \cdot \left(\AND_R \circ \OR_N\right)(x) - \nonumber \\
&\qquad\qquad \left( \sum_{x \in \bits^{N \cdot R} \setminus G} ( \Phi \ls \psi)(x) \cdot \left(\AND_R \circ \OR_N\right)(x) + \sum_{x \in\bits^{N \cdot R} \setminus G} |(\Phi \ls \psi)(x)| \right) \nonumber \\
\ge &\sum_{x \in \bits^{N \cdot R}}  ( \Phi \ls \psi)(x) \cdot \left(\AND_R \circ \OR_N\right)(x) - 2 \sum_{x \in \bits^{N \cdot R} \setminus G} |(\Phi \ls \psi)(x)|. \label{eqn:buttz}
\end{align}
We estimate each term of Expression~\eqref{eqn:buttz} separately, beginning with the first term. Just as in the proofs of \Cref{prop:error-amp} and \Cref{prop:degree-amp}, we have
\begin{align*}
\sum_{x \in \bits^{N \cdot R}} (\Phi \ls \psi)(x) \cdot (\AND_R \circ \OR_N)(x) &= \sum_{z \in \bits^R} \Phi(z) \cdot \E_y[\AND_R(\dots, y_i z_i, \dots)]
\end{align*}
where $y \in \bits^R$ is a random string whose $i$th bit independently takes the value $-1$ with probability $2 \sum_{x \in A_{z_i}} |\psi(x)|$. For $z = -\mathbf{1}$, we have by~\eqref{eqn:mybullshitequation} that $2 \sum_{x \in A_{-1}} |\psi(x)| = 0$, so the contribution of the corresponding term to the sum is $1/2$. For $z = \mathbf{1}$, we use the fact that $2 \sum_{x \in A_{+1}} |\psi(x)| \le \delta$ to compute
\begin{align*}
\frac{1}{2} \cdot \E_{y}[\AND_R(\dots, y_i, \dots)] &= \frac{1}{2} \cdot \left(1 - 2\Pr_y[\AND_R(\dots, y_i, \dots) = -1 ] \right) \\
&\ge \frac{1}{2} \cdot \left(1 - 2\delta^R\right).
\end{align*}
Hence, the first summand of~\eqref{eqn:buttz} is at least $1-\delta^R$.

We now estimate the second summand, $2 \sum_{x \notin G} |(\Phi \ls \psi)(x)|$. As in the proofs of \Cref{prop:error-amp} and \Cref{prop:degree-amp}, we let $\lambda$ denote the distribution with probability mass function $\lambda(x) = |\psi(x)|$. Then
\begin{align*}\hspace{-2.2em}
2\sum_{x \notin G} |(\Phi \ls \psi)(x)| &= 2^{R+1} \E_{\lambda^{\otimes R}} [|\Phi(\dots, \sgn \psi(x_i), \dots)| \cdot \mathbb{I}(x \notin G)] \\
&= 2\sum_{z \in \bits^R} |\Phi(z)| \cdot \Pr_{\lambda^{\otimes R}}[x \notin G | (\dots, \sgn(\psi(x_i)), \dots) = z] \\
&= \Pr_{\lambda^{\otimes R}}[x \notin G | (\dots, \sgn(\psi(x_i)), \dots) = \mathbf{-1}] + \Pr_{\lambda^{\otimes R}}[x \notin G | (\dots, \sgn(\psi(x_i)), \dots) = \mathbf{1}].
\end{align*}
We analyze each term of this sum separately. For the first term, observe that by one-sided error of $\psi$, it follows that if $x = (x_1, \dots, x_R)$ is any input $\sgn(\psi(x_i)) = -1$ for all $i$ then $\OR_N(x_i) = -1$ for all $i$. Thus we are guaranteed that $x \in G$, so the contribution of the first summand is zero. To analyze the second summand, let us denote by $r_i \in \{0, 1\}$ the indicator random variable for the event $\OR_N(x_i) = -1$ when $x_i$ is drawn from the conditional distribution $(\lambda | \sgn(\psi(x_i)) = 1)$. Then
\begin{align*}
\Pr[r_i = 1] &= \Pr_{x_i \sim \lambda}[\OR_N(x_i) = -1 | \sgn(\psi(x_i)) = 1] \\
&= 2\sum_{x \in A_{+1}} |\psi(x_i)| \\
&\le \delta
\end{align*}
by~\eqref{eqn:mybullshitequation}. Hence,
\begin{align*}
 \Pr_{\lambda^{\otimes R}}[x \notin G | (\dots, \sgn(\psi(x_i)), \dots) = \mathbf{1}] &\le \Pr\left[ \sum_{i = 1}^R r_i  > \gamma R\right] \\
 &\le \exp\left(-\frac{(\gamma - \delta)R}{3}\right)
\end{align*}
by the multiplicative Chernoff bound.\footnote{The formulation we use here is as follows. Let $r_1, \dots, r_R$ be independent $\{0, 1\}$-valued random variables, $S = \sum_{i=1}^R r_i$, and $\mu = \E[S]$. Then for $\eta > 1$, we have $\Pr[S > (1+\eta) \mu] \le \exp(-\eta \mu / 3)$. In this application, we are taking $\mu \le \delta R$ and $\eta = \gamma R / \mu - 1 > 1$.} Thus, we have
\[2 \sum_{x \notin G} |(\Phi \ls \psi)(x)| \le \exp\left(-\frac{(\gamma - \delta)R}{3}\right).\]
Putting everything together, we see that Expression~\eqref{eqn:buttz} is at least $1 - \delta^R - \exp\left(-(\gamma - \delta) R / 3\right)$ as we wanted to show.
\end{proof}

\subsection{Lower Bound for Junta Testing}
It follows from a reduction 
in Ambainis et al. \cite[Section 6]{juntatesting}
 that for any $N=O(R)$ and sufficiently small constant $\gamma > 0$, a $\tilde{\Omega}(R^{1/2})$ lower bound for the approximate
 degree or quantum query complexity of $\SE_{N, R}^{\gamma}$ implies 
 an $\tilde{\Omega}(k^{1/2})$ approximate degree or quantum query lower bound for $k$-junta testing for proximity parameter $\eps = 1/3$. Hence,
 \Cref{sethm} has the following corollary.
 
 \begin{corollary}
Any quantum tester that distinguishes $k$-juntas from functions that are $(1/3)$-far
from any $k$-junta with error probability at most $1/3$ makes $\tilde{\Omega}(k^{1/2})$ queries to the function.
 \end{corollary}

\subsection{Lower Bound for \texorpdfstring{$\SDU$}{SDU}}

The goal of this section is to derive a lower bound
for approximating the statistical distance of an input distribution
from uniform up to some additive constant error. We formalize this problem as follows.

Given an input $(s_1, \dots, s_N) \in [R]^N$, and $i \in [R]$, let
$f_i = |\{j \colon s_j=i\}|$, and let $p$ be the probability distribution over $[R]$
such that $p_i=f_i/N$.  For $N \ge R$ and $0 < \gamma_2 < \gamma_1 < 1$, define the partial function $\SDU_{N, R}^{\gamma_1, \gamma_2}$ as follows. 

\begin{definition} Define $$\SDU_{N, R}^{\gamma_1, \gamma_2}(s_1, \dots, s_N) = \begin{cases} 
-1 & \text{ if } \frac{1}{2} \sum_{i=1}^R |p_i - 1/R| \leq \gamma_1\\
1 & \text{ if } \frac{1}{2} \sum_{i=1}^R |p_i - 1/R| \geq \gamma_2\\
\text{undefined} & \text{otherwise}.\end{cases}$$
Above, $ \frac{1}{2} \sum_{i=1}^R |p_i - 1/R| $ is the statistical
distance between $p$ and the uniform distribution.
\end{definition}

The $\SDU$ problem reduces to $\SE$ in the sense that any approximating polynomial for $\SDU$ implies the existence of an approximation to $\SE$ of the same degree. Hence, the approximate degree of $\SDU$ is at least as large as that of $\SE$. For intuition as to why this is true, let us relate $\SE^{1/2}_{N, R}$ to $\SDU^{1/2, 0}_{N, R}$ in the special case where $N = R$ and no dummy (i.e., 0) items appear in the input to $\SE$. If $(s_1, \dots, s_N)$ is a true input to $\SE^{1/2}_{N, R}$, i.e., $\SE^{1/2}_{N, R}(s_1, \dots, s_N) = -1$, then every index $i \in [R]$ must appear in the input list exactly once. Hence, the distribution represented by $(s_1, \dots, s_N)$ is exactly uniform, so $\SDU^{1/2, 0}_{N, R}(s_1, \dots, s_N) = -1$. On the other hand, if $\SE^{1/2}_{N, R}(s_1, \dots, s_N) = 1$, then at most $R/2$ indices $i \in [R]$ appear in the input list, so the list represents a distribution with statistical distance at least $1/2$ from uniform. Thus, an approximating polynomial for $\SDU^{1/2, 0}_{N, R}$ is also an approximating polynomial for $\SE^{1/2}_{N, R}$.

We will actually need a more general relationship between the approximate degrees of $\SDU$ and $\SE$ to handle the fact that we cannot take $N$ to be exactly equal to $R$ in the $\SE$ lower bound, as well as to handle the occurrences of dummy items in the definition of $\SE$.

\begin{theorem} \label{sduthm}
 For some $N = O(R)$, and some constants $0 < \gamma_2 < \gamma_1 < 1$, $\adeg\left(\SDU_{N, R}^{\gamma_1, \gamma_2}\right) \ge \tOmega(R^{1/2})$.
 The same lower bound applies to the quantum query complexity of $\SDU_{N, R}^{\gamma_1, \gamma_2}$.
 \end{theorem}
\label{s:seimplications}
\begin{proof}
Fix $R >0$, let $c$ be the constant from \Cref{sethm},
$\gamma<(2/3c)^2$ be a sufficiently small constant, and $N= c \cdot \gamma^{-1/2} \cdot R$. As inputs in $\SE_{N, R}^{\gamma}$ are interpreted as
elements of $[R]_0^N$, we can equivalently interpret them as elements of $[R+1]^N$,
i.e., as inputs to $\SDU_{N, R+1}^{\gamma_1, \gamma_2}$, for any desired $0 < \gamma_1 < \gamma_2 < 1$.

Set $\gamma_1 = 1-3\gamma/2$ and $\gamma_2=1-\gamma^{1/2}/c$.  
Observe that since $\gamma < (2/3c)^2$, $\gamma_1$ is strictly
greater than $\gamma_2$.
We claim that
\begin{align*} 
&\left(\SE_{N, R}^{\gamma}\right)^{-1}(-1) \subseteq \left(\SDU_{N, R+1}^{\gamma_1, \gamma_2}\right)^{-1}(-1), \text{ and} \\
&\left(\SE_{N, R}^{\gamma}\right)^{-1}(+1) \subseteq \left(\SDU_{N, R+1}^{\gamma_1, \gamma_2}\right)^{-1}(+1).
\end{align*}
Indeed, since inputs in $\left(\SE_{N, R}^{\gamma}\right)^{-1}(-1)$ define a probability distribution over $[R+1]$ with support size at least $R$, with all probabilities being integer multiples of $1/N=\gamma^{1/2}/(cR)$, the statistical distance between any such distribution and
the uniform distribution is at most $1-R/N = 1-\gamma^{1/2}/c$. This follows from the following calculation. 
Amongst probability distributions $(p_1, \dots, p_{R+1})$ over $[R+1]$ with support size at least $R$ and all probabilities $p_i$ being integer multiples of $1/N$,
it is not hard to see that one maximizes the statistical distance from the uniform distribution over $[R+1]$ by setting $p_1=1-\frac{R-1}{N}$, $p_2=p_3 = \dots = p_{R}=1/N$, and $p_{R+1}=0$. 
The statistical distance from uniform is: 
\begin{align*} 
\frac{1}{2} \left( \left(1-\frac{R-1}{N}-\frac{1}{R+1}\right) + \left(R-1\right)\left(\frac{1}{R+1} - \frac1N\right) + \frac{1}{R+1} \right)\\
= 1-\frac{R}{N} + \frac1N - \frac{1}{R+1} \leq 1-\frac{R}{N},
\end{align*}
where we have assumed that $N \geq R+1$, which is true for sufficiently small choice of $\gamma$.

Similarly, since inputs in $\left(\SE_{N, R}^{\gamma}\right)^{-1}(1)$ define a probability distribution over $[R+1]$ with support size at most $\gamma \cdot R +1$, 
the statistical distance between any such distribution and
the uniform distribution is at least 
$1-3\gamma/2$. To see
this, let $p=(p_1, \dots, p_{R+1})$ be any distribution of support size at most $\gamma \cdot R$, and let $S = \{i \colon p_i=0\}$ and $\bar{S}$ be the complement of $S$.
Then the statistical distance of $p$ from uniform is at least 
\begin{align*}
\frac{1}{2} \left(\sum_{i \in S} \left(p_i - \frac{1}{R+1} \right) \right) + 
\frac{1}{2} \left(\sum_{i \in \bar{S}} \left( p_i - \frac{1}{R+1} \right) \right) 
=
\frac{1}{2} \left(\frac{|S|}{R+1} + \left(\sum_{i \in \bar{S}} p_i\right) - \frac{|\bar{S}|}{R+1}\right)\\
= \frac{1}{2} \left(\frac{|S|}{R+1} + 1 - \frac{|\bar{S}|}{R+1}\right)
\geq  \frac{1}{2} \left(\frac{(1-\gamma) R}{R+1} + 1 - \frac{\gamma R + 1}{R+1}\right)
=  \frac{1}{2} \left(1 +  \frac{R-2\gamma R - 1}{R+1}\right)\\
= 1 - \frac{\gamma R + 1}{R+1} = 1 - \gamma - \frac{\gamma + 1}{R+1} \geq 1-3\gamma/2.
\end{align*}

It is an easy consequence of the above that any $\eps$-approximating polynomial of degree $d$
for $\SDU_{N, R+1}^{\gamma_1, \gamma_2}$ implies an approximation to $\SE_{N, R}^{\gamma}$ of the same degree (i.e., that $\adeg_{\eps}(\SE_{N, R}^{\gamma}) \leq d$).
\Cref{sduthm} then follows from \Cref{sethm}.
\end{proof}

\subsection{Lower Bound for Entropy Comparison and Approximation}

Given a distribution $p$ over $[R]$, the Shannon entropy of $p$, denoted $H(p)$,
is defined to be $H(p):=\sum_{i \in [R]} p_i \log_2(1/p_i)$. Following Goldreich and Vadhan \cite{goldreichvadhansurvey}, we define a partial function $\mathsf{GapCmprEnt}_{N, R}^{\alpha, \beta}$ capturing the problem of comparing the entropies of two distributions.

The function $\mathsf{GapCmprEnt}_{N, R}^{\alpha, \beta}$  takes as input two 
vectors in $[R]^N$ and interprets each vector $i \in \{1, 2\}$ as a probability
distribution $p_i$ over $[R]$, with $p_i(j)=f_{i,j}/N$ where $f_{i,j}$ is the frequency of $j$
in the $i$th vector. The function $\mathsf{GapCmprEnt}_{N, R}^{\alpha, \beta}$ evaluates to 
$$\begin{cases}
-1 & \text{ if } H(p_1)-H(p_2) \leq \beta\\
1 & \text{ if } H(p_1)-H(p_2) \geq \alpha \\
\text{undefined} & \text{otherwise}.
\end{cases}
$$

\begin{theorem} \label{entropysucks} There exist constants $0 < \beta < \alpha < 1$ such that $\adeg(\mathsf{GapCmprEnt}^{\alpha, \beta}_{N, R}) = \tilde{\Omega}(R^{1/2}).$ The same lower bound applies to the quantum query complexity of 
$\mathsf{GapCmprEnt}^{\alpha, \beta}_{N, R}$.\end{theorem}
\begin{proof}
Vadhan \cite[Claim 4.4.2 and Remark 4.4.3]{vadhan} showed that as long as $H( (1+\gamma_1)/2) < 1-\gamma_2 - \lambda$, 
then $\SDU_{N, R}^{\gamma_1, \gamma_2}$ is reducible
to $\mathsf{GapCmprEnt}^{\alpha, \beta}_{4N, 2R}$ for some constants $\alpha, \beta$
such that $\alpha-\beta=\lambda$. 
This reduction (described next for completeness) implies that $\adeg_{\eps}(\SDU_{N, R}^{\gamma_1, \gamma_2}) \leq 
\adeg_{\eps}(\mathsf{GapCmprEnt}^{\alpha, \beta}_{4N, 2R})$.

For completeness,
we sketch this transformation, closely following
the presentation of Goldreich and Vadhan \cite{goldreichvadhansurvey}.
At a high level, the reduction transforms
an input in $[R]^N$ to $\SDU_{N, R}^{\gamma_1, \gamma_2}$ (interpreted as a distribution $p$ over $[R]$) into
two distributions $p_1, p_2$ over $[R] \times \{0, 1\}$ as follows.
Both $p_1$ and $p_2$ start by sampling an $s \in \{0, 1\}$ at random. If $s=0$, then a random sample $r$ is chosen from $p$, 
and if $s=1$, then $r$ is set to a uniform random sample from $[R]$. Distribution $p_2$ outputs $(r, s)$, while $p_1$ outputs $(r, b)$ for a random $b \in \{0, 1\}$. 

The entropy of $p_1$
is always $v + 1$, where $v=\frac{1}{2} H(p) + \frac{1}{2} \log_2(R)$. 
As for the entropy of $p_2$, if $p$ is far from the uniform distribution, then the selection bit $s$ will be essentially determined by the sample $r$. Hence,
the entropy of $p_2$ will be approximately $v$, which is noticeably smaller than
the entropy of $p_1$. On the other hand, if the two input distributions are close then
(even conditioned on the sample selected) the selection bit $s$ will be almost
random and so $H(p_2) \approx v+1$, which is approximately the same as $H(p_1)$.
Quantitatively, Vadhan \cite{vadhan} shows that if the statistical distance between $p$ and the uniform distribution is $\delta$,
then $1-\delta \leq H(p_1) - H(p_2) \leq H( (1+\delta)/2)$.

Since we are considering distributions specified as vectors in $[R]^N$, this transformation
can be equivalently described as follows. Assume for simplicity that $R$ divides $N$. 
If $p$ is specified by a vector $u$ in $[R]^N$, then $p_2$ is specified by a vector $w$ in $([R]\times\{0,1\})^{4N}$ defined as follows. For all $i \in [N]$ and $j \in \{0, 1\}$, $w_{i, j}=(u_i, 0)$, and for $j \in \{2, 3\}$, $w_{i, j}=(\lceil Ri/N\rceil, 1)$. Similarly, 
$p_1$ is specified by a vector $v$ in $([R]\times\{0,1\})^{4N}$. For $i \in [R]$ and $j \in \{0, 1\}$, $v_{i, j}=(u_i, j)$, and for $j \in \{2, 3\}$, $v_{i, j} = (\lceil Ri/N\rceil, j-2)$. 
Observe that when representing $u$, $v$, and $w$ as vectors in $\bits^{N \log_2(R)}$
or $\bits^{4N \cdot \log_2(R)}$, each bit of $v$ and $w$ depends on at most one bit of $u$.  

Recall that in the statement of \Cref{sduthm}, $\gamma_1 = 1-3\gamma/2$ and $\gamma_2=1-\gamma^{1/2}/c$,
where $\gamma$ is an arbitrary constant less than $(2/3c)^2$, where $c>1$
is the constant from \Cref{sethm}.
Clearly, 
$H( (1+\gamma_1)/2)) = H(1-3 \gamma/4) = H(3 \gamma/4)$.
Using the fact that for any $p \in [0, 1/2]$,
\begin{align*}
H(p) &= -p \log_2(p) - (1-p) \log_2(1-p) \\
&\le -2p\log_2(p) \\
&\le 6p^{3/4},
\end{align*}
it follows that for some constant $\gamma \le 1/648c^4$, we have $H(3 \gamma/4) < \gamma^{1/2}/c$.
Hence, Vadhan's reduction from $\SDU_{N, R}^{\gamma_1, \gamma_2}$ to $\mathsf{GapCmprEnt}^{\alpha, \beta}_{4N, 2R}$ applies to this setting
of $\gamma_1$ and $\gamma_2$, and this shows
that any degree $d$ $\eps$-approximating polynomial for $\mathsf{GapCmprEnt}^{\alpha, \beta}_{4N, 2R}$ implies a degree $d$ polynomial $\eps$-approximating
polynomial for  $\SDU_{N, R}^{\gamma_1, \gamma_2}$.

Combined with \Cref{sduthm}, this implies that
$\adeg(\mathsf{GapCmprEnt}^{\alpha, \beta}_{4N, 2R}) = \tilde{\Omega}(R^{1/2}).$
\end{proof}

Clearly, a quantum query algorithm
that approximates entropy
up to additive error $(\alpha-\beta)/4$ can be used to solve $\mathsf{GapCmprEnt}^{\alpha, \beta}_{4N, 2R}$, by (approximately) computing the entropies of each of the two input distributions, and determining whether the difference is at most $(\beta + \alpha)/2$. Hence, \Cref{entropysucks} implies the following
lower bound for approximating entropy to additive error $\alpha-\beta$.

\begin{corollary}
Let $N=c \cdot R$ for a sufficiently large constant $c$. 
Interpret an input in $[R]^N$ as a distribution $p$ 
in the natural way (i.e., for each $j \in [R]$, $p_j=f_j/N$, where $f_j$
is the number of times $j$ appears in the input). 
There is a constant $\eps > 0$ such that any quantum algorithm
that approximates the entropy of $p$ up to additive error $\eps$
with probability at least $2/3$
requires $\tilde{\Omega}(R^{1/2})$ queries.
\end{corollary}

 
\section{Conclusion and Open Questions}
\label{s:conclusion}
We conclude by briefly describing some additional consequences of our results,
as well as a number of open questions and directions for future work.

\subsection{Additional Consequences: Approximate Degree Lower Bounds for DNFs and \texorpdfstring{AC$^0$}{AC0}}
For any constant $k > 0$, $k$-distinctness is computed by a DNF of polynomial size.
Our $\tilde{\Omega}\left(n^{3/4-1/(2k)}\right)$ is the best known lower bound on
the approximate degree of polynomial size DNF formulae. The previous best was $\tilde{\Omega}(n^{2/3})$ for Element Distinctness (a.k.a., $2$-Distinctness) \cite{aaronsonshi}, although Bun and Thaler did establish, for any $\delta>0$, an $\Omega(n^{1-\delta})$ lower bound on the approximate degree of \emph{quasi}polynomial size DNFs.

Similarly, for any constant $k\geq 1$, Bun and Thaler exhibited an AC$^0$ circuit of depth $2k-1$
with approximate degree $\tilde{\Omega}\left(n^{1-2^{k-2}/3^{k-1}}\right)$. 
Our techniques can be used to give a polynomial improvement for any fixed $k\geq 2$, to
$\tilde{\Omega}\left(n^{1-2^{-k}}\right)$ (\Cref{thm:introsurj} 
is the special case of $k=2$, as $\SURJ$ is computed by an AC$^0$ 
circuit of depth three). We omit further details of this result for brevity. 

\subsection{Open Problems}
The most obvious direction for future work is to extend our techniques
to resolve the approximate degree and quantum query complexity of additional
problems of interest in the study of quantum algorithms. These include
triangle finding problem~\cite{triangles3,triangles5}, graph collision~\cite{triangles3}, and verifying matrix products~\cite{robin3,robin2}.
It would also be interesting to close the gap between our $\Omega(n^{3/4-1/(2k)})$
lower bound for $k$-distinctness and Belovs' 
$O\left(n^{3/4-1/(2^{k+2} - 4})\right)$ upper bound, especially for small values of $k$ (e.g., $k=3$).

 Although we prove a lower bound of $\tilde{\Omega}(R^{1/2})$
for $\SDU^{\gamma_1, \gamma_2}_{N, R}$ for some constants $0 < \gamma_2 < \gamma_1$, we leave open whether or not $\SDU^{2/3, 1/3}_{N, R} = \tilde{\Omega}(R^{1/2})$.
It may be tempting to suspect that \Cref{sduthm} implies
an $\tilde{\Omega}(R^{1/2})$ lower bound on $\SDU_{N, R}^{2/3, 1/3}$,
by invoking the well-known Polarization Lemma of Sahai and Vadhan \cite{sahaivadhan}. The Polarization Lemma 
reduces $\SDU^{\gamma_1, \gamma_2}_{N, R}$ for any pair of constant $\gamma_1, \gamma_2$
with $\gamma_2 < \gamma_1^2$ to  $\SDU^{2/3, 1/3}_{N', R'}$ for an appropriate choice of $N'$ and $R'$. Unfortunately, $N'$ and $R'$
may be polynomially larger than $N$ and $R$, so this reduction does not give an $\tilde{\Omega}(R^{1/2})$ lower bound for $\SDU^{2/3, 1/3}_{N, R}$ itself.

Another important direction is to resolve
the approximate degree of specific classes of functions, especially polynomial
size DNF formulae, and AC$^0$ circuits. 
As mentioned in the previous subsection, our $k$-distinctness lower bound (\Cref{thm:introkdist}) gives the best known lower bound on polynomial size DNFs. A compelling candidate for improving this lower bound is the $k$-sum function, which may have approximate degree as large as $\Theta(n^{k/(k+1)})$ (it is known that the quantum query complexity of $k$-sum is $\tilde{\Theta}(n^{k/(k+1)})$ \cite{ambainis, belovsspalek}). On the upper bounds side, it may be possible to extend the techniques underlying our $\tilde{O}(n^{3/4})$ upper bound on the approximate degree of $\SURJ$ to yield a sublinear upper bound for \emph{every} DNF formula of polynomial size. 

\begin{openproblem} \label{op4} For every constant $c>0$ and every DNF formula $f\colon \bits^n \to \bits$ of size at most $n^c$,
is there a $\delta > 0$ (depending only on $c$) such that $\adeg(f) = O(n^{1-\delta})$? \end{openproblem}

A positive answer to \Cref{op4} 
would have major algorithmic consequences, including a subexponential time algorithm for agnostically learning
DNF formulae \cite{kkms} (and PAC learning depth three circuits \cite{ksdnf}) of any fixed polynomial size. 

For general AC$^0$ circuits, an $\Omega(n^{1-\delta})$ approximate degree lower bound is already known \cite{adegsurj}. 
It would be very interesting to improve this lower bound to an optimal $\Omega(n)$.
Until recently, $\SURJ$ was a prime candidate for exhibiting such a lower bound. However, owing to Sherstov's upper bound~\cite{sherstovpersonal} and \Cref{thm:introsurj},
$\textsf{SURJ}$ is no longer a candidate function. 
However, we are optimistic about the following closely related candidate. An approximate majority
function is any total Boolean function that evaluates to $-1$ (respectively $+1$) whenever at least $2/3$ of its 
inputs are $-1$ (respectively $+1$). 
It is well-known (via the probabilistic method) that there are approximate majorities computable by depth 3 circuits 
of quadratic size and logarithmic bottom fan-in \cite{ajtai}. 
It is possible that \emph{every} approximate majority has approximate degree $\Omega(n)$;
proving this would resolve a question of Srinivasan \cite{filmus2014real}. 


\section*{Acknowledgements} We are grateful to Sasha Sherstov for an inspiring conversation at the BIRS 2017 workshop on Communication Complexity and Applications, II, which helped to spark this work. We also thank the anonymous STOC and Theory of Computing reviewers for comments improving the presentation of this manuscript.

This work was done while M.~B.~was a postdoctoral researcher at Princeton University. Some of this work was performed when R.~K.~was a postdoctoral associate at MIT and was partly supported by NSF grant CCF-1629809. 

\bibliographystyle{alpha}
\phantomsection\addcontentsline{toc}{section}{References} 
\bibliography{clean}

\end{document}